\def\version{16.2.2018, 8h}\def\WHO{nbd}
\documentclass[final]{birkmult}  
\def\users{}\def\gt{no}
 \def\users{final-layout}  

\def\BLUE{\color{black}}

\usepackage{epsf}
\usepackage{graphicx}
\usepackage{amsmath}
\usepackage{amssymb}
\usepackage{mathrsfs}
\numberwithin{equation}{section}
\usepackage{srcltx}
\usepackage{color}
\usepackage{todonotes}
\usepackage{hyperref}
\usepackage{background}
\usepackage{upgreek}
\usepackage{psfrag}

\definecolor{textcolor}{HTML}{0A75A8}
\definecolor{second}{RGB}{0,0,255}
\definecolor{third}{RGB}{0,0,255}
\usepackage{ifthen}

\ifthenelse{\equal{\users}{final-layout}}{}
{\newcount\hour \newcount\minute
\hour=\time
\divide \hour by 60
\minute=\time
\loop \ifnum \minute > 59 \advance \minute by -60 \repeat

\newcommand\Text{File: \jobname.tex (version \version), compiled:
\number\day.\number\month.\number\year\ at
\the\hour:\ifnum\minute<10 0\fi\the\minute\ h
}
}

\newcounter{includepagecounter}
1
\ifthenelse{\equal{\gt}{yes}}
{

}

\usepackage[normalem]{ulem}
\ifthenelse{\equal{\users}{final-layout}}{

	\newcommand{\COMMENT}[1]{}
        \newcommand{\TODO}[1]{}
        \newcommand{\INTERNAL}[1]{}
        \newcommand{\QUESTION}[1]{}
	\newcommand{\DELETE}[1]{}

        \newcommand{\REM}[1]{\marginpar{\bfseries\tiny{\color{blue}}}}
         }
{\usepackage[notcite,notref]{showkeys}
\definecolor{gt}{RGB}{0,0,80}
\definecolor{appunti}{RGB}{0,255,255}

 \newcommand{\COMMENT}[1]{{\color{red}\uuline{#1}\color{black}}}
 \newcommand{\TODO}[1]{{\color{red}\uuline{#1}\color{black}}}
\newcommand{\INTERNAL}[1]{\footnote{#1}}
 \newcommand{\QUESTION}[1]{{\color{brown}\uuline{#1}\color{black}}}
 \newcommand{\DELETE}[1]{{\color{brown}\sout{#1}\color{black}}}

 \newcommand{\REM}[1]{\marginpar{\bfseries\tiny{\color{blue}#1}}}
}

\SetBgColor{textcolor}
\SetBgOpacity{1}
\SetBgAngle{90}
\SetBgPosition{current page.center}
\SetBgVshift{0.305\textwidth}
\SetBgScale{1.8}

\usepackage{color}
\usepackage{amssymb,amsmath,amssymb,amsthm}
\usepackage{hyperref} 
\usepackage{bm} 
\usepackage{todonotes}
\usepackage{framed}
\usepackage{mathrsfs}  
\usepackage{endnotes} 
\usepackage{pdfpages} 

\newtheorem{lemma}{Lemma}
\newtheorem{proposition}{Proposition}
\newtheorem{theorem}{Theorem}
\newtheorem{definition}{Definition}
\newtheorem{remark}{Remark}
\newcommand\comment[1]{}

\newcommand\dx{\mathop{{\rm d}x}}
\newcommand\dt{\mathop{{\rm d}t}}

\newcommand\mb[1]{\mathbf{#1}}

\newcommand\gradop{\operatorname{grad}}

\newcommand{\R}{\mathbb{R}}
\newcommand{\N}{\mathbb{N}}
\renewcommand{\d}{\mathrm{d}}
\newcommand\bigdot{\text{\LARGE$\cdot$}}
\newcommand\DT[1]{\mathchoice
                 {{\buildrel{\hspace*{.1em}\text{\LARGE.}}\over{#1}}}
                 {{\buildrel{\hspace*{.1em}\text{\Large.}}\over{#1}}}
                 {{\buildrel{\hspace*{.1em}\text{\large.}}\over{#1}}}
                 {{\buildrel{\hspace*{.1em}\text{\large.}}\over{#1}}}}
\newcommand\DDT[1]{\mathchoice
   {{\buildrel{\hspace*{.1em}\text{\Large.\hspace*{-.1em}.}}\over{#1}}}
   {{\buildrel{\hspace*{.1em}\text{\large.\hspace*{-.1em}.}}\over{#1}}}
   {{\buildrel{\hspace*{.1em}\text{\large.\hspace*{-.1em}.}}\over{#1}}}
   {{\buildrel{\hspace*{.1em}\text{\large.\hspace*{-.1em}.}}\over{#1}}}}
\newcommand{\Vdots}{\vdots}

\newcommand{\nablaS}{\nabla_{\scriptscriptstyle\textrm{\hspace*{-.3em}S}}^{}}
\newcommand{\divS}{\mathrm{div}_{\scriptscriptstyle\textrm{\hspace*{-.1em}S}}^{}}
\newcommand{\Cof}{\operatorname{\rm{Cof}}}
\newcommand{\eq}{\eqref}
\newcommand{\eps}{\varepsilon}

\newcommand{\Ceen}{\frac{\Cof(\nabla\bm\chi_{\eps})}{\sqrt{\det(\nabla\bm\chi_{\eps})}}}
\newcommand{\Ceenk}{\frac{\Cof(\nabla\bm\chi_{\sigmanu k})}{\sqrt{\det(\nabla\bm\chi_{\sigmanu k})}}}

\newcommand{\jump}[1]{\ensuremath{[\![#1]\!]} }

\renewcommand{\nu}{\sigma}
\def\sigmanu{\varepsilon}

\begin{document}
\begin{sloppypar}
\allowdisplaybreaks

\title[\BLUE{}A thermodynamical model \color{black} of magneto-elastic 
materials at large strains \COMMENT{\large{ Version: \version, Now: \WHO}}]
{\BLUE{}A thermodynamically consistent model \color{black} of 
magneto-elastic\\materials under diffusion at large strains \BLUE{}and
its analysis\color{black}.}
\author[T. Roub\'\i\v cek]{Tom\'a\v s Roub\'\i\v cek}
\address{
Mathematical Institute, Charles University\\
Sokolovsk\'a 83, CZ-186~75~Praha~8, Czech Republic,
\\
and\\
Institute of Thermomechanics of the Czech Academy of Sciences,\\ 
Dolej\v skova 5, CZ-182 00 Praha 8, Czech Republic}
\email{tomas.roubicek@mff.cuni.cz}

\thanks{This research was partly supported through the grants Czech Science 
Foundation 
16-03823S ``Homogenization and multi-scale computational modelling of flow and 
nonlinear interactions in porous smart structures''
and 16-34894L ``Variational structures in continuum thermomechanics of solids''
and by the Austrian-Czech projects 
FWF/MSMT \v CR no.\ 7AMB16AT015, as well as through the institutional project 
RVO:\,61388998 (\v CR). The second author also 
acknowledges financial support of INdAM-GNFM through grant ``Progetto Giovani''.}

\author[G. Tomassetti]{Giuseppe Tomassetti}
\address{Universit\`a degli Studi Roma Tre
- Dipartimento di Ingegneria - Sezione di Ingegneria Civile,\\
Via Volterra 62 - 00146 Rome, Italy}
\email{giuseppe.tomassetti@uniroma3.it}



\date{}

\maketitle


\COMMENT{\LARGE \WHO\LARGE{'s currently working on the file}} 

\begin{abstract}
A theory of elastic magnets is formulated under possible diffusion and heat 
flow governed by Fick's and Fourier's laws in the deformed (Eulerian) 
configuration, respectively. The concepts of nonlocal nonsimple materials 
and viscous Cahn-Hilliard equations are used. The formulation of the problem 
uses Lagrangian (reference) configuration while the transport processes are 
\BLUE pulled \color{black}
back. Except the static problem,
the demagnetizing 
\BLUE{}energy
\color{black} is ignored and only local non-selfpenetration
is considered. The analysis as far as existence of weak solutions of the 
(thermo)dynamical problem is performed by a careful regularization 
and approximation by a Galerkin method, suggesting also a numerical strategy.
Either ignoring or combining particular aspects, the model
has numerous applications as ferro-to-paramagnetic transformation 
in elastic ferromagnets, diffusion of solvents in polymers possibly accompanied
by magnetic effects (magnetic gels), or metal-hydride phase transformation
in some intermetalics under diffusion of hydrogen accompanied possibly
by magnetic effects (and in particular ferro-to-antiferromagnetic phase 
transformation), all in the full thermodynamical context under large strains. 

\medskip


\medskip


\end{abstract}

\maketitle


\section{Introduction}
Recent experiments have been putting into perspective interesting interplay between hydrogenation and magnetic 
properties of ferromagnetic specimens. At the microscopic level, the atomic lattice parameters are substantially 
enlarged when diffusing hydrogen atoms occupy 
the interstitial positions  without changing  
the cubic (or sometimes hexagonal) structure, \color{black} leading to large strains. We the speak 
about metal-to-hydride phase transformation. If the metal is
ferromagnetic, the resulted hydride may be antiferromagnetic, as documented 
experimentally e.g.\ in \cite{RSSH01HPSM,TKFK07SHMR}. 
Of course, this phase transformation is in addition to the ferro-to-paramagnetic
phase transformation in the parent metal phase when temperature
rises the Curie point.
Also specific heat can be 
\BLUE substantially \color{black} influenced by the 
concentration of hydrogen, as experimentally documented in the case of the hydrogenization/deuterization of some intermetalics e.g.\ in \cite{TKFK07SHMR}. Although mathematical models of hydrogenation have been formulated and studied specifically for hydrogen storage with \cite{duda2015stress,RouTom14THSM} or without \cite{BonetCL2012NATMA,BonetCT2015non} strain effects these models do not take into account large strains and the interaction between diffusant and 
magnetization, see also \cite{kolomiets1997rnial,KolomHRBNYHDID2000Magnetic,KolomHSYHo1999Structural,KolomHYA1997Hydrogen,KolomHYA1999Hydrogenation}.

Besides specific applications to hydrogenation, our model can cover situations when magnetic effects are coupled with substantial large strains, such as elastomers \cite{BorceB2001magneto} and polymer gels \cite{CollinAuernhammerGavatEtAl2003,BohliusBrandPleiner2004} with magnetic inclusions. These materials are obtained by embedding metallic or ferromagnetic particles into solid matrix and can undergo controlled large strains, with potential applications for biomechanics and biomimetics \cite{ZrS2001Muscular}.
\color{black}

We also point out that instead of magnetization, one can thus equally think about polarization and ferroelectric materials instead
of ferromagnetic. 

The range of applicability of the model is the ability of magnetic-field-sensitive gels to undergo a quick controllable change of shape can be used to create an artificially designed system possessing sensor- and actuator functions internally in the gel itself. The peculiar magneto-elastic properties may be used to create a wide range of motion and to control the shape change and movement, that are smooth and gentle similar to that observed in muscle. Magnetic field sensitive gels provide attractive means of actuation as artificial muscle for biomechanics and biomimetic applications.

To partially fill this gap, in this paper we propose and study a large-strain thermomechanical 
model describing a magnetized solid 
permeated by a diffusant and undergoing non-isothermal processes. As usual for solid mechanics, 
we formulate the problem in the referential setting, choosing as reference configuration an 
unbounded domain ${\varOmega}\subset\R^d$ having a Lipschitz boundary ${\varGamma}=\partial{\varOmega}$. 
At variance with most treatments in ferromagnetism \cite{Brown1963} we do not impose the saturation 
(so called Heisenberg) constraint 
on the magnitude of the magnetization which is, in fact, relevant
rather only below the Curie point, cf.\ also the discussion in 
 \cite{PoRoTo10TCTF} and references cited therein. 

\BLUE{}Rather as a side effect, forgetting magnetic long-range 
interaction, the model covers also the flow in poro-elastic materials
at large strains, as e.g.\ rocks or soils.

Our ultimate goal is to develop a model which is both thermodynamically 
consistent and also amenable to mathematical analysis. \color{black} 
Let us briefly summarize the mathematical challenges set forth by the problem 
we examine:

While existence of minimizers in nonlinear elastostatics is well understood, 
the question whether these minimizers correspond to actual weak solutions of 
the Euler-Lagrange equations is an open issue 
\cite{Ball02SOPE}, essentially because the strain energy blows up as the 
determinant of the deformation gradient tends to null. Such technical 
hindrance is exacerbated when, instead of looking for weak solutions in 
nonlinear elastostatics, one considers the evolution problem of nonlinear 
elastodynamics. 

In addition, handling the effects of magnetization and diffusion at large strains requires some care. In fact, the equations of magnetostatics are naturally formulated in the actual configuration and, in order to pull the relevant fields back to the reference configuration the injectivity of the deformation map is mandatory. Now, while the injectivity of the deformation map can be guaranteed in the variational setting by enforcing the Ciarlet-Ne\v cas condition \cite{CiaNec87ISCN}, this is not mathematically feasible when seeking weak solutions of the evolution equations of nonlinear elastodynamics.\color{black}

Second, the Fick-type relations between the flux of the diffusant and its driving force, namely, the gradient of chemical potential, as well as the Fourier law relating heat flux and temperature gradient, are often formulated in the deformed configuration (see for example \cite{DuSoFi10TSMF}). When these constitutive laws are pulled back to the reference configuration (see \eqref{M-K-push-back} below), the reciprocal of the determinant of the deformation gradient enters into play, it would be desirable to have the determinant is away from zero. 
\color{black}

\BLUE{}Following an idea from \cite{HeaKro09IWSS}, we will include in the 
free energy a regularizing term to ensure that the determinant of the 
deformation gradient stays away from zero, cf.\ also 
Lemma~\ref{lem:1} below. However, since we are interested 
in performing mathematical analysis also in the dynamical setting, we opt 
for a quadratic regularizing term, whose contribution to the Fr\'echet 
derivative of the free energy has the nice property of being linear, which 
makes it possible to pass to the limit in the nonlinear hyperbolic problem in 
our proof of existence of weak solutions, see Sec.~4 below. 
Now, it turns out that in order for this approach to work, the deformation 
gradient should be at least in the Sobolev-Slobodetski\u\i\ space $H^{2+\gamma}$, 
with $\gamma$ a positive exponent smaller than 1. To meet this requirement, 
one option would be to consider a local non-simple material of grade 3, whose 
free energy depends on the third gradient of the deformation.  Instead, we 
have chosen a model consistent with the definition of the space $H^{2+\gamma}$, 
namely, the space of functions in $H^2$ whose  second gradient is in the 
fractional space $H^\gamma$. 

\color{black}
\BLUE{}The model which we adopt thus involves a nonlocal dependence on 
the second gradient, combining therefore the concept of gradient-theories for 
strains, which is usually referred as nonsimple, cf.\ \cite{Toup62EMCS} or 
also e.g.\
\cite{FriGur06TBBC,MinEsh68FSTL,Podi02CISM,Silh88PTNB,TriAif86GALD}, with the 
nonlocal concept like that introduced, for instance, in \cite{Erin02NCFT} with 
the nonlocality in the first gradients, however. Thus, this model may be 
referred to as {\it material of grade $2+\gamma$} or also be regarded as a 
{\it 2nd-grade non-simple material} cf.\ 
Remark~\ref{rem-nonlocal}.\color{black} 

The other higher-order contribution to the free energy  {\BLUE} involving  
$\nabla \bm m$ and $\nabla\zeta$ \color{black} are quite standard and 
widely accepted. In particular, as in standard micromagnetics, 
\BLUE $\nabla \bm m$ models \color{black}
exchange interactions by including a gradient term on the magnetization \BLUE{}\cite{Brown1963}.\color{black}  
As far as concentration of diffusant is concerned, 
\BLUE $\nabla\zeta$ \color{black} includes an interfacial 
energy of Cahn-Hilliard type \cite{CahnH1958JCP}. This leads us to
the overall thermomechanical Helmholtz free energy formulated in the reference 
configuration:
\begin{align}\label{Helmholtz}
\mathcal H_{\rm th}(\bm\chi,\bm m,\zeta,\theta)=\int_{\varOmega}
\psi(\nabla\bm\chi,\bm m,\zeta,\theta)
+\frac{\kappa_1}2|\nabla\bm m|^2+\frac{\kappa_2}2|\nabla\zeta|^2\,\d x
+\mathscr{H}(\nabla^2\bm\chi),
\end{align}
where, using the placeholder ${\mathbb G}$ for $\nabla^2\bm\chi$, 
\color{black} we define the quadratic form
\begin{align}\label{def-of-H}
\mathscr{H}({\mathbb G}):=\frac14\int_{\varOmega}\int_{\varOmega}({\mathbb G}(x){-}{\mathbb G}(\tilde x))^\top
\Vdots\mathscr{K}(x{-}\tilde x)\Vdots({\mathbb G}(x){-}{\mathbb G}(\tilde x))\,\d x\d\tilde x\BLUE,\color{black}
\end{align}
\BLUE{}where \color{black}the kernel $\mathscr{K}:\R^d\to{\R^d}^6$ takes values in the space of sixth-order tensors. Here we adopt the convention triple dots joining a sixth-order tensor (here $\mathscr K$) and a third-order tensor (here ${\mathbb G}(x){-}{\mathbb G}(\tilde x)$) denote contraction of the last three indices of the former with the indices of the latter. More generally, we denote by ``$\,\cdot\,$'' and ``$\,:\,$'' and ``$\,\Vdots\,$'', the scalar product between vectors, tensors, and 3rd-order tensors, respectively. For our mathematical purpose, it will be 
desired if the kernel is singular around the origin, having the asymptotic 
character 
\begin{align}\label{kernel}
\exists\,\epsilon>0\ \forall\,x\!\in\!\R^d,\ |x|\le\epsilon:\ \ \ 
\BLUE \inf_{|{\mathbb G}|=1}{\mathbb G}^\top\Vdots\mathscr{K}(x)\Vdots{\mathbb G}\ge\epsilon
\color{black}
/|x|^{d+2\gamma}.
\end{align}
\BLUE{}As pointed out above,  \color{black}
the growth condition \eq{kernel} is \BLUE inspired by the usual 
Gagliardo seminorm \color{black} on Sobolev-Slobodetski\u\i\ space $H^{2+\gamma}({\varOmega};\R^d)$. \BLUE Cf.\ also Remark~\ref{rem-nonlocal} below
for a discussion of the mechanical concept.
\color{black}

The state variables are the displacement $\bm\chi$, the magnetization $\bm m$,
concentration $\zeta$ of a diffusant (typically some liquid, gas, or some 
solvent and, depending on specific applications, it may be hydrogen, deuterium, 
water, etc.), and (absolute) temperature $\theta$. The particular 
equations of the system considered in this paper are the momentum equilibrium, 
the flow rule for magnetization, the 
balance of mass of the diffusant, and the heat-transfer equation. The basic 
notation is  summarized in Table\,1.

\vspace{.5em}

\begin{center}
\fbox{
\begin{minipage}[t]{0.37\linewidth}\small\medskip
$\varOmega$ the reference domain\\ 
$\bm\chi$ deformation  \\
$\Omega=\bm\chi(\varOmega)$ the actual deformed domain\\ 
$\varGamma$ the boundary of $\varOmega$\\ 
$\bm v$ velocity \\
$\bm m,\bm{\mathsf m}$ magnetization vectors \\
$\zeta$ concentration\\
$\theta$ temperature\\
$\bm F$ deformation gradient\\
$\bm C$ right Cauchy-Green tensor\\
${\mathbb G}$ gradient of $\bm F$\\
$\bm S$ stress tensor\\
$\mathfrak{H}$ hyperstress 3rd-order tensor\\
$\mathscr{H}$ potential of the hyperstress\\
$\mu$ chemical potential\\
$\phi,\upphi$ magnetic potentials\\
$\bm M,\bm{\mathsf M}$ mobility tensors\\
$\bm K,\bm{\mathsf K}$ heat-conductivity tensors
\end{minipage}
\begin{minipage}[t]{0.45\linewidth}\small\medskip
$\varrho$ mass density\\
$\psi$ free energy\\ 
$s$  entropy\color{black}\\
$e_{\textsc{th}}=w$ thermal part of internal energy\\
$c_{\rm v}=c_{\rm v}(\bm m,\zeta,\theta)$ heat capacity\\
$\upmu_0$ vacuum permeability\\
$\psi_{\textsc{me}},\psi_{\textsc{th}}$ chemo-mechanical and thermal free energies\\
$\bm f,\bm{\mathsf f}$ bulk forces (e.g.\ gravity)\\
\BLUE$\bm g$ traction force\color{black}\\
$\bm h_{\rm e},\bm{\mathsf h}_{\rm e}$ external magnetic field\\
$\mu_{\rm e}$ boundary chemical potential\\
$K>0$ transmission coefficient for heat supply on $\varGamma$
\\
$M>0$ transmission coefficient for diffusant 
 on $\varGamma$\\
\BLUE$\eps>0$ a regularization parameter\color{black}\\
$k\in\N$ a numbering of the Galerkin space discretisation\\
$r$ heat-production rate\\
$\tau_1,\tau_2>0$ relaxation times
\\
 $\kappa_1,\kappa_2>0$ length-scale parameters
\end{minipage}
}\end{center}

\vspace{-.4em}

\begin{center}
{\small\sl Table\,1.\ }
\begin{minipage}[t]{.9\textwidth}\baselineskip=8pt
{\small
Summary of the basic notation used through this paper. 
The Italics font indicates the reference material (Lagrangian) configuration 
(as in Fig.\,\ref{fig-elastic-magnet}) while Roman indicated the actual deformed (Eulerian) 
configuration. 
}
\end{minipage}
\end{center}

\medskip

The plan of the paper is the following. In Section~\ref{sec-static} we explore equilibrium states, that is rest states characterized by uniform temperature and chemical potential. 
This will allow us to  carry out the rigorous mathematical treatment of a model that takes into account the complete energetic of the system. In particular, we shall be able to handle the demagnetizing energy by guaranteeing the invertibility of the deformation through the Ciarlet-Ne\v cas \cite{CiaNec87ISCN}  condition $$\int_{\varOmega}\det(\nabla\bm\chi(x))\,\d x\le {\rm meas}_d(\bm\chi({\varOmega}))$$ which, together with $\det(\nabla\bm\chi)>0$ a.e.\ on ${\varOmega}$, ensures
existence of $\bm\chi^{-1}$ a.e.\ on $\bm\chi({\varOmega})$. In Section~\ref{sec-dynamic} we lay down the evolution system, including all relevant thermodynamical couplings. We show existence of weak solutions in Section~\ref{sec-Galerkin}. This is, to some
extent, a constructive method which suggest (when using e.g.\ finite-element 
method for the Galerkin approximation) a numerically stable and 
convergent computationally implementable strategy for solution of
the dynamical problem.

Thorough the whole paper, 
we will use the standard notation for the Lebesgue $L^p$-spaces and
$W^{k,p}$ for Sobolev spaces whose $k$-th distributional derivatives 
are in $L^p$-spaces. We will also use the abbreviation $H^k=W^{k,2}$. 
Moreover, we use the standard notation  $p'=p/(p{-}1)$, and 
$p^*$ for the Sobolev exponent $p^*=pd/(d{-}p)$ for $p<d$ while
$p^*<\infty$ for $p=d$ and $p^*=\infty$ for $p>d$,
and the ``trace exponent'' $p^\sharp$ defined 
as $p^\sharp=(pd{-}p)/(d{-}p)$ for $p<d$ while
$p^\sharp<\infty$ for $p=d$ and $p^\sharp=\infty$ for $p>d$.
Thus, e.g., $W^{1,p}({\varOmega})\subset L^{p^*}\!({\varOmega})$ or 
$L^{{p^*}'}\!({\varOmega})\subset W^{1,p}({\varOmega})^*$=\,the dual to $W^{1,p}({\varOmega})$. 
In the vectorial case, we will write $L^p({\varOmega};\R^n)\cong L^p({\varOmega})^n$ 
and $W^{1,p}({\varOmega};\R^n)\cong W^{1,p}({\varOmega})^n$. Also,
we admit $k$ noninteger with the reference to the Sobolev-Slobodetski\u\i\ 
spaces. Note that, in this notation, we have the compact embedding $H^{2+\gamma}({\varOmega})\subset W^{2,p}({\varOmega})$
if $p>2d/(d-2\gamma)$ and $W^{2,p}({\varOmega})\subset W^{1,p^*}({\varOmega})$.
In particular $H^{2+\gamma}({\varOmega})\subset C^1(\bar{\varOmega})$ if 
$d<p<2d/(d-2\gamma)$, which can be satisfied if $\gamma>d/2-1$ as
employed in \eqref{ass-HK} to facilitate usage of the results from 
\cite{HeaKro09IWSS}. We also denote by $\operatorname{meas}_d$ the $d$-dimensional Hausdorff measure.

\section{Static model in the Lagrangian formulation}\label{sec-static}

Null entropy production is an essential character of equilibrium states, a character that distinguishes them from the more encompassing class of rest (\emph{i.e.} steady) states. In the presence of thermal conduction and chemical diffusion, null entropy production demands that the heat flux and the flux of diffusant vanish (for a general discussion in the context of classical continuum thermodynamics we refer to \cite[Chap. 13]{Silha1997Mechanics}). 

In this section we investigate the existence of equilibrium states for a body in a conservative mechanical environment and thermal and chemical environment of isolation type. The first part of this section is dedicated to the construction of the internal energy and the potential energy. As usual, for the purpose of the derivation of the model, we shall assume that all fields of interest are as smooth as needed for our manipulations to make sense. 

\medskip

\paragraph{State variables.} 
We identify a rest state with the following quadruplet of \emph{state fields}:
\begin{itemize}
\item the \emph{deformation} $\bm\chi:\varOmega\to\mathbb R^d$;
\item the \emph{Lagrangian magnetization} $\bm m:\varOmega\to\mathbb R^d$;
\item the \emph{concentration} $\zeta:\varOmega\to\mathbb R^+$;
\item the \emph{temperature} $\theta:\varOmega\to\mathbb R^+$. 
\end{itemize}
These fields are defined on the reference configuration $\varOmega$, which 
we assume to be an open, bounded smooth domain in $\mathbb R^d$, with $d$ 
the space dimension. 

\medskip

\BLUE{}In this section, \color{black}we shall work with admissible deformations that are 
one-to-one, locally orientation-preserving mappings. Local orientation preservation is the 
requirement that the determinant of the deformation gradient be positive:
\begin{equation}
  J=\operatorname{det}\bm F> 0,\qquad \bm F=\nabla\bm\chi.
\end{equation}
If this requirement is met, global invertibility can be guaranteed through the 
condition
\begin{equation}
  \label{eq:25}  
\displaystyle\int_{\varOmega}J\,\d x\le{\rm meas}_d(\bm\chi({\varOmega})),
\end{equation}
which was introduced by Ciarlet \& Ne\v{c}as \cite{CiaNec87ISCN} as a device 
to preclude minimizers of the variational problem of nonlinear elastostatic 
from self-penetration, \BLUE and later used in the context of 
static magneto-elasticity in \cite{Roge88NVPE,Roge93ERLD}.\color{black}

These requirements on the deformation are essential to attribute physical meaning to the referential fields $\bm m$, $\zeta$, and $\theta$. Indeed, from the knowledge of these fields one can reconstruct the \emph{spatial magnetization density in the body}, the \emph{spatial concentration in the body}, and the \emph{spatial temperature in the body}, respectively,
\begin{equation}
  \bm{\mathsf m}:\Omega\to\mathbb R^d,\qquad \upzeta:\Omega\to\mathbb R^+,\qquad \text{and}\qquad \uptheta:\Omega\to\mathbb R^+,
\end{equation}
where 
\begin{equation}
  \Omega=\bm\chi(\varOmega)
\end{equation}
is the region occupied by the body in its actual configuration. These fields are defined by
\begin{align}\label{m-pull-back}
\bm{\mathsf m}=(J^{-1}\bm F\bm{m}){\circ}\bm{\chi^{-1}},\qquad \upzeta=(J^{-1}\zeta){\circ}\bm\chi^{-1},\qquad\text{and}\qquad \uptheta=\theta{\circ}\bm\chi^{-1}.
\end{align}
The spatial fields are more amenable to physical interpretation than
the corresponding referential quantities, since their value at a point
$z\in\Omega$ represents quantities that in principle are
accessible to physical measurement. Precisely: $\bm{\mathsf m}(z)$
is the density of magnetic moments per unit spatial volume; $\upzeta(z)$ is the density of diffusant per
unit spatial volume; $\uptheta(z)$ is the temperature in the
position $z$. 

\medskip

\paragraph{The environment.}  As we have anticipated, we assume that the chemical environment is of isolation type. This means that the flux of chemical species vanish at the boundary. On account of this we impose, as a constraint, the total amount of chemical species be  equal to a given constant:
\begin{equation}\nonumber
\int_\varOmega\zeta\,\dx=Z_{\rm tot}.
\end{equation}
We next specify the standard mechanical interaction of the body with its environment. To this effect, we assume that the body is clamped on a part $\varGamma_{\rm D}$ of $\partial\Omega$ having positive Hausdorff two-dimensional measure. Accordingly, we require that admissible deformations satisfy 
\begin{equation}\label{eq:fixation}
\bm\chi=\bm\chi_{\rm D}\ \text{ on }\ \varGamma_{\rm D},
\end{equation}
with $\bm\chi_{\rm D}$ given. We model interaction with the environment of purely mechanical origin through the \emph{mechanical potential energy} \cite{Ciarl1988}:
\begin{equation}\label{eq:17}
  \mathcal L(\bm\chi):=\int_{\varOmega} \bm f\cdot\bm\chi\dx+\int_{{\varGamma}_{\rm N}}\bm g\cdot\bm\chi\operatorname{d}\!S,
\end{equation}
where $\bm g:\varGamma_{\rm N}\to\mathbb R^d$ and 
$\bm f:\Omega\to\mathbb R^3$ is a system of dead loads, 
see Fig.~\ref{fig-elastic-magnet}. \BLUE Let us remark
that, in fact, more general load is possible in particular also 
on the boundary due to the used concept of nonsimple materials,
but we do not want to bring additional technicalities into the model. 
\color{black}

We next turn to the description of magnetic interactions with the environment. Such interactions depend on magnetic moments associated to possibly electric currents or other magnetized bodies outside the region $\Omega$. We take these interaction collectively into account through the following \emph{magnetic potential energy}:
\begin{equation}\label{eq:1}
\mathcal Z(\bm\chi,\bm m)\color{black}=\int_{{\varOmega}}\!\bm h_{\rm e}\cdot {\bm m}\,\d x
\end{equation}
where $\bm h_{\rm e}:\varOmega\to\mathbb R^d$, the Lagrangian external magnetic field, depends on $\bm\chi$ through the relation
\begin{equation}
  \label{eq:19}
\bm h_{\rm e}=\bm F^\top [\bm{\mathsf h}_{\rm e}{\circ}\bm{\chi}],\qquad\bm F=\nabla\bm\chi,
\end{equation}
with $\bm{\mathsf h}_{\rm e}:\mathbb R^d\to\mathbb R^d$ an externally imposed magnetic field. From the first of \eqref{m-pull-back}, we have
\begin{equation}\label{eq:14}
\bm m=J\,\bm F^{-1}\,[\bm{\mathsf m}{\circ}\bm{\chi}],
\end{equation}
and hence
\begin{align}
\mathcal Z(\bm\chi,\bm m)
&=\int_{\varOmega}\!\Big(\bm F^\top[\bm{\mathsf h}_{\rm e}{\circ}\bm{\chi}]\Big)\cdot
\Big(J\color{black}\bm F^{-1}[\bm{\mathsf m}{\circ}\bm{\chi}]\Big)\,\d x
\nonumber\\&
=\int_{\varOmega}\![\bm{\mathsf h}_{\rm e}{\circ}\bm{\chi}]\cdot[\bm{\mathsf m}{\circ}\bm{\chi}]J\,\d x
=\int_{\bm\chi({\varOmega})}\bm{\mathsf h}_{\rm e}\cdot\bm{\mathsf m}\,\d z
=\int_{\mathbb R^d}\!\bm{\mathsf h}_{\rm e}\cdot\overline{\bm{\mathsf m}}\,\d z.
\label{Zeeman-energy}
\end{align}
where, for $\bm{\mathsf m}$ defined in \eqref{m-pull-back},
\begin{equation}\label{eq:12}
  \overline{\bm{\mathsf m}}:\mathbb R^d\to\mathbb R^d,\qquad\text{defined by }\overline{\bm{\mathsf m}}(z)=
\left\{
\begin{array}{ll}
\bm{\mathsf m}(z)&\text{if }z\in\Omega,\\[0.4em]
\bm 0&\text{otherwise},
\end{array}
\right.
\end{equation}
is the trivial extension of $\bm{\mathsf m}$ to $\mathbb R^d$. It is easy to check that the last term in the chain of equalities \eqref{Zeeman-energy} coincides with the standard Zeeman energy \cite{HuberS1998}.

An illustration of the relation between the spatial and Lagrangian fields of 
interest is in 
Fig.~\ref{fig-elastic-magnet}.

\begin{figure}[ht]
\begin{center}
\psfrag{x}{\small $x$}
\psfrag{y(x)}{\small $\chi(x)$}
\psfrag{m}{\small $\bm m$}
\psfrag{h_r}{\small $\bm h_{\rm e}$}
\psfrag{m_s}{\small $\bm{\mathsf m}$}
\psfrag{reference}{\small\bf reference}
\psfrag{(material)}{\small\bf (material)}
\psfrag{deformed}{\small\bf deformed}
\psfrag{(spatial)}{\small\bf (spatial)}
\psfrag{configuration}{\small\bf configuration}
\psfrag{homogeneous}{\small\bf homogeneous}
\psfrag{magnetic}{\small\bf magnetic}
\psfrag{field}{\small\bf field $\bm{\mathsf h}_{\rm e}$}
\psfrag{GDir}{\large ${\varGamma}_\text{\sc D}$}
\psfrag{W}{\large ${\varOmega}$}
\psfrag{y(O)}{\large $\bm\chi({\varOmega})$}
\includegraphics[width=.85\textwidth]{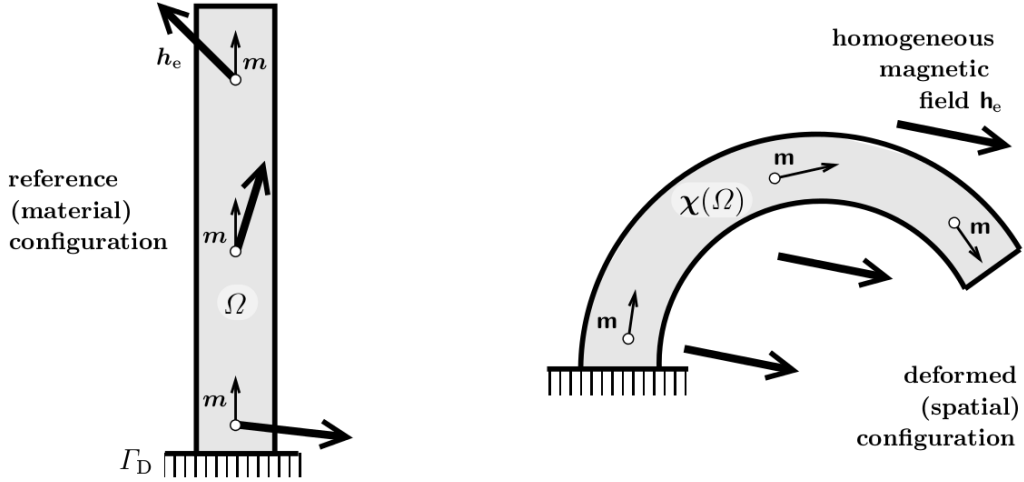}
 \end{center}
 \vspace*{-.9em}
 \caption{Sketch of the relation between the Lagrangian magnetization $\bm m=J\bm F^{-1}[\bm m{\circ}\bm\chi^{-1}]$ and external magnetic field $\bm h_{\rm e}=\bm F^\top[\bm{\mathsf h}_{\rm e}{\circ}\bm\chi^{-1}]$ and their spatial counterparts $\bm{\mathsf m}$ and $\bm{\mathsf h}_{\rm e}$. In the present case we consider a body that is elastically soft and magnetically hard with homogeneous Lagrangian magnetization $\bm m$. The body is fixed at the bottom by Dirichlet condition is deformed 
 in a spatially-constant external magnetic field $\bm{\mathsf h}_{\rm e}$. 
}
 \label{fig-elastic-magnet}
 \end{figure}

\medskip

\paragraph{The free energy of the body.}
The total Helmholtz free energy is the sum
\begin{equation}\label{eq:15}
  \mathcal H(\bm\chi,\bm m,\zeta,\theta)=\mathcal H_{\rm th}(\bm\chi,\bm m,\zeta,\theta)+\mathcal H_{\rm mag}(\bm\chi,\bm m)
\end{equation} 
of the \emph{thermomechanical free energy} $\mathcal H_{\rm th}(\bm\chi,\bm m,\zeta,\theta)$ defined in \eqref{Helmholtz} and the \emph{magnetostatic free energy} \cite{HuberS1998}:
\begin{equation}\label{eq:18}
  \mathcal H_{\rm mag}(\bm\chi,\bm m)=\frac {\upmu_0}2\int_{\mathbb R^d}|\operatorname{\mathsf{grad}}\upphi|^2\,\d z,
\end{equation}
where  $\upmu_0$ is the permeability of vacuum and where $\upphi$, the spatial scalar magnetic potential, is a solution of 
\begin{align}\label{maxwell+}
-\operatorname{\sf{div}}\big(\mu_0\operatorname{\mathsf{grad}}\upphi\big)=\operatorname{\sf{div}}\big(\overline{\bm{\mathsf m}}\big)\quad\text{on }\ \R^d\ \text{ with }\ \overline{\bm{\mathsf m}}\text{ from }\eqref{m-pull-back}\text{ and }\eqref{eq:12}
\end{align}
in the sense of distributions\BLUE; to be more specific,
we also consider a ``boundary condition at $\infty$'' 
by seeking a solution to \eq{eq:12} in $H^1(\R^d)$, cf.\ Theorem~\ref{thm1} below. \color{black}

On using $\upphi_{\bm \chi,\bm m}$ as a test function in \eqref{maxwell+} we find
\begin{align}\nonumber
\mathcal H_{\rm mag}(\bm\chi,\bm m) & {} = {} \frac12\int_{\mathbb R^d}\!\!{\mathsf{grad}}\upphi\cdot\overline{\bm{\mathsf m}}\,\d z {} = {} \frac12\int_{\bm{\chi}({\varOmega})}\!\!{\mathsf{grad}}\upphi\cdot\bm{\mathsf m}\,\d z.
\end{align}
Moreover, by introducing the referential 
\BLUE{}scalar magnetic potential\color{black}
\begin{align}\label{phi-h-pull-back}
\phi=\upphi{\circ} \bm\chi,
\end{align}
we have
\begin{equation}
  \operatorname{\mathsf{grad}}\upphi=[\bm F^{-\!\top}\nabla\phi]{\circ}\bm\chi^{-1}\qquad\text{in }\Omega,
\end{equation}
hence, on recalling \eqref{eq:14}, we can write, by changing from $\bm\chi(\varOmega)$ to $\varOmega$ the domain of integration,
\begin{align}
 \mathcal H_{\rm mag}(\bm\chi,\bm m) 
&
= {} \frac12\int_{\varOmega}\!\big(\bm F^{-\!\top}\nabla\phi\big)\cdot
[\bm{\mathsf m}{\circ}\bm{\chi}]J\,\d x= \frac12\int_{\varOmega}\!\nabla\phi\cdot
J\bm F^{-1}[\bm{\mathsf m}{\circ}\bm{\chi}]\,\d x
=\frac 12\int_{{\varOmega}}\!\nabla\phi{\cdot} {\bm m}\,\d x.
\label{demagnetising-energy}
\end{align}
By combining \eqref{Helmholtz}, \eqref{eq:15}, and \eqref{demagnetising-energy}, we obtain the total Helmholtz free energy of the body:
\begin{equation}\label{eq:16}
  \mathcal H(\bm\chi,\bm m,\zeta,\theta)=\int_\Omega \psi(\nabla\bm\chi,\bm m,\zeta,\theta)
+\frac{\kappa_1}2|\nabla\bm m|^2+\frac{\kappa_2}2|\nabla\zeta|^2+\frac 12 \mathfrak H(\nabla^2\bm\chi)\Vdots\nabla^2\bm\chi+\frac 12\nabla\phi\cdot\bm m\,\d x,
\end{equation}
Here $\mathfrak{H}:=\mathscr{H}'$ is the (G\^ateaux) derivative of the 
nonlocal energy $\mathscr H$ defined in \eqref{def-of-H}, namely 
\begin{align}\label{def-of-H++}
[\mathfrak H({\mathbb G})](x):=\big[\mathscr{H}'({\mathbb G})\big](x)=\int_{\varOmega}
\mathscr{K}(x{-}\tilde x)\Vdots({\mathbb G}(x){-}{\mathbb G}(\tilde x))\,\d\tilde x,
\end{align}
as can be seen by the Fubini theorem to evaluate the nonlocal hyperstress 
explicitly 
\BLUE provided \color{black}
the symmetry of $\mathscr K$ \BLUE in the sense 
$\mathscr K(x)=\mathscr K(-x)$ is additionally assumed\color{black}.

\medskip

\paragraph{Assumptions on the bulk free energy.}
We assume that the bulk free-energy mapping $\psi$ is continuous and 
twice-continuously differential with respect to $\theta$ on $[0,+\infty)$. 
We also assume that it satisfies \BLUE the \color{black} frame 
indifference and  \BLUE exhibits a \color{black} (sufficiently \BLUE fast) blow-up when ${\rm det}\bm F$ approaches zero 
in order to prevent local self-penetration,  \color{black} 
and that it is a 
strictly concave function \BLUE and \color{black} 
coercive at infinity with respect to $\theta$. 
Thus, we require
\begin{subequations}\label{ass-1}
\begin{align}\label{frame}
&\forall\bm Q\in{\rm SO}(d):\quad\
\psi(\bm Q\bm F,\bm m,\zeta,\theta)=\psi(\bm F,\bm m,\zeta,\theta),
\\&\label{ass-HK}
{\rm det}\bm F>0\ \ \ \Rightarrow\ \ \ 
\psi(\bm F,\bm m,\zeta,\theta)\ge\epsilon/({\rm det}\bm F)^p\ \ 
\text{ for some}
\ \ 
p>
\frac{2d}{\BLUE2\gamma{+}2{-}d\color{black}} 
,\ \ \gamma>d/2{-}1,
\\&\label{ass-HK+}
{\rm det}\bm F\le0\ \ \ \Rightarrow\ \ \ 
\psi(\bm F,\bm m,\zeta,\theta)=\infty
\\&\forall {\bm Q}\!\in\!\mathrm{SO}(d)\ \forall {\mathbb G}\!\in\!\R^{d\times d\times d}:\!
\sum_{i,j,k,l,m,n=1}^d\!\!\!\!\!\!{Q}_{in}G_{njp}\mathscr{K}_{ijpklq}(x){Q}_{km}G_{mlq}
=\sum_{i,j,k,l,m,n=1}^d\!\!\!\!\!\!G_{ijp} \mathscr{K}_{ijpklq}(x)G_{klq},\\[0.7em]
&\label{ass-1-entropy1}-\partial^2_{\theta\theta}\psi(\bm F,\bm m,\zeta,\theta)>0,\ \ \lim_{\theta\to 0+ }\partial_\theta\psi(\bm F,\bm m,\zeta,\theta)\le C,\ \text{ and }\  \lim_{\theta\to+\infty}\partial_\theta\psi(\bm F,\bm m,\zeta,\theta)=-\infty,\\[0.7em]
&\lim_{\theta\to+\infty}\frac{\psi(\bm F,\bm m,\zeta,\theta)}{\theta^\delta}=0 \ \text{ for some}\ \delta> 1,\label{ass-1-entropy2}\\[0.7em]
 &\psi(\bm F,\bm m,\zeta,s)\le C\left(1+|\bm m|^{\BLUE{}q\color{black}}+|\zeta|^{\BLUE{}q\color{black}}\right)
\BLUE\ \ \ \text{ with some }\ 1\le q<2d/(d{-}2)\color{black}\color{black}.
\end{align} 
\end{subequations}
for some $\epsilon>0$, $C\BLUE<+\infty\color{black}$, and with $\gamma$ the exponent appearing in 
\eqref{kernel}. Here $Q_{ij}$, $G_{mlq}$, $\mathscr K_{ijpklq}$ denote the 
Cartesian components of the tensors $\bm Q$, ${\mathbb G}$, and $\mathscr K$ 
respectively; moreover 
${\rm SO}(d):=\{\bm Q\in\R^{d\times d};\ \bm Q^\top=\bm Q^{-1},\ \det\bm Q>0\}$ 
denoting the special orthogonal group (i.e.\ the group of 
orientation-preserving rotations). The nonlocal energy $\mathscr{H}$ defined 
in \eqref{def-of-H} enjoys \emph{frame-indifference} in the sense that 
$\mathscr{H}({\mathbb G})=\mathscr{H}({\bm Q}{\mathbb G})$ for all 
${\bm Q}\in\mathrm{SO}(d)$ for any field 
${\mathbb G}:{\varOmega}\to\R^{d\times d\times d}$ 
(here $({\bm Q}{\mathbb G})_{ijk}=Q_{il}G_{ljk}$). 

\medskip

\paragraph{Entropy as state variable.}
The total free energy is defined as the total Helmholtz free energy minus the 
potential energy of the environment, 
\[
\mathcal P(\bm\chi,\bm m,\zeta,\theta)=\mathcal H(\bm\chi,\bm m,\zeta,\theta)-\mathcal L(\bm\chi)-\mathcal Z(\bm\chi,\bm m),
\]
where we recall that $\mathcal H$, $\mathcal L$, and $\mathcal Z$ are defined 
in \eqref{eq:16}, \eqref{eq:17}, and \eqref{eq:1}, respectively. Accordingly, 
the total energy is 
\[
\mathcal U(\bm\chi,\bm m,\zeta,\theta)=\mathcal E(\bm\chi,\bm m,\zeta,\theta)-\mathcal L(\bm\chi)-\mathcal Z(\bm\chi,\bm m),
\]
where 
\[
\mathcal E(\bm\chi,\bm m,\zeta,\theta)=\int_\Omega e(\nabla\bm\chi,\bm m,\zeta,\theta)
+\frac{\kappa_1}2|\nabla\bm m|^2+\frac{\kappa_2}2|\nabla\zeta|^2+\frac 12 \mathfrak H(\nabla^2\bm\chi)\Vdots\nabla^2\bm\chi+\frac 12\nabla\phi\cdot\bm m\,\d x,
\]
with
\begin{equation}
  \label{eq:22}
  e(\bm F,\bm m,\zeta,\theta)=\psi(\bm F,\bm m,\zeta,\theta)+\theta s(\bm F,\bm m,\zeta,\theta),\qquad  s(\bm F,\bm m,\zeta,\theta)=-\partial_\theta\psi(\bm F,\bm m,\zeta,\theta)
\end{equation}
being, respectively, the \emph{bulk} internal energy and the entropy.

The Principle of Minimum of the Total Energy  \cite[15.2.4]{Silha1997Mechanics} states that stable equilibrium states minimize the total energy $\mathcal U$ under the constraint that the total entropy be constant:
\begin{equation}
  \label{eq:20}
\int_\varOmega s(\bm F,\bm m,\zeta,\theta)\,\dx=S_{\rm tot}.
\end{equation}
When using this principle as selection criterion, we find it convenient to 
replace temperature with entropy as independent field,  a device that makes it 
easier for us to handle the constraint \eqref{eq:20}. Such device is at our 
disposal thanks to Assumption \eqref{ass-1-entropy1}, which entails that the 
mapping 
$$
\R^+\ni\theta\mapsto s(\bm\chi,\bm m,\zeta,\theta)
\in[s_{\rm min}(\bm F,\bm m,\zeta),+\infty)
$$
is strictly concave for every choice of the arguments $(\bm F,\bm m,\zeta)$, 
hence it is invertible, with the lower bound on entropy 
$s_{\rm min}(\bm F,\bm m,\zeta)=s(\bm F,\bm m,\zeta,0)\ge -C$
 being a consequence of \eqref{ass-1-entropy1}. To this effect, 
Convex Analysis comes in 
hand,
for one can check that the mapping that delivers the bulk internal energy $e$ 
as function of the state variables $(\bm F,\bm m,\zeta,s)$ is the Legendre 
transform of $-\psi$. We set
\begin{subequations}\label{eq:5}
\begin{equation}
\widetilde e(\nabla\bm\chi,\bm m,\zeta,s)
=\sup_{\theta\in\BLUE\mathbb R\color{black}}\psi_*(\nabla\bm\chi,\bm m,\zeta,\theta)
+\theta s,
\end{equation}
where $\psi_*(\nabla\bm\chi,\bm m,\zeta,\cdot)$ is the unique continuously differentiable 
affine extension of $\psi(\nabla\bm\chi,\bm m,\zeta,\cdot)$ to $\mathbb R$.

As desired, we now have the expression of the internal energy 
\begin{equation}
  \label{eq:23}
\widetilde{\mathcal E}(\bm\chi,\bm m,\zeta,s)=\int_\Omega \widetilde e(\nabla\bm\chi,\bm m,\zeta,s)
+\frac{\kappa_1}2|\nabla\bm m|^2+\frac{\kappa_2}2|\nabla\zeta|^2+\frac 12 \mathfrak{H}(\nabla^2\bm\chi)\Vdots \nabla^2\bm\chi+\frac 12\nabla\phi\cdot\bm m\,\d x,  
\end{equation}
as function of entropy at our disposal. We can then define the total energy as:
\begin{equation}
  \label{eq:24}
  \widetilde{\mathcal U}(\bm\chi,\bm m,\zeta,s)=\begin{cases}
&\widetilde{\mathcal E}(\bm\chi,\bm m,\zeta,s)-\mathcal L(\bm\chi)-\mathcal Z(\bm\chi,\bm m)\qquad \text{if}\qquad \int_\varOmega\widetilde{e}(\bm\chi,\bm m,\zeta,s)\ \dx<+\infty,\\
&+\infty \qquad\text{ otherwise}.
\end{cases}
\end{equation}
\end{subequations}
Having used the affine extension of $\psi$ we guarantee that
\begin{subequations}\label{ass-21}
\begin{equation}
\widetilde e(\bm F,\bm m,\zeta,s)=+\infty\ \text{ if }\ s< s_{\rm min}(\bm F,\bm m,\zeta).
\end{equation}
In addition the function $\widetilde e$ inherits from $\psi$ the following properties:
\begin{align}
&\label{ass-2HK}
{\rm det}\bm F>0\ \ \ \Rightarrow\ \ \ 
\widetilde e(\bm F,\bm m,\zeta,s)\ge\epsilon/({\rm det}\bm F)^p\ \ 
\text{ for some}
\ \ 
p>\frac{2d}
{\BLUE2\gamma{+}2{-}d\color{black}} ,\ \ \gamma>d/2{-}1,
\\&\label{ass-2HK+}
{\rm det}\bm F\le0\ \ \ \Rightarrow\ \ \ 
\widetilde e(\bm F,\bm m,\zeta,s)=\infty,
\end{align} 
\end{subequations}
together with frame indifference.

Noteworthy, the Lagrange multiplier to the constraint $\int_{\varOmega}\zeta\,\d x
=Z_{\rm tot}$ is the (spatially constant) chemical potential, while the Lagrange multiplier to the constraint $\int_\varOmega s\, \d x=S_{\rm tot}$ is temperature.

We remark that for the isothermal nonmagnetic case cf.\ also \cite{KruRou18MMCM}
while for the isothermal nondiffusion in mixed Eulerian/Lagrangian
formulation cf.\ also \cite{RybLus05EEMM}. The latter 
case shows that finer arguments allow for 
admitting $0<\gamma\le d/2-1$ and, avoiding usage of \cite{HeaKro09IWSS},
even for a simpler local 2nd-grade model with the highest-order quadratic
potential of the type 
$\mathscr{H}(\nabla^2\bm\chi)=\int_{\varOmega}|\nabla^2\bm\chi|^2\,\d x$.
We used the nonlocal quadratic variant, \BLUE i.e.\ the concept of
nonsimple materials of the $(2{+}\gamma)$-grade, \color{black}
rather for the later purposes in the dynamical case, \BLUE cf.\ also 
Remark~\ref{rem-nonlocal}.\color{black}

Under assumptions \eqref{ass-1-entropy1} and \eqref{ass-1-entropy2}, the bulk internal energy is coercive with respect to entropy:  
\begin{equation}\label{lem:100}
  \lim_{s\to+\infty}\frac {\widetilde e(\bm F,\bm m,\zeta,s)}{s^{\delta'-\varepsilon}}=+\infty,
\end{equation}
where $\delta'=\delta/(\delta-1)$ is the conjugate exponent of $\delta$ and $0<\varepsilon<\delta'-1$.
Indeed, for every $\theta_*\in\mathbb R$ we have
\[
\widetilde e(\bm F,\bm m,\zeta,s)\ge \psi(\bm F,\bm m,\zeta,\theta_*)+\theta_*s.
\]
In particular, for $\theta_*=s^{\delta'-1}$ we find
\begin{equation}
  \lim_{s\to\infty} \frac{\widetilde e(\bm F,\bm m,\zeta,s)}{|s|^{\delta'-\varepsilon}}\ge \lim_{s\to\infty}\bigg(\frac{\psi(\bm F,\bm m,\zeta,s^{\delta'-1})}{s^{\delta'}}+s^\varepsilon\bigg)=\lim_{s\to\infty}\bigg(\frac{\psi(\bm F,\bm m,\zeta,s^{\frac 1{\delta-1}})}{(|s|^{\frac 1{\delta-1}})^{\delta}}+s^{\varepsilon}\bigg)=+\infty.
\end{equation}
On account of 
\eqref{lem:100}, we choose
\begin{equation}
  1<q<\delta'.
\end{equation}
so that
\begin{equation}\label{eq:27}
  \lim_{s\to+\infty}\frac {\widetilde e(\bm F,\bm m,\zeta,s)}{s^q}=+\infty,
\end{equation}
and we take $L^q(\varOmega)$ as admissible space for the entropy field $s$.

\medskip

\paragraph{The minimization problem.} We can now formulate the following problem:
\begin{align}\label{poroelastic-largestrain}
\left.\begin{array}{ll}
\text{Minimize }& \widetilde{\mathcal U}(\bm\chi,\bm m,\zeta,s),
\displaystyle
\\[.2em]\text{subject to }&
 \displaystyle{\int_{\varOmega}\zeta\,\d x=Z_{\rm tot}\ \text{ and }\ \displaystyle\int_{\varOmega}s\,\d x=S_{\rm tot}},
\\[.8em]
&\displaystyle{\int_{\varOmega}\det(\nabla\bm\chi(x))\,\d x\le{\rm meas}_d(\bm\chi({\varOmega}))}\ \text{  and }\ \ \ 
\bm\chi\big|_{{\varGamma}_\text{\sc D}}^{}\!\!=\bm\chi_{_\text{\sc D}}^{}.
\end{array}\right\}
\end{align}

\medskip

\paragraph{Domain of the energy functional.} 
We still need to provide the formulation \eqref{poroelastic-largestrain} with a proper function-analytic setting. As domain for $\widetilde{\mathcal U}$ we select the space 
\begin{equation}\label{eq:26}
  X=H^{2+\gamma}(\Omega;\R^d)\times H^1(\Omega;\R^d)\times H^1(\Omega)\times L^q(\Omega).
\end{equation}
As a first step to guarantee that the functional $\widetilde{\mathcal U}$ given by \eqref{eq:24} be well defined on the domain $X$, we want to ensure that the  Zeeman-energy functional $\mathcal Z$ given by \eqref{Zeeman-energy}--\eqref{eq:19} makes sense whenever $\bm\chi\in H^{2+\gamma}(\Omega;\R^d)$ and $\bm m\in H^1(\Omega;\R^d)$. To this aim, we assume
\begin{subequations}\label{eq:21}
\begin{align}
 &\bm f\in L^1(\varOmega;\R^d),\\
 &\bm g\in L^1(\varGamma_{\rm N};\R^d),\\
 & \bm{\mathsf h}_{\rm e}\in L^2(\R^d,\R^d).\label{eq:21c}
\end{align}
\end{subequations}
As we show below, \eqref{eq:21} guarantees that our requirements are met. It also shows that the magnetostatic energy is well defined. 

The next result is \BLUE a modification \color{black}
of \cite[Theorem 3.1]{HeaKro09IWSS} 
\BLUE formulated originally for the $W^{2,p}$-spaces. \BLUE Beside 
modifying it for the Hilbert-type $H^{2+\gamma}$-spaces, for reader's convenience \color{black}
here we provide an alternative proof:

\begin{lemma}\label{lem:1}
  Let $\gamma$ and $p$ satisfy assumption \eqref{ass-2HK}, namely,
  \begin{equation}\label{eq:35}
    \gamma>d/2-1\qquad\text{ and }\qquad p>\frac{2d}
{\BLUE2\gamma{+}2{-}d\color{black}} .
  \end{equation}
Assume that $\bm\chi\in H^{2+\gamma}(\varOmega;\R^d)$ has positive determinant 
$J=\operatorname{det}\nabla\bm\chi>0$ in $\varOmega$ satisfying 
$\int_\varOmega J^{-p}\dx<M$ for some constant $M$. Then 
\begin{equation}\label{eq:34}
  \bm\chi\in C^{1,\alpha}(\overline\varOmega),\qquad \text{ with }\alpha=\gamma-\left(\frac d 2-1\right),
\end{equation}
and 
\begin{equation}\label{eq:34+}
  \operatorname{det}\nabla\bm\chi\ge \BLUE\eta\color{black}>0\ \ \text{ in }\overline\varOmega,
\end{equation}
where the constant $\BLUE\eta\color{black}$ depends on 
$\varOmega$, $p$, $\gamma$, and $M$, but not on $\bm\chi$.
\end{lemma}
\begin{proof}
As a start, we recall that \eqref{eq:34} follows from the inclusion 
$H^{2+\gamma}(\varOmega;\R^d)\subset C^{1,\alpha}(\overline\varOmega;\R^d)$. Thus, 
in particular,
  \begin{equation}
    J=\operatorname{det}\nabla\bm\chi\in C^\alpha(\overline\varOmega),
  \end{equation}
  and thus there exists a constant $C_\alpha>0$ such that
  \begin{equation}
    |J(x)-J(y)|\le C_\alpha |x-y|^\alpha\quad\forall x,y\in\overline\varOmega.
  \end{equation}
  Next, given $\eta>0$, define
    $A_\eta:=\{x\in\varOmega:J(x)\le \eta\}$.
  Then $\operatorname{meas}_d(A_\eta)\eta^{-p}\le \int_\Omega J^{-p}\ \d x\le M$. Hence, the $d-$dimensional Lebesgue measure of $A_\eta$ satisfies 
    $\operatorname{meas}_d(A_\eta)\le M \eta^p$.
    We assume that $\eta$ be sufficiently small so that $A_\eta^c=\BLUE{}\overline{\varOmega\setminus A_\eta}$ is a nonempy set\color{black}. Then, given $x \in A_\eta$, we have
  \begin{equation}
    x\in A_\eta\quad\Rightarrow\quad\operatorname{dist}(x,A_\eta^c)\le (\operatorname{meas}_d(A_\eta))^{1/d}.
  \end{equation}
  Take $y\in A_\eta^c$ such that $|x-y|=\operatorname{dist}(x,A_\eta^c).$ 
Then $|J(y)-J(x)|\le C_\alpha|x-y|^\alpha\le C_\alpha (\operatorname{meas}_d(A_\eta))^{\alpha/d}$. Hence, since $J(y)>\eta$,
  \begin{equation}
    J(x)>J(y)-|J(x)-J(y)|>\eta-C_\alpha|x-y|^\alpha\ge  \eta-C_\alpha(\operatorname{meas}_d(A_\eta))^{\alpha/d}\ge \eta- C_\alpha M^{\alpha/d}\eta^{p\alpha/d}.
  \end{equation}
  It follows from  \eqref{eq:35} that
  $p\alpha/d>1$. Thus,
  \begin{equation}
  \displaystyle J(x)\ge\operatorname{sup}_{0<\eta<1}(\eta- C_\alpha M^{\alpha/d}\eta^{p\alpha/d})=:\eta>0.
    \end{equation}
\end{proof}
\BLUE{}So far, the definitions \eqref{eq:1} of the Zeeman energy and \eqref{eq:18} are only formal. In fact, \eqref{eq:1} involves the pull-back of the external magnetic field into the reference configuration. Thus, in order for the integral on the right-hand side of \eqref{eq:1} to make sense, the pull-back must be an element of $L^2(\varOmega;\mathbb R^d)$. Likewise, the magnetostatic energy defined in \eqref{eq:19} involves a magnetic potential obtained by solving the elliptic problem \eqref{maxwell+}. The right-hand side of \eqref{maxwell+} involves the push-forward $\bm{\mathsf m}$ defined in the first of \eqref{m-pull-back}. To guarantee that the expression \eqref{eq:19} makes sense, we shall show that the extension of the spatial magnetization defined in the first of \eqref{m-pull-back} is in $L^2(\mathbb R^d;\mathbb R^d)$. We shall accomplish this task in the next two lemmas. Although the proofs use standard tools from measure theory, we include them for  the reader's convenience. \color{black} 

\begin{lemma}\label{lem:55}
Assume that $\bm\chi:\varOmega\to\R^d$ be an injective mapping of class 
$C^1(\varOmega)$ satisfying 
$\operatorname{det}(\nabla\bm\chi)\BLUE\ge\eta\color{black}>0$ in 
$\varOmega$. Then, \BLUE{}if $L$ is a Lebesgue measurable subset of $\bm\chi(\varOmega)$, then $\bm\chi^{-1}(L)$ is a 
Lebesgue measurable subset of $\varOmega$.\color{black}
\end{lemma}

\begin{proof}
Every Lebesgue measurable subset of $\bm\chi(\varOmega)$ is the union of 
\BLUE some \color{black} 
$F_\sigma$-set\BLUE{}s (\color{black}\emph{i.e.} countable union\BLUE{}s \color{black} of relatively closed subsets 
of $\varOmega$) and a set with null Lebesgue measure. Hence,
\[
\text{$L=V\bigcup N$, \ \ where $V=\bigcup_n V_n$ with $V_n\subset\varOmega$ relatively closed, and $\operatorname{meas}_d(N)=0$.}
\]

Since $\chi$ is continuous on $\Omega$, $\chi^{-1}$ maps relatively closed subsets of $\bm\chi(\varOmega)$ into relatively closed subsets of $\varOmega$. Thus $\chi^{-1}(V)=\chi^{-1}(\cup_n V_n)=\cup_n\chi^{-1}(V_n)$ is a union of relatively closed sets of $\varOmega$, that is, a $F_\sigma$ subset of $\varOmega$. Furthermore, since $N$ has null Lebesgue measure, it can be covered with a countable family $\{U_n\}_n$ of relatively open sets $U_n\subset\bm\chi(\varOmega)$ whose union has arbitrarily small volume $\varepsilon$:
\[
\text{$N\subset\bigcup_n U_n$ with $U_n\subset\varOmega$ relatively open and $\operatorname{meas}_d\Big(\bigcup_n U_n\Big)\le\varepsilon$.}
\]
Since $\nabla\bm\chi\in C(\varOmega)$ and since its Jacobian $\operatorname{det}(\nabla\bm\chi)$ is bounded from below by a positive constant, it follows from the inverse function theorem  that $\bm\chi^{-1}$ is uniformly Lipschitz continuous. Hence there is a constant $M>0$ such that
\[
\operatorname{meas}_d\Big(\chi^{-1}\Big(\bigcup_n U_n\Big)\Big)\le M\varepsilon.
\]
The arbitrariness of $\varepsilon$ implies that the exterior measure of $\chi^{-1}(N)$ is zero and hence, by the completeness of the Lebesgue measure, $\chi^{-1}(N)$ has null Lebesgue measure. Thus the proof is completed.
\end{proof}


\begin{lemma}\label{lem:2}
Assume that $\bm\chi$ satisfies all assumption of Lemma \ref{lem:55} and that $\bm m\in H^1(\varOmega;\R^d)$. Then:
 
a) if the spatial external field $\bm{\mathsf h}_{\rm e}$ satisfies \eqref{eq:21c}, then the function $\bm{h}_{\rm e}:\varOmega\to\R^d$ defined in \eqref{eq:19} satisfies
$\bm{h}_{\rm e}\in L^2(\varOmega;\R^d)$,
and hence the \BLUE{}integral on the right-hand side of \eqref{eq:1}
makes sense. \color{black}

b) if $\bm m\in H^1(\varOmega;\R^d)$ then the spatial magnetic field $\overline{\bm{\mathsf m}}$ defined in \eqref{m-pull-back} and \eqref{eq:12} satisfies
$\overline{\bm{\mathsf m}}\in L^2(\R^d;\R^d)$,
and hence the elliptic problem \eqref{maxwell+}, which defines 
$\upphi\in H^1(\R^d;\R^d)$  is well posed, \BLUE
thus the integral on the right-hand side of \eqref{eq:18} makes sense. \color{black}
\end{lemma}

\begin{proof}
By Lemma \ref{lem:55}, the function $\bm{\mathsf h}_{\rm e}\circ\bm\chi$ is 
Lebesgue measurable. In fact, if $B\subset\R^d$ is a Borel set then $(\bm{\mathsf h}_{\rm e}\circ\bm\chi)^{-1}(B)=\bm\chi^{-1}(\bm{\mathsf h}_{\rm e}^{-1}(B)\cap\Omega)$. The set $L=\bm{\mathsf h}_{\rm e}^{-1}(B)\cap\Omega$ is Lebesgue measurable, and hence $\bm\chi^{-1}(L)$ is Lebesgue measurable. Moreover, we have $\bm{\mathsf h}_{\rm e}\circ\bm\chi\in L^2(\varOmega;\R^d)$, since, by the change of variables formula and by the boundedness of $\nabla\bm\chi$ on $\varOmega$,
\[
  \int_\varOmega|\bm{\mathsf h}_{\rm e}\circ\bm\chi|^2\d x=\int_{\bm\chi(\varOmega)}(J^{-1}\circ\bm\chi^{-1})|\bm{\mathsf h}_{\rm e}|^2\d z\le C \|\bm{\mathsf h}_{\rm e}\|_{L^2(\R^d;\R^d)}^2,\qquad J^{-1}=\frac 1{\operatorname{det}(\nabla\bm\chi)}.
\]
To conclude the proof, we observe that the mapping $L^2(\Omega;\R^d)\ni \bm f\mapsto \bm F^\top\bm f\in L^2(\Omega;\R^d)$ is a Nemytski\u{\i} operator, since $\bm F\in C(\overline\varOmega)$. Thus, $(a)$ is proved. 

We now prove $(b)$. We first show that $\overline{\bm{\mathsf m}}:\R^d\to\R^d$ is a measurable function. Given a Borel set $B\subset\R^d$, we have $\overline{\bm{\mathsf m}}^{-1}(B)=(\overline{\bm{\mathsf m}}^{-1}(B)\cap\Omega)\cup (\overline{\bm{\mathsf m}}^{-1}(B)\cap\Omega^c)$ where $\Omega^c$ is the complement of $\Omega=\bm\chi(\varOmega)$ in $\R^d$. Since $\overline{\bm{\mathsf m}}$ vanishes outside $\Omega$, either $\overline{\bm{\mathsf m}}^{-1}(B)\cap\Omega^c=\Omega^c$, or $\overline{\bm{\mathsf m}}^{-1}(B)\cap\Omega^c=\emptyset$. In both cases, $\overline{\bm{\mathsf m}}^{-1}(B)\cap\Omega^c$ is Lebesgue measurable, because the set $\Omega$, being the image of $\varOmega$ under the homeomorphism $\bm\chi$, is a Borel set and hence it is Lebesgue measurable. In addition, we have $\overline{\bm{\mathsf m}}^{-1}(B)\cap\Omega=\bm{\mathsf m}^{-1}(B)$, where $\bm{\mathsf m}$ is given in \eqref{m-pull-back}. Recalling that $\bm{\mathsf m}^{-1}=\bm g\circ\bm\chi^{-1}$, where $\bm g=J^{-1}\bm F\bm m$, we have $\bm{\mathsf m}^{-1}(B)=\bm\chi(\bm g^{-1}(B))=\bm\chi(L)$, where $G=\bm g^{-1}(B)$ is a Lebesgue measurable set since $\bm g$ is an element of $L^2(\varOmega;\R^d)$. 
\BLUE{}It is then easy to conclude that  
$\bm{\mathsf m}^{-1}(B)$ is Lebesgue measurable.\color{black}  
\end{proof}

\BLUE{}We now can use Lemma~\ref{lem:2} to show \color{black}that $\widetilde{\mathcal U}$ is well 
defined on the admissible space. In fact, if 
$\widetilde{\mathcal E}(\bm\chi,\bm m,\zeta,s)<+\infty$, then by the 
coercivity \eqref{ass-2HK}, the deformation $\bm\chi$ satisfies the hypotheses 
of Lemma~\ref{lem:1}. Hence, by Lemma~\ref{lem:2}, $\mathcal Z(\bm\chi,\bm m)$ 
the \BLUE{}definitions of $\mathcal H_{\rm mag}(\bm\chi,\bm m)$ make sense.\color{black} 
\BLUE Therefore, assuming that \color{black}
the applied loads and the external magnetic field satisfy \eqref{eq:21},
the expression \eqref{eq:24} defines a functional from the space $X$ \BLUE specified \color{black} in \eqref{eq:26} 
to $\R\cup\{+\infty\}$. 
We are now going to show that the functional has at least a minimizer:

\begin{theorem}[Existence of minimizing static configurations]\label{thm1}
Let $\widetilde e:\R^{d\times d}\times\R^d\times\R\times\R\to\R^+\cup\{+\infty\}$ 
be continuous and $\theta\to\psi(\bm F,\bm m,\zeta,\theta)$ be twice 
continuously differentiable for all $(\bm F,\bm m,\zeta)\in\mathbb R^{d\times d}\times\mathbb R^d\times\R$ satisfying the coercivity \eqref{ass-21} and also \eq{kernel}, 
and let 
\BLUE $\mathscr{H}$ be given by \eqref{def-of-H} with the kernel $\mathscr{K}$
\color{black} satisfying \eq{kernel}, 
$\kappa>0$, $p>d$, ${\rm Meas}_{d-2}({\varGamma}_\text{\sc D}^{})>0$
and $\bm\chi_\text{\sc D}^{}$ allow an extension to ${\varOmega}$ that renders  
$\widetilde{\mathcal U}(\bm\chi_{\rm D},\bm m_*,\zeta_*,\theta_*)$ finite for some $(\bm m_*,\zeta_*,\theta_*)$  in $H^1({\varOmega};\R^d)\times H^1({\varOmega})\times H^1({\varOmega})\times H^1(\varOmega)$ such that $\theta_*>0$. 
Then problem \eqref{poroelastic-largestrain} with $\widetilde{\mathcal U}$ given by \eqref{eq:5} has a solution
$(\bm\chi,\bm m,\zeta,s,\phi)\in H^{2+\gamma}({\varOmega};\R^d)\times H^1({\varOmega};\R^d)\times H^1({\varOmega})\times H^1({\varOmega})\times 
\BLUE{}H^1(\R^d)\color{black}$.
\end{theorem}

\BLUE
Let us note that we do not make any statement about possible
non-negativity of temperature $\theta$ which can be reconstructed 
from the solution $(\bm\chi,\bm m,\zeta,s,\phi)$ through
$\theta=\partial_s\widetilde e(\nabla\bm\chi,\bm m,\zeta,s)$, cf.\ \eq{eq:5},
in contrast to the evolution case later where non-negativity can be granted. 
\color{black}

\begin{proof}[Proof of Theorem~\ref{thm1}]
We consider a minimizing sequence
$\{(\bm\chi_k,\bm m_k,\bm \zeta_k,s_k)\}_{k\in\N}\subset X$ for 
the functional $\widetilde{\mathcal U}$. We denote by $\bm{\mathsf m}_k$ and 
$\overline{\bm{\mathsf m}}_k$ the corresponding spatial magnetizations, 
defined by \eqref{m-pull-back} and \eqref{eq:12} with the pair 
$(\bm\chi,\bm m)$ being replaced by $(\bm\chi_k,\bm m_k)$. 
From the coercivity properties of the energy we have that the sequence 
\begin{align}
&\bm\chi_k\to \bm\chi&&\text{weakly in }H^{2+\gamma}(\varOmega;\R^d),&&&&&&\\
&\bm m_k\to \bm m&&\text{weakly in }H^{1}(\varOmega;\R^d),&&\\
&\zeta_k\to \zeta&&\text{weakly in }H^{1}(\varOmega),&&\\
&s_k\to s&&\text{weakly in }L^q(\varOmega).&
\end{align}
Hence, by compact embedding,
\begin{subequations}\label{eq:28}
\begin{align}
&&&\bm\chi_k\to\bm\chi&&\text{strongly in }C^{1,\alpha}(\overline\varOmega;\R^d)
&&\text{for some }0<\alpha<1,&&&&\label{eq:28a}\\
&&&\bm m_k\to \bm m&&\text{strongly in }L^{\BLUE{}q\color{black}}(\varOmega;\R^d)
&&\text{for all }\BLUE1\le q<2d/(d{-}2)\color{black},
&&\label{eq:28b}
\\
&&&\zeta_k\to \zeta&&\text{strongly in }L^{\BLUE{}q\color{black}}(\varOmega)
&&\text{for all }\BLUE1\le q<2d/(d{-}2)\color{black}.&&
\end{align}
\end{subequations}
Moreover, by Lemma \ref{lem:1} and by the coercivity property \eqref{ass-2HK} 
we have that
\begin{equation}
  J_k:=\operatorname{det}\nabla\bm\chi_k\ge \BLUE\eta\color{black}>0
\quad\text{in }\varOmega,
\end{equation}
where the constant $\BLUE\eta\color{black}>0$ does not depend on $k$. 
We have, that $J_k$ converges uniformly to $J$ in $\varOmega$, and hence 
\begin{equation}
  J=\operatorname{det}\nabla\bm\chi\ge\BLUE\eta\color{black}\qquad \text{in }\varOmega.
\end{equation}
We denote by $\overline{\bm{\mathsf m}}$ the spatial magnetization corresponding to the pair $(\bm\chi,\bm m)$ and by $\bm{\mathsf m}$ its restriction on $\bm\chi(\varOmega)$. It follows from \eqref{eq:28}
\begin{equation}
  \int_{\R^d}\!|\overline{\bm{\mathsf m}}_k|^2\,\d z
= \int_{\bm\chi(\varOmega)}\!|\bm{\mathsf m}_k|^2\,\d z
=\int_{\varOmega}\!J_k^{-1}|\bm F_k\bm m_k|^2\,\d x\to
\int_{\varOmega}\!J^{-1}|\bm F\bm m|^2\d x
= \int_{\bm\chi(\varOmega)}\!|\bm{\mathsf m}|^2\,\d z=\int_{\R^d}\!|\overline{\bm{\mathsf m}}|^2\,\d z,
\end{equation}
and also that, for every test field $\bm\upvarphi\in H^1(\R^d;\R^d)$,
\begin{equation}\label{eq:29}
  \int_{\R^d}(\overline{\bm{\mathsf m}}_k-\bm{\mathsf m})\cdot\bm\upvarphi\,\d z\to 0.
\end{equation}
Thus, the $L^2$ norms of $\overline{\bm{\mathsf m}}_k$ converge to those of $\overline{\bm{\mathsf m}}$, hence the weak convergence in \eqref{eq:29} yields strong convergence:
\begin{equation}
  \overline{\bm{\mathsf m}}_k\to\overline{\bm{\mathsf m}}\quad\text{strongly in }L^2(\R^d;\R^d).
\end{equation}
Since $\int_{\R^d}|\gradop\upphi_k|^2\, \d z=\int_{\R^d}\overline{\bm{\mathsf m}}\cdot\gradop\upphi_k\,\d z$, we have that $\upphi_k$ is bounded in $H^1(\R^d)$, hence $\upphi_k\to\upphi$ weakly in $H^1(\R^d)$. Moreover, for $\upvarphi\in H^1(\R^d)$ a test function, by passing to the limit in the equation $\int_{\R^d}\gradop\upphi_k\cdot\gradop\upvarphi\, \d z=\int_{\R^d}\overline{\bm{\mathsf m}}_k\cdot\gradop\upvarphi\,\d z$ we obtain that $\phi$ is the unique $H^1(\R^d)$ solution of \eqref{maxwell+}. From the strong convergence of $\overline{\bm{\mathsf m}}_k$ it follows that
\begin{equation}\int_{\R^d}|\gradop\upphi_k|^2\, \d z\to \int_{R^d}\overline{\bm{\mathsf m}}\cdot\gradop\upphi\,\d z=\int_{\R^d}|\gradop\upphi|^2\, \d z.
\end{equation}
Thus, in particular, $\upphi_k\to\upphi\quad\text{strongly in }H^1(\R^d)$. Now, thanks to the growth properties of $\widetilde e$, and to its convexity with respect to $s$, we have that
\begin{equation}
  \liminf_{k\to\infty}\int_\varOmega \widetilde e(\bm\chi_k,\bm m_k,\zeta_k,s_k)\,\BLUE\d x\color{black}
\ge \int_\varOmega \widetilde e(\bm\chi_k,\bm m_k,\zeta_k,s_k)\,\BLUE\d x\color{black}.
\end{equation}
The conclusion of the proof follows from weak lower semicontinuity of the remaining quadratic terms of the energy.
\end{proof}

\begin{remark}[Lagrangian versus mixed Eulerian/Lagrangian setting]\label{rem-frame-indifference}
\upshape
It is noteworthy to realize the relation to the  mixed Eulerian/Lagrangian
setting used  e.g.\ in
\cite{JamKind93TMAT,RybLus05EEMM}, worked with the magnetization in the 
deformed configuration $\bm{\mathsf m}$ and denoting the energy used there 
by $\varphi_{_{\rm EL}}$. 
Let us now denote ``our'' energy $\varphi$ used here by $\varphi_{_{\rm L}}$.
In contrast to \eqref{frame} for $\varphi_{_{\rm L}}$, the frame indifference for 
$\varphi_{_{\rm EL}}$ means that
$\varphi_{_{\rm EL}}(x,{\bm R}{\bm F},{\bm R}\tilde {\bm m},\zeta,\theta)=\varphi_{_{\rm EL}}(x,{\bm F},\tilde {\bm m},\zeta,\theta)$ for all ${\bm R}\in{\rm SO}(d)$.
Both approaches are mutually equivalent.
Indeed, 
taking $\varphi_{_{\rm EL}}(x,{\bm F},\tilde {\bm m}):=\tilde\varphi_{_{\rm M}}(x,{\bm F}^\top {\bm F},{\bm F}^\top \tilde {\bm m})$ with some ``material'' stored energy $\tilde\varphi_{_{\rm M}}:{\varOmega}\times\R^{d\times d}\times\R^d\to\R\cup\{\infty\}$ and with $\tilde {\bm m}$ a 
placeholder for $\bm{\mathsf m}{\circ}\bm{\chi}$ as used in \cite{JamKind93TMAT,RybLus05EEMM}, 
is the same as taking $\varphi_{_{\rm L}}(x,{\bm F},{\bm m}):=\varphi_{_{\rm M}}(x,{\bm C},{\bm m})
=\tilde\varphi_{_{\rm M}}(x,{\bm C},{\bm C}{\bm m}/\sqrt{\det {\bm C}})$ with 
${\bm m}=(\det {\bm F}){\bm F}^{-1}\bm{\mathsf m}{\circ}\bm{\chi}$ the pulled-back magnetization according 
\eqref{m-pull-back}; here, for notational simplicity, 
we avoided $\zeta$ and $\theta$ dependence. More in detail,
\begin{align*}
\varphi_{_{\rm L}}(x,{\bm F},{\bm m})&:=\varphi_{_{\rm M}}(x,{\bm C},{\bm m})
=\tilde\varphi_{_{\rm M}}\Big(x,{\bm C},\frac{{\bm C}{\bm m}}{\sqrt{\det {\bm C}}}\Big)
\\&=\tilde\varphi_{_{\rm M}}\Big(x,{\bm F}^\top {\bm F},(\det {\bm F}){\bm F}^\top {\bm F}\frac{{\bm F}^{-1}\bm{\mathsf m}{\circ}\bm{\chi}}{\det {\bm F}}\Big)
=\tilde\varphi_{_{\rm M}}(x,{\bm F}^\top {\bm F},{\bm F}^\top \bm{\mathsf m}{\circ}\bm{\chi})=\varphi_{_{\rm EL}}(x,{\bm F},\tilde {\bm m}).
\end{align*}
Alternatively, one can consider  
$\varphi_{_{\rm EL}}(x,{\bm F},\tilde {\bm m})
:=\tilde\varphi_{_{\rm M}}(x,{\bm F}^\top {\bm F},{\bm F}^{-1}\tilde {\bm m})$
which also guarantees the desired frame indifference. 
Then it is the same as taking 
$\varphi_{_{\rm L}}(x,{\bm F},{\bm m}):=\varphi_{_{\rm M}}(x,{\bm C},{\bm m})
=\tilde\varphi_{_{\rm M}}(x,{\bm C},{\bm m}/\sqrt{\det {\bm C}})$
because
\begin{align*}
\varphi_{_{\rm L}}(x,{\bm F},{\bm m})&:=\varphi_{_{\rm M}}(x,{\bm C},{\bm m})
=\tilde\varphi_{_{\rm M}}\Big(x,{\bm C},\frac{{\bm m}}{\sqrt{\det {\bm C}}}\Big)
\\&=\tilde\varphi_{_{\rm M}}\Big(x,{\bm F}^\top {\bm F},(\det {\bm F}){\bm F}^{-1}\frac{\bm{\mathsf m}{\circ}\bm{\chi}}{\det {\bm F}}\Big)
=\tilde\varphi_{_{\rm M}}(x,{\bm F}^\top {\bm F},{\bm F}^{-1}\bm{\mathsf m}{\circ}\bm{\chi})=\varphi_{_{\rm EL}}(x,{\bm F},\tilde {\bm m}).
\end{align*}
Anyhow, although both ``our'' fully Lagrangian and the mixed 
Lagrangian/Eulerian approaches are mutually equivalent as far as the stored 
energy concerns, evolution formulated for the reference magnetization ${\bm m}$ 
is more amenable to mathematical analysis than for the magnetization in the 
deformed configuration $\bm{\mathsf m}$, not speaking about a 
formulation of the problem in the fully Eulerian setting as used e.g.\ 
in \cite{DesPod96CTDF}.
For the quasistatic extension of the incompressible model formulated in 
mixed Lagrangian/Eulerian setting 
we refer to \cite{KrStZe15ERIM} where the solution concept fully relies on the energetic-solution formulation.
\end{remark}

\BLUE{}
\begin{remark}[Concept of nonlocal nonsimple materials]\label{rem-nonlocal}
\upshape
The $H^1$-regularization for $\bm m$ and $\zeta$  versus the 
$H^{1+\gamma}$-regularization for the deformation gradient $\nabla\bm\chi$ is 
conceptually a bit inconsistent. Of course, using the nonlocal 
$H^{1+\gamma}$-regularization for all variables, i.e.\ also for $\bm m$ and 
$\zeta$, would make the desired mathematical effects, too, but it would 
be more heavy as far as notation concerns, which is the reason why we avoid 
it. In the static case in this section, another conceptually consistent 
option would be to consider the nonlinear (i.e.\ governed by non-quadratic 
potential) but local $W^{1,p}$-regularization with $p>d$ for $\nabla\bm\chi$, 
$\bm m$, and $\zeta$. This latter scenario however would not work for 
evolutionar problems in Sect.~\ref{sec-dynamic} for analytical reasons, 
leading to hyperbolic problems nonlinear in the main part, which is why we 
do not consider it here. On the other hand, the concept of on-locality in continuum theory is well accepted \cite{Erin02NCFT,Edelen1976}, and 
has also been involved in attempts to resolve certain questions in the theory of material defects \cite{Kuni82-3EMM}. 
\end{remark}

\color{black}

\section{\BLUE{}A PDE system describing dynamics in the Lagrangian formulation}
\label{sec-dynamic}
\BLUE{}
In this section we formulate an evolution problem for the state fields 
$\bm\chi$, $\bm m$, $\zeta$, $\theta$. In accordance with a standard 
thermodynamical practice, the natural energy to work with is the free energy. 
We present the collection of relevant balance equations and we make a 
selection of thermodynamically consistent constitutive equations. The result 
is a system of hyperbolic-parabolic partial differential equations 
\eqref{system}. This system, along with the boundary conditions \eqref{BC} 
and the initial conditions \eqref{IC}, is the basis for the notion of weak 
solutions we shall provide in the next section. 

As the inertia is now involved, we do not need to consider a particular
fixation of the body as we made in \eq{eq:fixation} and, 
rather for notational \color{black} simplicity we assume that there are no 
external constraints on a part of the boundary,
\BLUE i.e.\  \color{black} $\varGamma_{\rm D}=\emptyset$ and 
$\varGamma_{\rm N}=\varGamma$. 
A technical assumption concerning the bulk part $\psi$ of the free energy is
\begin{equation}\label{ansatz}
  \partial_{\bm F\theta}^2\psi(\bm F,\bm m,\zeta,\theta)=\bm 0.
\end{equation}
Thanks to this assumption, we can write the free energy as
\begin{equation}\label{ansatz2}
\psi(\bm F,\bm m,\zeta,\theta)=\psi_{\textsc{me}}(\bm F,\bm m,\zeta)+ \psi_{\textsc{th}}(\bm m,\zeta,\theta),
\end{equation}
where $\psi_{\textsc{me}}(\bm F,\bm m,\zeta)=\psi(\bm F,\bm m,\zeta,0)$ and $\psi_{\textsc{th}}(\bm m,\zeta,\theta)=-\int_0^\theta \vartheta\partial^2_{\vartheta\vartheta}\psi(\bm I,\bm m,\zeta,\vartheta)\operatorname{d}\!\vartheta+\theta\partial_\theta\psi(\bm I,\bm m,\zeta,\theta)$. \BLUE{}We notice here for later use that an immediate consequence of \eqref{ansatz2} is that the bulk part $e(\bm F,\bm m,\zeta,\theta)$ of the internal energy, which we have introduced in \eqref{eq:22}, can be written as:
\begin{equation}\label{eq:36}
  e(\bm F,\bm m,\zeta,\theta)=e_{\textsc{th}}(\bm m,\zeta,\theta)+\psi_{\textsc{me}}(\bm F,\bm m,\zeta),
\end{equation}
where
\begin{equation}
  \label{eq:10}
e_{\textsc{th}}(\bm m,\zeta,\theta)=\psi_{\textsc{th}}(\bm m,\zeta,\theta)-\theta\partial_\theta\psi_{\textsc{th}}(\bm m,\zeta,\theta)
\end{equation}
is the thermal part of the internal energy.\color{black}

The restriction \eqref{ansatz} uncouples temperature from the deformation gradient, but not from magnetization and concentration.
Thus, it allows us to model 
\BLUE{}the influence of temperature 
\color{black}on magnetic and chemical behavior, 
such as the ferro-to-para-magnetic 
phase transformation, as in \cite{PoRoTo10TCTF,RoubT2013ARMA}, or the or metal-hydride phase transformation 
like in \cite{Anan11TMCT,RouTom14THSM} and combination
of both as in the references we cite in the introduction. Unfortunately, this restriction excludes 
other thermally-sensitive phenomena such as the martensite/austenite phase transformation. Yet, it might be removed
by adding more ingredients to our model. For example, by introducing an auxiliary ``phase indicator'', as explained for instance in \cite{RoubT2010ZAMM,RoubT2013ARMA}, or by introducing a viscous contribution to the stress, which can be made physical using the approach in \cite{LewicM2013local} or in \cite{MielkeRoubicek2016viscolargestrains}.


\BLUE{}Compared to the above presented static model, an essential simplification consists in neglecting \color{black}the influence of the demagnetizing \BLUE{}energy\color{black}. 
This is motivated purely mathematically because the injectivity (at least 
almost everywhere) of the deformation $\bm{\chi}$ is not granted in combination 
with inertia which is, however, needed to control time derivative of 
$\nabla\bm{\chi}$ under absence of viscosity (which would otherwise bring 
\BLUE{}other \color{black} serious difficulties). 
This injectivity is needed in the 
\BLUE{}magnetic potential,
\color{black}which inevitably
involves \eqref{maxwell+} where  $\bm{\chi}^{-1}$ occurs, otherwise we can benefit
from our purely Lagrangian formulation of the problem. Ignoring of the 
demagnetizing
\BLUE{}energy
\color{black}is to some extent eligible in situations when the magnet 
is long like in Figure~\ref{fig-elastic-magnet} (or a toroidal shape)
so that the hysteretic loops are rather rectangular. 
On the other hand, ignoring of the possible selfcontact (often 
accepted in engineering simulations) is to some extent 
eligible in geometrically ``bulky'' situations or under 
particular loading.

It is worth noticing that the mechanical actions of the magnetic field manifest themselves not only through a body force, but also through a stress. This can be easily seen by computing the variation of the Zeeman energy \eqref{eq:1}:
\[
{\rm D}\mathcal Z(\bm\chi,\bm m)[\widetilde{\bm\chi},{\widetilde{\bm m}}] \color{black}=\int_{\varOmega}((\bm{\mathsf h}_{\rm e}{\circ}\bm\chi)\otimes\bm m)\cdot \nabla\widetilde{\bm\chi}+((\gradop\bm{\mathsf h}_{\rm e}){\circ}\bm\chi)^\top\nabla\bm\chi\,\bm m\cdot\widetilde{\bm\chi}+((\nabla\bm{\chi})^\top \bm{\mathsf h}_{\rm e}{\circ}\bm{\chi})\cdot{\widetilde{\bm m}}.
\]
Guided by this result, we write the balance of linear momentum as:
\begin{equation}\label{eq:39}
  \varrho\DDT{\bm \chi}-\operatorname{div}({\bm S}-\operatorname{div}\mathbb S)=\bm f_{\textsc{mag}}+\operatorname{div}{\bm S}_{\textsc{mag}}+\bm f, 
\end{equation}
where $\bm S$ and $\mathbb S$ are respectively, the standard stress and the hyperstress, moreover $\bm S_{\textsc{mag}}=(\bm{\mathsf h}_{\rm e}{\circ}\bm\chi)\otimes\bm m$ is the magnetic stress, 
\BLUE$((\gradop\bm{\mathsf h}_{\rm e}){\circ}\bm\chi)^\top\nabla\bm\chi\,\bm m$
\color{black}is the magnetic force, and $\bm f$ is \BLUE{}the mechanical \color{black}body force. The accompanying boundary conditions are:
\begin{equation}
  (\bm S-\operatorname{div}_{\textsc S}\mathbb S)\bm n=\bm g,\qquad \mathbb S:\bm n\otimes\bm n=\bm 0.
\end{equation}
Proceeding in a similar fashion, we write the balance of magnetic forces as
\begin{align}\label{eq:7}
  -\operatorname{div}\bm L+\bm l=\bm l_{\textsc{mag}},
\end{align}
where $\bm L$ is the magnetic stress, $\bm l$ is the magnetic internal force, and $\bm l_{\textsc{mag}}=(\nabla\bm{\chi})^\top \bm{\mathsf h}_{\rm e}{\circ}\bm{\chi}$ is the magnetic external force.
 Moreover, we suppose that the evolution of concentration is influenced, besides, diffusion, by a system of microforces obeying the balance equation:
\begin{equation}
  -\operatorname{div}\bm\varsigma+\varsigma=0,
\end{equation}
where $\bm\varsigma$ is a vectorial microstress and $\varsigma$ is a microforce. The internal power expended by the aforementioned force systems is
\begin{equation}
  \pi_{\textsc{int}}=\bm S:\nabla\DT{\bm\chi}+\mathbb S\Vdots\nabla^2\DT{\bm\chi}+\bm C:\nabla\DT{\bm m}+\bm c\cdot\DT{\bm m}+\bm\varsigma\cdot\nabla\DT\zeta+\varsigma\DT\zeta.
\end{equation}
The balance laws expressing conservation of mass and energy are
\begin{align}\label{eq:9}
  &\DT\zeta+\operatorname{div}\mb j=0
\BLUE\quad\text{ and }\quad\color{black} 
\DT e_{\textsc{tot}}+\operatorname{div}(\bm q+\mu\bm j)=\pi_{\textsc{int}},
\end{align}
where $\mb j$ is the flux of diffusant, $\bm q$ is the heat flux, 
\BLUE $\mu$ is the chemical potential\color{black}, 
and the total energy density is
\begin{equation}\label{eq:37}
   \BLUE{}e_{\textsc{tot}}=\color{black}
e+\frac {\kappa_1}2\nabla\bm m:\nabla{\bm m}+\frac{\kappa_2}2 \nabla\zeta\cdot\nabla\zeta+\frac 12 \mathfrak H(\nabla^2\bm\chi)\Vdots \nabla^2{\bm\chi}\,.
\end{equation}

\BLUE
Let us point out that, in contrast to Sect.~\ref{sec-static} which exploited 
merely concentration, the chemical potential is here a primitive state 
variable for which there is later an equation \eqref{system.c}.
\color{black}

Once balance equations have been established, we \BLUE are now to 
\color{black} provide  \BLUE them by \color{black} constitutive 
\BLUE relations. \color{black} 
These are selected \BLUE to be consistent \color{black} 
with the entropy inequality, which in the bulk reads:
\begin{equation}
  \DT s+\operatorname{div}(\theta^{-1}\mb q)\ge 0.
\end{equation}
Combining the entropy inequality with the balance of energy and with the balance of mass we \BLUE{}arrive at the following form of the dissipation inequality:\color{black}
\begin{align}\label{eq:2}
\DT \psi+s\DT\theta+\kappa_1\nabla\bm m:\nabla\DT{\bm m}&
+\kappa_2 \nabla\zeta\cdot\nabla\DT\zeta+\mathfrak H(\nabla^2\bm\chi)\Vdots \nabla^2\DT{\bm\chi}+\theta^{-1}\bm q\cdot\nabla\theta\nonumber\\
&
\le\bm S:\nabla\DT{\bm\chi}+\mathbb S\Vdots\nabla^2\DT{\bm\chi}+\bm C:\nabla\DT{\bm m}+\bm c\cdot\DT{\bm m}+\bm\varsigma\cdot\nabla\DT\zeta+(\mu+\varsigma)\DT\zeta-\bm j\cdot\nabla\mu.
\end{align}
\BLUE{}The last inequality serves as a selection criterion for thermodynamically consistent constitutive equations. For the sake of brevity, and to avoid the introduction of excessive notation, we shall limit ourselves to a limited class of constitutive equations. Precisely, we restrict attention to the following:\color{black}
\begin{subequations}\label{eq:6}
\begin{align}
  &&&&&\bm S=\partial_{\bm F}\psi_{\textsc{me}}(\bm F,\bm m,\zeta,\theta),\qquad &&\mathbb S=\mathfrak H(\nabla^2\bm\chi),&&&&&&\label{eq:6a}\\
  &&&&&\bm c=\partial_{\bm m}\psi(\bm F,\bm m,\zeta\theta)+\tau_1\DT{\bm m},\qquad &&\bm C=\kappa_2\nabla\bm m,&&\label{eq:6b}\\
  &&&&&\varsigma=-\mu+\partial_\zeta\psi(\bm F,\bm m,\zeta,\theta)+\tau_2\DT\zeta,\qquad&&\bm\varsigma=\kappa_2\nabla\zeta,&&\label{eq:6c}\\
  &&&&&\bm j=-\bm M(\bm F,\bm m,\zeta,\theta)\nabla\mu
&&\hspace*{-6em}\BLUE\text{(=\,Fick's law)},\color{black}\\
  &&&&&\bm q=-\bm K(\bm F,\bm m,\zeta,\theta)\nabla\theta
&&\hspace*{-6em}\BLUE\text{(=\,Fourier's law)}\color{black}.\label{eq:6e}
\end{align}
\end{subequations}\BLUE
Substitution of \eqref{eq:6} into \eqref{eq:2} yields the following inequality
\begin{equation}\label{eq:8}
  \tau_1|\DT{\bm m}|^2+\tau_2|\DT\zeta|^2+\bm M(\bm F,\bm m,\zeta,\theta)\nabla\mu\cdot\nabla\mu+\bm K(\bm F,\bm m,\zeta,\theta)\nabla\theta\cdot\nabla\theta\ge 0.
\end{equation}
Now, consider an evolution process in which the independent variables $\bm F,\bm m,\zeta,\theta$, and $\nabla\theta$ attain an arbitrary value, and $\DT{\bm m}$, $\DT\zeta$, and $\nabla\theta$ vanish at a given point at a given time. Then, for that process, the inequality \eqref{eq:8} reduces to $\bm K(\bm F,\bm m,\zeta,\theta)\nabla\theta\cdot\nabla\theta\ge 0$, and the arbitrariness of the independent variables on the left-hand side of the inequality entails that the 
tensor $\bm K(\bm F,\bm m,\zeta,\theta)$ must be 
\BLUE positive semidefinite \color{black}
for every choice of the quadruplet $(\bm F,\bm m,\zeta,\theta)$. A similar argument shows that $\bm M$ must be 
\BLUE positive semidefinite, \color{black} as well. Similarly, we can 
argue that $\tau_1$ and $\tau_2$ must be non negative.

The aforementioned requirements on $\bm M$ and $\bm K$ are sufficient to 
guarantee thermodynamical compatibility of the model. Yet, they are 
\BLUE too generic
\color{black} 
to make the model amenable to 
mathematical analysis. In particular, the dependence on $\bm F$ can in 
principle set problems when passing to the limit in a proof of existence of 
weak solutions. On the other hand, as we shall show in the next section, it 
is still possible to handle the dependence on the deformation gradient for 
a quite encompassing class of conductivity and mobility tensors having the 
following form:\color{black}
 \begin{subequations}\label{M-K-push-back}
\begin{align}\nonumber
\bm M(\bm F,\bm m,\zeta,\theta)&=(\det\bm F)\bm F^{-\top}\bm{\mathsf M}(\bm m,\zeta,\theta)\bm F^{-1}
=(\Cof\bm F)^\top\bm{\mathsf M}(\bm m,\zeta,\theta)\bm F^{-1}
\\&=\frac{(\Cof\bm F)^\top\bm{\mathsf M}(\bm m,\zeta,\theta)\Cof\bm F}{\det\bm F}
\qquad\text{ and similarly also}
\\\bm K(\bm F,\bm m,\zeta,\theta)
&=\frac{(\Cof\bm F)^\top\bm{\mathsf K}(\bm m,\zeta,\theta)\Cof\bm F}{\det\bm F}.
\end{align}\end{subequations}
\BLUE{}where $\bm{\mathsf M}=\bm{\mathsf M}(\bm m,\zeta,\theta)$ and 
$\bm{\mathsf K}=\bm{\mathsf K}(\bm m,\zeta,\theta)$ are interpreted as the 
phenomenological spatial mobility and spatial conductivity tensors. One can 
check that the aforementioned positivity requirements on $\bm M$ and $\bm K$ 
translate into the same requirements for ${\bm{\mathsf M}}$ and 
${\bm{\mathsf K}}$. We remark that \eqref{M-K-push-back}, which \color{black} 
is the usual transformation of 2nd-order covariant tensors, is reasonable 
rather for the isotropic case (
cf.\ e.g.\ \cite[Formula (67)]{DuSoFi10TSMF} 
in the case of mass transport).

We next consider the issue of prescribing boundary fluxes. We consider the boundary $\varGamma$ that separates the body from its environment. We denote by  $\bm\theta_{\rm e}$ and $\mu_{\rm e}$ the temperature and the chemical potential of the environment, respectively. We denote by $\jump{\bm q}=\bm q_{\rm e}-\bm q_{\rm i}$ the jump of the jump of the heat flux at the boundary, 
\BLUE i.e.\ \color{black} 
the difference between 
the heat flux $\bm q_{\rm e}$ outside the body and the trace \BLUE
on $\varGamma$ \color{black} 
of the heat 
flux $\bm q_{\rm i}$ inside the body $\Omega$. In a similar fashion, we define the jump $\jump{\bm j}=\bm j_{\rm e}-\bm j_{\rm i}$ of the mass flux at the boundary.

First, if no diffusant is trapped on the surface, mass conservations dictates that $\jump{\bm j}=\bm 0$. Second, if no energy can be stored at the boundary of the body, energy balance dictates that $\jump{\mb q}\cdot\bm n=\jump{\mu}\mb j\cdot\bm n$, where $\jump{\mu}=\mu_{\rm e}-\mu_{\rm i}$, namely, the difference between $\mu_{\rm e}$ the trace on $\varGamma$ of the chemical potential field of the environment, and $\mu_{\rm i}$, the trace of the chemical potential field within the body $\varOmega$. Finally, if there is no entropy production localized at the boundary, the entropy inequality takes the form:  $\jump{\theta^{-1}\bm q}\cdot\bm n=0$. 

Combining these conditions we get the following thermodynamical compatibility condition relating the outwards heat flux $\bm q\cdot\bm n$ and the flux of diffusant $\bm j\cdot\bm n$:
\begin{equation}\label{eq:4}
  \left(\frac 1 {\theta_{\rm e}}-\frac 1 \theta\right)\bm q\cdot\bm n+\frac{\mu_{\rm e}-\mu}{\theta_{\rm e}}\bm j\cdot\bm n\ge 0,
\end{equation}
where $\theta_{\rm e}$ and $\mu_{\rm e}$ are the temperature and the chemical potential of the environment. We select the following constitutive prescription 
for the boundary fluxes: 
\begin{equation}\label{eq:3}
 \bm q\cdot\bm n={K}(\theta-\theta_{\rm e})
\qquad\text{and}\qquad \bm j\cdot\bm n={M}(\mu-\mu_{\rm e}).
\end{equation}
\BLUE{}On substituting the constitutive equations \eqref{eq:6} into the 
balance equations \eqref{eq:36}, \eqref{eq:7}, and \eqref{eq:9}
we obtain the \color{black}following system of semilinear hyperbolic/parabolic 
\BLUE integro-differential \color{black} equations on $Q$:
\begin{subequations}\label{system}
  \begin{align}
&\varrho
\DDT{\bm \chi}=\operatorname{div}(\bm S-
\operatorname{div}\mathfrak H(\nabla^2\bm \chi))+((\gradop\bm{\mathsf h}_{\rm e}){\circ}\bm\chi)^\top\nabla\bm\chi\,\bm m+\bm f\nonumber
\\[-.1em]&\qquad\qquad\qquad\qquad\qquad\qquad\text{ where}\quad 
\bm S=\partial_{\bm F}\psi_{\textsc{me}}(\nabla\bm \chi,\bm m,\zeta)-(\bm{\mathsf h}_{\rm e}{\circ}\bm\chi)\otimes\bm m,
\label{system.a}
\\
&\tau_1\DT{\bm m}=\kappa_1{\Delta{\bm m}}-\partial_{\bm m}\psi(\nabla\bm \chi,\bm m,\zeta,\theta)+(\nabla\bm{\chi})^\top \bm{\mathsf h}_{\rm e}{\circ}\bm{\chi},\label{system.b}\\
&\DT\zeta-\operatorname{div}(\bm M(\nabla\bm\chi,\bm m,\zeta,\theta)\nabla\mu)=0\ \ \text{ with }\ \mu=\partial_\zeta\psi(\nabla\bm \chi,\bm m,\zeta,\theta)
+\tau_2\DT\zeta-\kappa_2\Delta\zeta,
\label{system.c}
\\
&\nonumber\DT w-\operatorname{div}(\bm K(\nabla\bm\chi,\bm m,\zeta,\theta)\nabla\theta)=\tau_1|\DT{\bm m}|^2+\tau_2\DT\zeta^2\nonumber
\\[-.3em]&\qquad\ \ +\bm M(\nabla\bm\chi,\bm m,\zeta,\theta)\nabla\mu{\cdot}\nabla\mu
+\partial_{\bm m}\psi_{\textsc{th}}(\bm m,\zeta,\theta){\cdot}\DT{\bm m}
+\partial_{\zeta}\psi_{\textsc{th}}(\bm m,\zeta,\theta)\DT\zeta\label{system.d}
\\[-.2em]&\text{with}\quad w=e_{\textsc{th}}(\bm m,\zeta,\theta),\label{system.e}
  \end{align}
\end{subequations}
where \BLUE $e_{\textsc{th}}$, the thermal part of the internal energy, 
has been defined in \eqref{eq:10}. \color{black}
This system is accompanied with some
boundary conditions. For convenience of exposition, we here limit ourselves to the following natural boundary conditions on ${\varSigma}$
\BLUE which accompany sucessively the equations (\ref{system}a-d)\color{black}:
\begin{subequations}\label{BC}
  \begin{align}\label{BC1}
&\bm S\bm n
-\divS\mathfrak{H}(\nabla^2\bm\chi)=\bm g\ \ \text{ and }\ \ \mathfrak{H}(\nabla^2\bm\chi){:}(\bm{n}\otimes\bm{n})=0,
\\&\kappa_1\nabla\bm m\cdot\bm n=0,
\\&\bm M(\nabla\bm\chi,\bm m,\zeta,\theta)\nabla\mu\cdot\bm n+{M}\mu={M}\mu_{\rm e}\ \ \text{ and }\ \ 
\kappa_2\nabla\zeta\cdot\bm n=0,
\\&\label{BC4}
\bm K(\nabla\bm\chi,\bm m,\zeta,\theta)\nabla\theta\cdot\bm n+
\BLUE{}K\theta=K\theta_{\rm e}\,,\color{black}
\end{align}
\end{subequations}
where $\bm g$ is the traction force, $\mu_{\rm e}$ is a chemical potential prescribed on the boundary and ${M}$ is a phenomenological coefficient 
\BLUE for \color{black} the flux of the diffusant through 
the boundary, and $h$ is the heat flux through the boundary. Moreover, 
``$\divS$'' in \eqref{BC1} denotes the surface divergence defined as
 $\divS={\rm tr}(\nablaS)$ with ${\rm tr}(\cdot)$ being the trace of a
$(d{-}1){\times}(d{-}1)$-matrix and 
$\nablaS v=\nabla v-\frac{\partial v}{\partial\bm{n}}\bm{n}$ 
denoting the surface gradient of $v$.
We will consider the initial-value problem for the system 
\eqref{system}--\eqref{BC}, prescribing the 
 initial conditions on the reference domain ${\varOmega}$:
\begin{subequations}\label{IC}\begin{align}\label{IC-1}
& \bm\chi|_{t=0}^{}=\bm\chi_0,\quad
\bm v|_{t=0}^{}=\bm v_0,\quad
\bm m|_{t=0}^{}=\bm m_0,\quad
\zeta|_{t=0}^{}=\zeta_0,\quad
\\&\label{IC-2}
\theta|_{t=0}^{}=\theta_0.
\end{align}\end{subequations}


As far as the magnetic part concerns, the model may be classified as rather 
macroscopical because we have neglected gyromagnetic effects in 
\eqref{system.b}. Mathematically, there would not be difficulties to handle a  
gyromagnetic term proportional to $\bm{\mathsf m}\times(\bm F\bm m)^{\bigdot}
\BLUE=\bm{\mathsf m}\times(\DT{\bm F}\bm m+\bm F\DT{\bm m})\color{black}$ 
which would have a good physical sense under displacements with small rotations but in general gyromagnetic effects interact with large deformations in a very nonlinear way. In fact, one would require a control on $\DT{\bm F}$, which is in fact not available due to the lack of mechanical viscosity. As already pointed out, mechanical viscosity as in \cite{LewicM2013local} or \cite{MielkeRoubicek2016viscolargestrains} would give the control of $\dot{\bm F}$ which would allow us to handle the gyromagnetic term.


Testing \eqref{system.a}--\eqref{system.d} with the corresponding boundary 
conditions respectively by $\DT{\bm\chi}$, $\DT{\bm m}$, $(\mu,\DT\zeta)$, 
and by a number $\alpha\in[0,1]$ reveals the \emph{energetics} of the model.
Thus we obtain the following identity:
\begin{align}
    &\int_{\varOmega} \Big(\alpha w(t)+\frac \varrho 2 |\DT{\bm\chi}(t)|^2
+\psi_{\textsc{me}}(\nabla\bm\chi(t),\bm m(t),\zeta(t))
+\frac{\kappa_1}2|\nabla\bm m(t)|^2+\frac{\kappa_2}2|\nabla\zeta(t)|^2
\Big)\dx
+\mathscr{H}(\nabla^2\bm\chi(t))\nonumber
\\\nonumber
 &\qquad\qquad\qquad\qquad\qquad
+(1{-}\alpha)\int_0^t\int_{\varOmega}(\tau_1|\DT{\bm m}|^2+\tau_2\DT\zeta^2
+\bm M(\nabla\bm\chi,\bm m,\zeta,\theta)\nabla\mu{\cdot}\nabla\mu
\\
 &\qquad\qquad\qquad\qquad\qquad\nonumber
+\partial_{\bm m}\psi_{\textsc{th}}(\bm m,\zeta,\theta)\cdot\DT{\bm m}
+\partial_{\zeta}\psi_{\textsc{th}}(\bm m,\zeta,\theta)\DT\zeta)\dx \dt\nonumber\\
&\nonumber =\int_{\varOmega} \Big(\alpha w_0+\frac\varrho2|{\bm v}_0|^2
+\psi_{\textsc{me}}(\nabla\bm\chi_0,\bm m_0,\zeta_0)
+\frac{\kappa_1}2|\nabla\bm m_0|^2
+\frac{\kappa_2}2|\nabla\zeta_0|^2
\Big)\dx+\mathscr{H}(\nabla^2\bm\chi_0)\\\nonumber
&\qquad\qquad\qquad\qquad\qquad
+\int_0^t\int_{\varOmega} \Big(\big(((\gradop\bm{\mathsf h}_{\rm e}){\circ}\bm\chi)^\top\nabla\bm\chi\,\bm m+\bm f\big)\cdot\DT{\bm\chi}
+(\bm{\mathsf h}_{\rm e}{\circ}\bm\chi)\otimes\bm m:\nabla\DT{\bm\chi}
\\
&\qquad\qquad\qquad\qquad\qquad+(\nabla\bm{\chi})^\top\bm{\mathsf h}_{\rm e}{\circ}\bm{\chi}\cdot\DT{\bm m}\Big)\dx\dt
+\int_0^t\int_{\varGamma} (\alpha K(\theta_{\rm e}-\theta)+\bm g\cdot\DT{\bm \chi}+
 M(\mu_{\rm e}-\mu)^2\color{black}\d S\dt.
\label{formal-energy}
\end{align}
For $\alpha=0$, the above identity is the \emph{mechano-magneto-chemical balance} while, for $\alpha=1$, this identity is the \emph{total energy balance}. 
We reaffirm again that these are only formal estimates because 
$\nabla\DT{\bm\chi}$ is not (and, within our model, will not be) well defined. 
In particular, the bulk term 
$(\bm{\mathsf h}_{\rm e}{\circ}\bm\chi)\otimes\bm m:\nabla\DT{\bm\chi}$ is not 
well defined (unless some additional regularity of the particular solutions
were proved) as well as the boundary term $\bm g\cdot\DT{\bm \chi}$.
Later, an integration by parts in time will be in order to cope with these 
terms, cf.\ \eq{est-of-y-m-z}. 

\BLUE{}
\begin{remark}
{\rm 
 \BLUE{}While the derivation of (\ref{system}.a-d) using the 
constitutive equations \eqref{eq:6} is immediate, the derivation of 
\eqref{system.d} from the balance of energy  \eqref{eq:9}
requires some 
intermediate steps. First, we substitute the assumption \eqref{eq:36} into 
the definition of $e_{\textsc{tot}}$ in \eqref{eq:37} and we use 
\BLUE the first equation in \color{black} \eqref{eq:9}
to write $\operatorname{div}(\mu\bm j)=\nabla \mu\cdot\bm j-\mu\DT\zeta$. Then the balance of energy can be rewritten as follows:
\begin{equation}\label{eq:38}
  \DT e_{\textsc{th}}+\operatorname{div}\bm q=-\bm j\cdot\nabla\zeta+\pi_{\textsc{int}}-\DT\psi_{\textsc{me}}+\mu\DT\zeta-{\kappa_1}\nabla{\bm m}:\nabla\DT{\bm m}-{\kappa_2} \nabla\zeta\cdot\nabla\DT\zeta-\mathfrak H(\nabla^2{\bm\chi})\Vdots \nabla^2\DT{\bm\chi}.
\end{equation}
The remaining steps from \eqref{eq:38} to \eqref{system.d} are the following. First,  we notice that, as a consequence of the constitutive equation \eqref{eq:6e}, the left-hand sides of \eqref{eq:38} and \eqref{system.d} coincide; then, we notice that, as a consequence of the remaining the constitutive equations (\ref{eq:6}.a-c), the internal power has the representation
\[
\pi_{\textsc{int}}=\DT\psi_{\textsc{me}}+\partial_{\bm m}\psi_{\textsc{th}}(\bm m,\zeta,\theta)\cdot\DT{\bm m}+(\partial_{\zeta}\psi_{\textsc{th}}(\bm m,\zeta,\theta)-\mu)\DT\zeta+{\kappa_1}\nabla{\bm m}:\nabla\DT{\bm m}+{\kappa_2} \nabla\zeta\cdot\nabla\DT\zeta+\mathfrak H(\nabla^2{\bm\chi})\Vdots \nabla^2\DT{\bm\chi}+\tau_1|\DT{\bm m}|^2+\tau_2|\DT\zeta|^2
\]
and, as a result, the right-hand sides of \eqref{eq:38} and \eqref{system.d} coincide.
}
\end{remark}

\begin{remark}
{\rm Equation \eqref{system.d}, being a consequence of the energy-balance equation, should be regarded as a generalization of the heat-conduction equation. It is possible to show that this equation admits the following form:
\begin{align}\nonumber
&c_{\rm v}(\bm m,\zeta,\theta)\DT\theta-\operatorname{div}(\bm K(\nabla\bm\chi,\bm m,\zeta,\theta)\nabla\theta)=\tau_1|\DT{\bm m}|^2+\tau_2\DT\zeta^2
+\bm M(\nabla\bm\chi,\bm m,\zeta,\theta)\nabla\mu{\cdot}\nabla\mu
\\[-.2em]&\qquad\qquad\qquad\qquad\qquad\qquad\qquad\qquad\ \ \ \ 
+\theta\partial^2_{\theta\bm m}\psi_{\textsc{th}}(\bm m,\zeta,\theta){\cdot}\DT{\bm m}
+\theta\partial^2_{\theta\zeta}\psi_{\textsc{th}}(\bm m,\zeta,\theta)\DT\zeta,
\label{system.d+}
\\[-.0em]&\nonumber\text{with the heat capacity}\ \ 
c_{\rm v}(\bm m,\zeta,\theta)=-\theta\partial^2_{\theta\theta}\psi_{\textsc{th}}(\bm m,\zeta,\theta).
\end{align}
However, for relying on the enthalpy the form \eq{system.d} 
is more convenient for the analytical treatment of the heat transfer. 
}
\end{remark}\color{black}

\BLUE{}\begin{remark}
{\rm 
  Within the class of models filtered by the constitutive assumptions \eqref{eq:6}, the ultimate thermodynamical requirements are those entailed by the inequality \eqref{eq:8}. Those structural requirements are still to general for us to carry out mathematical analysis, and some specialization is required.  The fact that we are able to handle analytially the form  (\ref{M-K-push-back}) is one of the important results of our paper.
}
\end{remark}\color{black}

\begin{remark}\label{rem:f}
{\rm
One may want to consider the referential body force $\bm f(x,t)$ as the pull back $(\det\bm F)\bm{\mathsf f}\circ\bm\chi$ of a time-dependent spatial force density $\bm{\mathsf f}(z,t)$. For example, if $\bm{\mathsf a}(z,t)$ is an acceleration field such as gravity, the referential force would be $\bm f=\varrho\,\bm{\mathsf a}\circ\bm\chi$. Consistency with such a choice would require, however, a similar prescription for the surface-force field $\bm g$ as the pull back $|(\operatorname{det}\bm F)^{-1}\bm F^\top\bm n|\bm{\mathsf f}$ of a surface force density $\bm{\mathsf f}$ on $\partial\bm\chi(\varOmega)$, a prescription that unfortunately does not fit within our framework. The motivation will be apparent in Remark \ref{rem:g} in the next section.
}
\end{remark}

\begin{remark}
{\rm 
In \cite[\S 1.2.2]{James2002CMaT} the evolution equation governing the deformation are derived by carrying over to dynamics the equations of mechanical equilibrium in the static case. The latter equations are obtained by writing the stationarity of the Gibbs energy $G(\bm\chi,\bm m,\zeta,\theta)$ defined in \eqref{eq:5} with respect to the deformation $\bm\chi$. When carrying out this procedure, the external fields (both the loadings and the applied fields) are considered as fixed. 
Accordingly, one finds
\begin{align}
  {\rm D}_{\bm\chi} G(\bm\chi,\bm m,\zeta,\theta)[\widetilde{\bm\chi}]&=\int_\Omega \Big((\partial_{\bm F}\psi_{\textsc{me}}(\nabla\bm \chi,\bm m,\zeta)-(\bm{\mathsf h}_{\rm e}{\circ}\bm\chi)\otimes\bm m):\nabla\delta\bm \chi+\mathfrak H(\nabla^2\bm\chi)\Vdots\nabla^2\widetilde{\bm\chi}\Big)\dx\nonumber\\
&\phantom{=}-\int_\Omega\big(\bm f+((\gradop\bm{\mathsf h}_{\rm e}){\circ}\bm\chi)^\top\nabla\bm\chi\,\bm m\big)\cdot\widetilde{\bm \chi} \dx-\int_\varGamma\bm g\cdot\widetilde{\bm\chi}\,\d S
\end{align}
\BLUE{}where ${\rm D}_{\bm\chi} G(\cdot)[\widetilde{\bm\chi}]$ denotes 
the directional derivarive of $G$ in the direction $\widetilde{\bm\chi}$. 
\color{black} 
Imposing stationarity with respect to all tests $\widetilde{\bm\chi}$ yields again equation \eqref{system.a} and the boundary conditions \eqref{BC1}.}
\end{remark}



\begin{remark}{\rm The viscous-like dissipative term $\tau_2\DT\zeta$ in 
\eqref{system.c} is needed to cope with the direct coupling of $\zeta$ with 
$\theta$
\BLUE through the term
$\psi_{\textsc{th}}(\bm m,\zeta,\theta)$ in \eqref{ansatz2}. More specifically, 
note that the heat equation \eqref{system.d} contains the adiabatic term
$\partial_{\zeta}\psi_{\textsc{th}}(\bm m,\zeta,\theta)\DT\zeta$ which ultimately
needs to have $\DT\zeta$ under control\color{black}. In doing 
this, we \BLUE involved the mentioned term $\tau_2\DT\zeta$, 
\color{black}
following the original Gurtin's ideas \cite{Gurt96GGLC},
cf.\ also e.g.\ \cite{BonetCT2015non,EllGar96CHED,miranville1999model,Ross05TCGV,Tomas2015Some}. 
}\end{remark}

\BLUE{}
\section{Analysis of the evolution system \eq{system}: 
         existence of weak solutions}
\color{black}\label{sec-Galerkin}

We consider the time interval $I=[0,T]$ with $T$ a fixed  time horizon 
considered for the evolution, and we denote by $L^p(I;X)$ the standard 
Bochner space of Bochner-measurable mappings $I\to X$ with $X$ a 
Banach space. 
Also, $W^{k,p}(I;X)$ denotes the Banach space of mappings from $L^p(I;X)$ 
whose $k$-th distributional derivative in time is also in $L^p(I;X)$.


\begin{definition}[Weak solution]\label{def}
We call the \BLUE{}five\color{black}-tuple 
$(\bm\chi,\bm m,\zeta,\mu,\theta)$ 
with $\bm\chi\in L^2(I;H^2(\Omega;\R^d))$, $\bm m\in L^2(I;H^1(\Omega;\R^d))$,
$\zeta,\mu\in L^2(I;H^1(\Omega))$, and $\theta\in L^1(I;W^{1,1}(\Omega))$
\BLUE such that $e_{\textsc{th}}(\bm m,\zeta,\theta)\color{black}
\in L^1(Q)$ \color{black} a weak solution to the initial-boundary-value
problem \eq{system}--\eq{BC}--\eq{IC} if $\mathscr{H}(\nabla^2\bm\chi)\in 
L^\infty(I)$ and 
\begin{subequations}\begin{align}\nonumber
&\int_Q\varrho{\bm\chi}{\cdot}\DDT{\bm v}+\bm S{:}\nabla\bm v
+\mathfrak{H}(\nabla^2\bm\chi)\Vdots\nabla^2\bm v
-((\gradop\bm{\mathsf h}_{\rm e}){\circ}\bm\chi)^\top\nabla\bm\chi\,\bm m
{\cdot}{\bm v}\,\d x\d t
\\[-.4em]&\qquad\qquad\qquad\qquad\qquad=\int_Q\bm f{\cdot}\bm v\,\d x\d t
+\int_{\varSigma}\bm g{\cdot}\bm v\,\d S\d t
+\int_{\varOmega}\varrho\bm v_0{\cdot}\bm v(0)-\varrho\bm\chi_0{\cdot}\DT{\bm v}(0)\,\d x
\end{align}
holds for $\bm v$ smooth with $\bm v(T)=\DT{\bm v}(T)=0$ and 
with $\bm S=\partial_{\bm F}\psi_{\textsc{me}}(\nabla\bm \chi,\bm m,\zeta)-(\bm{\mathsf h}_{\rm e}{\circ}\bm\chi)\otimes\bm m$,
\begin{align}
&\int_Q\kappa_1\nabla\bm m{:}\nabla\bm v
+\partial_{\bm m}\psi(\nabla\bm \chi,\bm m,\zeta,\theta)
-(\nabla\bm{\chi})^\top \bm{\mathsf h}_{\rm e}{\circ}\bm{\chi}{\cdot}\bm v
-\tau_1\bm m{\cdot}\DT{\bm v}\,\d x\d t
=\int_{\varOmega} \bm m_0{\cdot}\bm v(0)\,\d x
\end{align}
holds for $\bm v$ smooth with $\bm v(T)=0$,
\begin{align}
&\int_Q\bm M(\nabla\bm\chi,\bm m,\zeta,\theta)\nabla\mu{\cdot}\nabla v
-\zeta\DT v\,\d x\d t+\int_{\varSigma}M \mu v\,\d S\d t
=\int_{\varSigma}M\mu_{\rm e} v\,\d S\d t+\int_{\varOmega} \zeta_0 v(0)\,\d x
\end{align}
holds for $v$ smooth with $v(T)=0$,
\begin{align}
&\int_Q\kappa_2\nabla\zeta{\cdot}\nabla v
+\big(\partial_\zeta\psi(\nabla\bm \chi,\bm m,\zeta,\theta)
-\mu\big)v-\tau_2\zeta\DT v\,\d x\d t=\int_\Omega\zeta_0 v(0)\,d x
\end{align}
holds for $v$ smooth with $v(T)=0$, 
\begin{align}\nonumber
&\int_Q\bm K(\nabla\bm\chi,\bm m,\zeta,\theta)\nabla\theta{\cdot}\nabla v-w\DT v
\,\d x\d t
=\int_Q\Big(\tau_1|\DT{\bm m}|^2+\tau_2\DT\zeta^2
+\bm M(\nabla\bm\chi,\bm m,\zeta,\theta)\nabla\mu{\cdot}\nabla\mu
\\[-.3em]&\qquad\ \ 
+\partial_{\bm m}\psi_{\textsc{th}}(\bm m,\zeta,\theta){\cdot}\DT{\bm m}
+\partial_{\zeta}\psi_{\textsc{th}}(\bm m,\zeta,\theta)\DT\zeta\Big)v\,\d x\d t
+\int_{\varSigma} 
K(\theta_{\rm e}-\theta)
 v\,\d S\d t+\int_{\varOmega} w_0 v(0)\,\d x
\end{align}\end{subequations}
holds for $v$ smooth with $v(T)=0$, and
with $w=e_{\textsc{th}}(\bm m,\zeta,\theta)$ on $Q$ and 
$w_0=e_{\textsc{th}}(\bm m_0,\zeta_0,\theta_0)$ on ${\varOmega}$.
\end{definition}



Let us first summarize the assumptions we will impose, apart 
from \eqref{kernel} with $\gamma>d/2-1$ and \eqref{ass-1}, 
and we will use in what follows both to qualify the integrals used above 
in the definition of the weak solution and for proving existence of such 
solutions, although we do not claim that they \BLUE 
cannot be weakened with only a slightly more involved argumentation in 
the proof. \color{black}.
For some $\epsilon>0$ and some $C<\infty$, we assume: 
\begin{subequations}\label{ass}
\begin{align}\nonumber
&\exists\varphi:\R^{d\times d}{\times}\R^d{\times}\R\to\R\ \text{continuous\BLUE{}ly differentiable and\color{black}}
\\&\qquad \nonumber
\xi:\R{\times}\R^d{\times}\R\to\R\cup\{\infty\}
\ \text{\BLUE{}continuously differentiable on 
$\R^+{\times}\R^d{\times}\R$\color{black}},\ 
,\ \xi(\cdot,\bm m,\zeta)\ \text{convex}:\ \ 
\\&\label{ansatz-psi}
\qquad\psi_{\textsc{me}}(\bm F,\bm m,\zeta)=\varphi(\bm F,\bm m,\zeta)+\xi(\det\bm F,\bm m,\zeta), 
\\&\label{ass-coercivity}
\qquad\epsilon(|\bm F|^{p_1}{+}|\bm m|^{p_2}{+}|\zeta|^{p_3})
\le\varphi(\bm F,\bm m,\zeta)\le 
C(1{+}|\bm F|^{p_4}{+}|\bm m|^{p_2}{+}|\zeta|^{p_3}),\ \ \ 
p_1,p_2>2,\ \ p_3\ge 2,\ \ p_4\ge p_1,
\\&\label{ass-coercivity+}
\qquad\xi(z,\bm m,\zeta)\begin{cases}\ge\epsilon/z^q&\text{if }z>0,\\=\infty&\text{if }z\le0,\end{cases},\ \ 
\ q>\frac{2d}{\BLUE2\gamma{+}2{-}d\color{black}
},
\\&\label{ass-M-K}\bm{\mathsf M},\bm{\mathsf K}:\R^d\times\R^2\to\R_{\rm sym}^{d\times d}\
\text{ continuous, bounded, and uniformly positive definite},
\\&\label{ass-psi/mz}
\BLUE\psi_{\textsc{th}}:\R^d\times\R^2\to\R
\text{ twice continuously differentiable},\ \ \ \ \color{black}
|\partial_{\bm m}\psi_{\textsc{th}}|\le C,
\quad|\partial_{\zeta}\psi_{\textsc{th}}|\le C,
\qquad
\\\label{ass-psith}
&\qquad|\partial_{{\bm m}\theta}^{\BLUE2\color{black}}\psi_{\textsc{th}}|\le C/(1{+}|\theta|),
\qquad|\partial_{\zeta\theta}^{\BLUE2\color{black}}\psi_{\textsc{th}}|\le C/(1{+}|\theta|),
\\\label{ass-cv}
&\BLUE{}c_{\rm v}:\R^{d\times d}{\times}\R^d{\times}\R\to\R\ 
\text{ continuously differentiable},\ \ \color{black}
\epsilon\le c_{\rm v}(\bm m,\zeta,\theta)\le C,
\\\label{ass-cv/mz}
&\qquad|\partial_{\bm m}c_{\rm v}|\le C/(1{+}|\theta|)^{1+\epsilon},
\qquad|\partial_\zeta c_{\rm v}|\le C/(1{+}|\theta|)^{1+\epsilon},
\\&\label{ass-h-e}
\bm f\in L^1(Q;\R^d),
\ \ \bm g\in W^{1,1}(I;W^{1,1}({\varOmega};\R^d)),\ \
\bm{\mathsf h}_{\rm e}\in W^{1,1}(I;W^{1,q}(\R^d;\R^d)),\ \ q>\frac{p_1p_2}{p_1p_2{-}p_1{-}p_2},
\\[-.3em]&
\mu_{\rm e}\in L^2({\varSigma}),\ \ \theta_{\rm e}\in L^1({\varSigma}),\ \ \theta_{\rm e}>0,
\\&\label{IC-qualif}
\bm\chi_0\!\in\! H^{2+\gamma}({\varOmega};\R^d),
\ \ \ \bm v_0\!\in\! L^2({\varOmega};\R^d),
\ \ \ \bm m_0\!\in\! H^1({\varOmega};\R^d),
\ \ \ \zeta_0\!\in\! H^1({\varOmega}),
\ \ \ \theta_0\!\in\! L^1({\varOmega}),\ \ \ \theta_0\ge0\ \text{ on }\ {\varOmega},
\\\label{ass-IC-det}
&\xi(\det\nabla\bm\chi_0,\bm m_0,\zeta_0)\le C.
\end{align}
In addition, we shall need
\begin{align}
  \label{ass-xi}
  &\BLUE\partial_{z\bm m}^2\color{black}\xi=0,
\qquad\BLUE\partial_{z\zeta}^2\color{black}\xi=0,
\qquad
\partial_{\bm m}\psi(\bm F,\bm m,\zeta,\theta)\le C,\qquad \partial_{\zeta}\psi(\bm F,\bm m,\zeta,\theta)\le C.
\end{align}
\end{subequations}

We should note that in the above assumptions, the variable $\theta$ 
\BLUE may range \color{black} also over negative values because the 
nonnegativity of temperature is granted only in the resulted continuous model 
but not in our regularized approximate scheme. \BLUE(Let us
point out that we do not particularly prevent possible negativity
of $\zeta$, which would need another technicalities 
by making $\bm M$ degenerate or by admitting 
a blow-up $\xi\to+\infty$ for $\zeta\to0+$.) \color{black}
%
The assumption \eqref{ass-psi/mz} is cast so that $\theta$ does not 
influence a-priori bounds in mechano-magneto-chemo part.

Our main analytical result, proved 
by rather constructive way when merging Lemma~\ref{lem-1} and 
Propositions~\ref{prop1}--\ref{prop2} below, is:

\begin{theorem}[Existence of weak solutions]\label{theorem}
Let \eqref{kernel} with $\gamma>d/2-1$, \eqref{ass-1},  \eqref{ansatz}, 
and \eqref{ass} hold.
Then there exists a weak solution $(\bm\chi,\bm m,\zeta,\mu,\theta)$
according Definition~\ref{def} and, moreover, 
$\bm\chi\in W^{1,\infty}(I;L^2(\Omega;\R^d))\cap L^{\infty}(I;H^{2+\gamma}(\Omega;\R^d))$ with 
$\min_Q\det(\nabla\bm\chi)>0$,
$\bm m\in L^\infty(I;H^1(\Omega;\R^d)\cap H^1(I;L^2(\Omega;\R^d))$,
$\zeta\in L^\infty(I;H^1(\Omega))\cap H^1(I;L^2(\Omega))$,
$\Delta\bm m\in L^2(Q;\R^d)$ 
and $\Delta\zeta\in L^2(Q)$. \BLUE Moreover, $\theta\ge0$. \color{black}
\end{theorem}




\BLUE
Note that Theorem~\ref{theorem} did not say anything about uniqueness. Indeed, 
this attribute seems to be very delicate 
in particular due to quadratic-like coupling terms 
in the right-hand side of \eq{system}, namely
$((\gradop\bm{\mathsf h}_{\rm e}){\circ}\bm\chi)^\top\nabla\bm\chi\,\bm m$,
$\tau_1|\DT{\bm m}|^2$, $\tau_2\DT\zeta^2$, and 
$\bm M(\nabla\bm\chi,\bm m,\zeta,\theta)\nabla\mu{\cdot}\nabla\mu$.
Thus uniqueness could be expected, under some additional assumptions, at 
most for sufficiently small data. We do not address this issue here,
however. \color{black}

We will prove Theorem~\ref{theorem} by rather constructive method 
making a 
regularization of \eq{system}--\eq{BC}--\eq{IC}
and then applying the Galerkin approximation, proving 
apriori estimates (that can be interpreted as a numerical stability)
and convergence (in terms of subsequences) towards weak solutions.


As we need to control the determinant of the deformation gradient,
we cannot impose semi-convexity assumption and cannot rely on 
a time discretisation. Therefore, we will use the Faedo-Galerkin 
combined with a regularization \BLUE of the 
heat sources to facilitate usage of standard $L^2$-parabolic theory at least
on the Galerkin level\color{black}. 

The successive estimation and successive limit passage must be executed 
\BLUE{}in such a way that the natural $L^1$-heat source and the 
corresponding heat-equation theory is used only on the continuous
level when non-negativity of temperature can be granted\color{black}.


Our assumptions on the thermal coupling lead to a relatively simple
scenario that allows for a-priori estimates of (in particular, we use 
$\bm\chi$ and $\bm m$ and $\zeta$ and $\mu$
independent of the discretization and regularization of the 
heat transfer equation. 

Using 
the parameter $\sigmanu>0$, 
we regularize 
also the right-hand side of the heat equation (both in the bulk and in the
boundary condition) in order to avoid the superlinear growth 
in the dissipation rates and the adiabatic terms too.
We thus arrive to the system \BLUE 
(\ref{system}a-c,e)--(\ref{BC}a-c)--\eq{IC-1} while 
\eq{system.d} is replaced by the regularized heat-transfer equation
\begin{align}
&\nonumber
\DT w-\operatorname{Div}(\bm K
(\nabla\bm\chi,\bm m,\zeta,\theta)\nabla\theta)=
\frac{\tau_1|\DT{\bm m}|^2}{1{+}\sigmanu|\DT{\bm m}|^2}
+\frac{\tau_2\DT\zeta^2}{1{+}\sigmanu\DT\zeta^2}
+\frac{\bm M
(\nabla\bm\chi,\bm m,\zeta,\theta)\nabla\mu{\cdot}\nabla\mu}{1+\sigmanu|\nabla\mu|^2}
\nonumber
\\[-.3em]&\hspace{22em}+\partial_{\bm m}\psi_{\textsc{th}}(\bm m,\zeta,\theta){\cdot}\DT{\bm m}+\partial_{\zeta}\psi_{\textsc{th}}(\bm m,\zeta,\theta)\DT\zeta
\label{system.d-reg}
\end{align}
with the boundary/initial conditions \eq{BC4}/\eq{IC-2} regularized as
\begin{subequations}\label{BC-IC-reg}\begin{align}
&&&\bm K
(\nabla\bm\chi,\bm m,\zeta,\theta)\nabla\theta\cdot\bm n
\BLUE+K\theta=K\frac{\theta_{\rm e}}{1+\sigmanu\theta_{\rm e}}\color{black}
&&\text{on } \Sigma,
\label{BC4-reg}
\\&&&\theta|_{t=0}^{}=\theta_{0,\sigmanu}:=\frac{\theta_0}{1+\sigmanu\theta_0}
&&\text{on }\Omega.&&&&
\label{IC-reg}
\end{align}
\end{subequations}
\color{black}
Let us explain that we did not regularize the last 
\BLUE{}two \color{black} terms in
\eq{system.d-reg} in order to keep the possibility to switch between
the internal-energy formulation and the temperature formulation like
in the original problem, cf.\ \eq{system.d} vs.\ \eq{system.d+}.

\BLUE{}Then\color{black}, without going into (standard) technical 
details, we make 
a Faedo-Galerkin approximation 
by exploiting some finite-dimensional subspaces of 
$H^{2+\gamma}({\varOmega};\R^d)$ for 
\eq{system.a} 
and of $H^1({\varOmega})$ for each 
of the equations 
(\ref{system}b-e). Even, it is important 
to have the same sequence of finite-dimensional
spaces used for 
\BLUE both equations in \color{black} \eq{system.c} 
to facilitate their cross-testing and thus to allow for a cancellation of 
the $\pm\DT\zeta\mu$-terms also on the Galerkin level, and also it is important 
to allow for a good sense of $\Delta\bm m$ and $\Delta\zeta$ on the
Galerkin level so that, in fact, we need rather 
finite-dimensional subspaces of $H^2({\varOmega};\R^d)$ and $H^2({\varOmega})$
for 
(\ref{system}b,c).
We denote by 
$k\in\mathbb N$ the discretisation
index of this approximation.
%

Let us denote such an approximate solution, i.e.\ Galerkin approximation 
of the \BLUE above specified \color{black}
regularized problem, 
by 
$(\bm\chi_{\sigmanu k},\bm m_{\sigmanu k},\zeta_{\sigmanu k},\mu_{\sigmanu k},\theta_{\sigmanu k})$.

Without any loss of generality if 
the qualification \eq{IC-qualif} and density of 
the finite-dimensional subspaces is assumed, we can also assume that
that all the (nested) finite-dimensional spaces 
used for the Galerkin approximation 
contain both $\bm\chi_0$ and $\bm v_0$, and that 
the finite-dimensional spaces used for the approximation of 
\eq{system.b} contain $\bm m_0$, while those used for 
\eq{system.c} contain $\zeta_0$ and those used for \eq{system.d-reg}
contain $\theta_{0,\sigmanu}$ from \eq{IC-reg}. 
We also introduce a seminorm  $|\cdot|_l$ on $L^2(I;H^1({\varOmega})^*)$
defined by
\begin{align}\label{seminorm}
|\xi|_l:=\sup_{\|v\|_{L^2(I;H^1({\varOmega}))}\le1,\ v(t)\in V_l,\ t\in I}\int_Q\xi v\,\d x\d t.
\end{align}
Equipped by the countable family of these seminorms $\{|\cdot|_l\}_{l\in\N}$,
the linear space $L^2(I;H^1({\varOmega})^*)$ becomes a metrizable locally 
convex space (i.e.\ a Fr\'echet space).

\begin{lemma}[Approximate solutions]\label{lem-1}
Let the assumptions of Theorem~\ref{theorem} hold
and the finite-dimensional spaces are qualified as above. Then, for each 
$k\in\N$, the Galerkin approximate solution 
$(\bm\chi_{\sigmanu k},\bm m_{\sigmanu k},\zeta_{\sigmanu k},\mu_{\sigmanu k},\theta_{\sigmanu k}
)$ 
to the regularized problem 
\BLUE 
(\ref{system}a-c,e)--(\ref{BC}a-c)--\eq{IC-1} with 
(\ref{system.d-reg})--(\ref{BC-IC-reg}) \color{black}
exists and satisfies the following a-priori estimates
with some constant $C$ dependent only on the data and and some $C_\sigmanu$
and $C_{\sigmanu\eta}$ dependent also on the regularizing parameters as indicated:
\begin{subequations}\label{est}
\begin{align}\label{est-y}
&\big\|\bm\chi_{\sigmanu k}\big\|_{L^\infty(I;H^{2+\gamma}({\varOmega};\R^d))
\,\cap\,W^{1,\infty}(I;L^2({\varOmega};\R^d))}^{}\le C,
\\&\label{est-m}
\big\|\bm m_{\sigmanu k}\big\|_{L^\infty(I;H^1({\varOmega};\R^d)\cap L^{p_2}({\varOmega};\R^d))\,\cap\,H^1(I;L^2({\varOmega};\R^d))}^{}\le C,
\\&\label{est-zeta}
\big\|\zeta_{\sigmanu k}\big\|_{
L^\infty(I;H^1({\varOmega})\cap L^{p_3}({\varOmega};\R^d))\,\cap\,H^1(I;L^2({\varOmega}))
}^{}\le C,
\\&\label{est-zeta+}
\bigg\|
\Ceenk\nabla\mu_{\sigmanu k}\bigg\|_{L^2(Q;\R^d)}^{}\le C,
\\\label{est-Delta-zeta}
&\big\|\Delta\zeta_{\sigmanu k}\big\|_{L^2(Q)}^{}\le C,
\\\label{est-Delta-m}
&\big\|\Delta\bm m_{\sigmanu k}\big\|_{L^2(Q;\R^d)}^{}\le C,
\\\label{est-nabla-theta}
&\BLUE
\bigg\|
\Ceenk\nabla\theta_{\sigmanu k}\bigg\|_{L^2(Q;\R^d)}^{}\le C,\color{black}
\\
\label{est-mu}
&
\big\|\mu_{\sigmanu k}\big\|_{L^2(I;H^1(\Omega))}^{}\le 
C,
\\
&\BLUE\big\|\theta_{\sigmanu k}\big\|_{L^2(I;H^1(\Omega))\,\cap\,L^\infty(I;L^2({\varOmega}))}^{}\color{black}\le 
C_{\sigmanu},
\label{est++}
\\&\label{est-w-grad}
\big\|\nabla w_{\sigmanu k}\|_{L^2(Q;\R^d)}\le C_{\sigmanu},
\qquad\text{ and }\qquad
\big|\DT w_{\sigmanu k}\big|_l
\le C_{\sigmanu}\qquad\text{ for all }\ k\ge l,\ l\in\N,
\end{align}
\end{subequations}
\BLUE{}with $w_{\sigmanu k}=e_{\textsc{th}}(\bm m_{\sigmanu k},\zeta_{\sigmanu k},\theta_{\sigmanu k}).$\color{black}
\end{lemma}



\begin{proof}
We split the proof in \BLUE 
four \color{black} steps.\medskip

\noindent\emph{Step 1 - construction of the Galerkin solution.}
The existence of the Galerkin solution can be argued by the 
successive prolongation argument, relying on the a-priori 
estimate obtained by means of
testing the particular equations successively by  
$\DT{\bm\chi}_{\sigmanu k}$, $\DT{\bm m}_{\sigmanu k}$, 
$\mu_{\sigmanu k}$, 
$\DT\zeta_{\sigmanu k}$, and $\theta_{\sigmanu k}$.
Note that all these tests are legitimate in the level of the 
Galerkin approximation provided 
the finite-dimensional spaces used in both equations in the Cahn-Hilliard
systems are the same.
First four tests give the discrete analog of the balance
of the mechano-magneto-chemical energy, i.e.\ \eq{formal-energy}
with $\alpha=0$ for one a current interval $[0,t]$ with $t\le T$. 
Actually, the Galerkin approximation of the viscous Cahn-Hilliard 
system 
\eq{system.c} leads, instead of an ordinary-differential
system as usual, to a differential-algebraic system for $\DT\zeta$ involving
the holonomic constraint  
\begin{align}\label{alg-constraint}
\tau_2\operatorname{Div}(
\bm M
(\nabla\bm\chi,\bm m,\zeta,\theta)\nabla\mu)
-\kappa_2\Delta\zeta+\partial_\zeta\psi
(\nabla\bm \chi,\bm m,\zeta,\theta)-\mu=0
\end{align}
to be understood in its Galerkin approximation. \BLUE More
in detail, \eq{alg-constraint} in the Galerkin approximation means 
the integral identity 
$\int_\Omega \tau_2\bm M(\nabla\bm\chi,\bm m,\zeta,\theta)\nabla\mu)
\cdot\nabla v+\mu v\,\d x
=-\int_\Omega\kappa_2\nabla\zeta\cdot\nabla v+\partial_\zeta\psi
(\nabla\bm \chi,\bm m,\zeta,\theta)v\,\d x$ for all $v$ from the 
corresponding finite-dimensional Galerkin space. This is uniquely 
solvable in $\mu$ in this finite-dimensional space, so that $\mu$ is a 
(nonlinear algebraic) function of $(\bm\chi,\bm m,\zeta,\theta)$
and can be eliminated when substituting into the Galerkin approximation
of the diffusion equation $\DT\zeta=\operatorname{Div}(
\bm M
(\nabla\bm\chi,\bm m,\zeta,\theta)\nabla\mu)$.
This reveals the structure of the so-called index-1 differential-algebraic 
system and the underlying ordinary-differential
system. The energy-type $L^\infty$-apriori estimates ensure existence
of its solution on the whole time interval $[0,T]$ by the usual prolongation 
arguments.


Moreover, it is here important that, due to the assumption (\ref{ass}a-c)
 with \eq{kernel},
we have at disposal the Healey-Kr\"omer theorem \cite{HeaKro09IWSS} in the 
modification of 
Lemma~\ref{lem:1}. More specifically, 
we can see that $\bm\chi_{\sigmanu k}$ is valued in a single 
(sufficiently large) level set of the functional
$\int_{\varOmega}\epsilon|\nabla\bm\chi|^{p_1}
+\epsilon/\det(\nabla\bm\chi)^q
+|\nabla^2\bm\chi|^p\,\d x$, cf.\ \eq{ass-coercivity+},
with some $p>d$ and $q\ge pd/(p{-}d)$;
here we used also the embedding 
$H^{2+\gamma}({\varOmega})\subset W^{2,p}({\varOmega})$ with 
$p<2d/(d-2\gamma)$ (if $\gamma\le d/2$) or $p=\infty$ (if $\gamma>d/2$)
and that the mentioned conditions $q\ge pd/(p{-}d)$ and $p>d$ are 
compatible provided $\gamma>d/2-1$, so that considering $p$ as large
as possible yields the condition on $q$ and $\gamma$ 
used in \eq{ass-coercivity+}, namely $q>2d/(d-2-2\gamma)$. 
Then we are eligible to use the Healey and Kr\"omer's results 
\cite{HeaKro09IWSS} and arguments 
as in \cite[Proof of Lemma 5.1]{MieRou??RIEF} to obtain some 
$\eta>0$ such that 
$\det(\nabla\bm\chi_{\sigmanu k})\ge\eta$ everywhere 
on $Q$. Therefore, by the mentioned successive-prolongation 
argument, $\det\nabla{\bm\chi}_{\sigmanu k}\ge\eta>0$ holds with $\eta$ as in 
\eq{eq:34+};
in particlar, the so-called 
Lavrentiev phenomenon is exluded for the Galerkin procedure. \color{black}

\medskip 

\allowdisplaybreaks
\noindent\emph{Step 2 - estimates \BLUE(\ref{est}a-g)\color{black}.} 
The estimates (\ref{est}a-d) and \eqref{est-mu} are consequence of the 
following partial energy balance, which is obtained by testing the Galerkin 
approximations of 
(\ref{system}a-c) 
by $\DT{\bm\chi}_{\sigmanu k}$, $\DT{\bm m}_{\sigmanu k}$, $\mu_{\sigmanu k}$, and $\DT\zeta_{\sigmanu k}$:
\begin{align}\nonumber
    &\int_{\varOmega} \Big(\frac \varrho 2 |\DT{\bm\chi}_{\sigmanu k}(t)|^2
+\psi_{\textsc{me}}(\nabla\bm\chi_{\sigmanu k}(t),\bm m_{\sigmanu k}(t),\zeta_{\sigmanu k}(t))
+\frac 12\kappa_1|\nabla\bm m_{\sigmanu k}(t)|^2
+\frac 12 \kappa_2|\nabla\zeta_{\sigmanu k}(t)|^2
\Big)\dx\nonumber
+\mathscr{H}(\nabla^2\bm\chi_{\sigmanu k}(t))
\\[-.4em]\nonumber&\qquad
+\int_0^t\bigg(\int_{\varOmega}\tau_1|\DT{\bm m}_{\sigmanu k}|^2+\tau_2\DT\zeta_{\sigmanu k}^2
+
\BLUE\Big|
\bm M^{1/2}(\nabla\bm\chi_{\sigmanu k},\bm m_{\sigmanu k},\zeta_{\sigmanu k},\theta_{\sigmanu k})
\Ceenk\nabla\mu_{\sigmanu k}\Big|^2\color{black}
\dx
+\int_{\varGamma}M\mu_{\sigmanu k}^2\,\d S\bigg)\dt\nonumber
\\
&\nonumber \qquad
=
\int_0^t\int_{\varGamma}(M\mu_{\rm e}\mu_{\sigmanu k}+{\bm g}\cdot\DT{\bm\chi}_{\sigmanu k})\d S\dt
+\int_0^t\int_{\varOmega}\Big(
(\bm{\mathsf h}_{\rm e}(t){\circ}\bm\chi_{\sigmanu k}(t))\otimes\bm m_{\sigmanu k}(t)\cdot\nabla\DT{\bm\chi}_{\sigmanu k}(t)
\\[-.3em]\nonumber&\qquad\qquad\qquad
+(((\gradop\bm{\mathsf h}_{\rm e}){\circ}\bm\chi_{\sigmanu k})^\top\nabla\bm\chi_{\sigmanu k}\,\bm m_{\sigmanu k}+\bm f)\cdot\DT{\bm\chi}_{\sigmanu k}
+(\nabla\bm{\chi}_{\sigmanu k})^\top \bm{\mathsf h}_{\rm e}{\circ}\bm{\chi}_{\sigmanu k}\cdot\DT{\bm m}_{\sigmanu k}
\\[-.0em]\nonumber&\qquad\qquad\qquad
-\partial_{\bm m}\psi_{\textsc{th}}(\bm m_{\sigmanu k},\zeta_{\sigmanu k},\theta_{\sigmanu k})\cdot\DT{\bm m}_{\sigmanu k}
-\partial_{\zeta}\psi_{\textsc{th}}(\bm m_{\sigmanu k},\zeta_{\sigmanu k},\theta_{\sigmanu k})\DT{\zeta}_{\sigmanu k}\Big)
\dx\dt
\\&\qquad\qquad\nonumber
+\int_{\varOmega}\!\Big(
\frac\varrho2|{\bm v}_0|^2+\psi_{\textsc{me}}(\nabla\bm\chi_0,\bm m_0,\zeta_0)
+\frac 12\kappa_1|\nabla\bm m_0|^2+\frac 12 \kappa_2|\nabla\zeta_0|^2
\Big)
\dx+\mathscr{H}(\nabla^2\bm\chi_0)
\\
&\nonumber \qquad
=\int_0^t\int_{\varGamma}(M\mu_{\rm e}\mu_{\sigmanu k}-\DT{\bm g}\cdot\bm \chi_{\sigmanu k})\,\d S\dt
+
\int_0^t\int_{\varOmega}\!\Big(
\bm f\cdot\DT{\bm\chi}_{\sigmanu k}
-(\nabla\bm{\chi}_{\sigmanu k})^\top \DT{\bm{\mathsf h}}_{\rm e}{\circ}\bm{\chi}_{\sigmanu k}\cdot{\bm m}_{\sigmanu k}
\\[-.0em]\nonumber&\qquad\qquad\qquad
-\partial_{\bm m}\psi_{\textsc{th}}(\bm m_{\sigmanu k},\zeta_{\sigmanu k},\theta_{\sigmanu k})\cdot\DT{\bm m}_{\sigmanu k}
-\partial_{\zeta}\psi_{\textsc{th}}(\bm m_{\sigmanu k},\zeta_{\sigmanu k},\theta_{\sigmanu k})\DT{\zeta}_{\sigmanu k}\Big)
\dx\dt
\\&\qquad\qquad\nonumber
+\int_{\varOmega}\!
\Big((\nabla\bm{\chi}_{\sigmanu k}(t))^\top\bm{\mathsf h}_{\rm e}(t){\circ}\bm{\chi}_{\sigmanu k}(t)\cdot{\bm m}_{\sigmanu k}(t)
+
\frac \varrho 2 |{\bm v}_0|^2+\psi_{\textsc{me}}(\nabla\bm\chi_0,\bm m_0,\zeta_0)\\[-.3em]&\qquad\qquad\qquad\nonumber
+\frac 12\kappa_1|\nabla\bm m_0|^2
+\frac 12 \kappa_2|\nabla\zeta_0|^2-
\BLUE(\nabla\bm{\chi}_{0})^\top\bm{\mathsf h}_{\rm e}(0){\circ}\bm{\chi}_{0}\cdot{\bm m}_{0}\color{black}
\,
\dx
\\&\qquad\qquad
+\int_{\varGamma}\!\bm g(t)\cdot\bm\chi_{\sigmanu k}(t)-\bm g(0)\cdot\bm\chi_0\,\d S+\mathscr{H}(\nabla^2\bm\chi_0).
\label{est-of-y-m-z}
\end{align}
In writing this estimate, for the last equality, we used the by-part 
integration of the Zeeman energy using the  chain rule
\begin{subequations}\label{by-part-for-magnetic}
\begin{align}\nonumber
&\frac{\d}{\d t}\int_\varOmega
(\nabla\bm{\chi}_{\sigmanu k})^\top\bm{\mathsf h}_{\rm e}{\circ}\bm{\chi}_{\sigmanu k}\cdot{\bm m}_{\sigmanu k}\,\d x
=\int_\varOmega\bm{\mathsf h}_{\rm e}{\circ}\bm{\chi}_{\sigmanu k}\otimes{\bm m}_{\sigmanu k}\nabla\DT{\bm{\chi}}_{\sigmanu k}
+((\gradop\bm{\mathsf h}_{\rm e}){\circ}\bm\chi_{\sigmanu k})^\top\nabla\bm\chi_{\sigmanu k}\,\bm m_{\sigmanu k}\cdot\DT{\bm\chi}_{\sigmanu k}
\\[-.4em]&\qquad\qquad\qquad\qquad\qquad\qquad\qquad\qquad
+(\nabla\bm{\chi}_{\sigmanu k})^\top \bm{\mathsf h}_{\rm e}{\circ}\bm{\chi}_{\sigmanu k}\cdot\DT{\bm m}_{\sigmanu k}
+(\nabla\bm{\chi}_{\sigmanu k})^\top\DT{\bm{\mathsf h}}_{\rm e}{\circ}\bm{\chi}_{\sigmanu k}\cdot{\bm m}_{\sigmanu k}\,\d x
\end{align}
as well as the chain rule 
\begin{align}\label{eq:13}
&\frac{\d}{\d t}
\int_{\varGamma} \bm g\cdot\bm\chi_{\sigmanu k}\,\d S
=\int_{\varGamma}\DT{\bm g}\cdot\bm\chi_{\sigmanu k}
+{\bm g}\cdot\DT{\bm\chi}_{\sigmanu k}\,\d S
\end{align}
\end{subequations}
integrated over $[0,t]$. We also note that \eq{est-of-y-m-z} is \eq{formal-energy} 
for $\alpha=0$ when the by-part 
integration \eq{by-part-for-magnetic} has been applied.
Using \eq{ass-coercivity} and \eq{ass-h-e}, by the H\"older and the Young and 
the Gronwall inequalities, we obtain the estimates (\ref{est}a-d). These
estimates are uniform in the regularization parameters.
We also used \BLUE that \color{black} $
\psi_{\textsc{me}}(\nabla\bm\chi_0,\bm m_0,\zeta_0)$
is bounded due to our assumptions \eq{ass-coercivity} and \eq{ass-IC-det}.
In particular we use that $\partial_{\bm m}\psi_{\textsc{th}}(\bm m,\zeta,\theta)$
and $\partial_{\zeta}\psi_{\textsc{th}}(\bm m,\zeta,\theta)$ can
be \BLUE estimated \color{black} uniformly with respect to $\theta$ so 
that the regularization and discretization of the heat equation
does not influence the estimates (\ref{est}a-d).

Now, the estimates \eq{est-Delta-zeta} and \eq{est-Delta-m} are obtained by 
comparison from the respective equations 
\eq{system.b} and \eq{system.c}. In fact, in these equations
all lower-order terms (which do not contain the Laplacian operator) are 
already estimated in $L^2(Q)$-spaces even on the Galerkin level, thanks to 
\eqref{est-y}, \eqref{est-m}, and also to 
\BLUE \eqref{ass-xi}\color{black}. Recall that we assumed the 
finite-dimensional subspaces to be contained in $H^2({\varOmega})$-spaces
so that the Laplacians have a good sense on the Galerkin level. 
\medskip

\noindent \emph{Step 3 - estimates 
\BLUE (\ref{est}g) and part of  (\ref{est}i)\color{black}}.
We test the regularized heat equation
by $\theta_{\sigmanu k}$. Denoting by
$\hat C_{\rm v}(\bm m,\zeta,\theta)$ \BLUE the \color{black} 
primitive function of 
$\theta\mapsto\theta c_{\rm v}(\bm m,\zeta,\theta)$ such that $\hat C_{\rm v}(\bm m,\zeta,0)=0$, i.e.\ 
\begin{align}
\hat C_{\rm v}(\bm m,\zeta,\theta):=\int_0^1r\theta^2c_{\rm v}(\bm m,\zeta,r\theta)\,\d r,
\label{hat-Cv}\end{align}
we 
\BLUE can write \color{black} 
 $\hat C_{\rm v}(\bm m,\zeta,\theta)\!\DT{\,^{}}
=\theta c_{\rm v}(\bm m,\zeta,\theta)\DT\theta
+\partial_{\bm m}\hat C_{\rm v}(\bm m,\zeta,\theta)\DT{\bm m}
+\partial_{\zeta}\hat C_{\rm v}(\bm m,\zeta,\theta)\DT\zeta$
with $\partial_{\bm m}\hat C_{\rm v}(\bm m,\zeta,\theta)
=\int_0^1r\theta^2\partial_{\bm m}c_{\rm v}(\bm m,\zeta,r\theta)\,\d r$
and $\partial_{\zeta}\hat C_{\rm v}(\bm m,\zeta,\theta)
=\int_0^1r\theta^2\partial_{\zeta}c_{\rm v}(\bm m,\zeta,r\theta)\,\d r$.
Therefore
\begin{align}\nonumber
\theta\DT w&=\theta e_{\textsc{th}}(\bm m,\zeta,\theta)\!\DT{\,^{}}
=\theta c_{\rm v}(\bm m,\zeta,\theta)\DT\theta
+\theta\partial_{\bm m}e_{\textsc{th}}(\bm m,\zeta,\theta)\DT{\bm m}
+\theta\partial_\zeta e_{\textsc{th}}(\bm m,\zeta,\theta)\DT\zeta
\\&=\hat C_{\rm v}(\bm m,\zeta,\theta)\!\DT{\,^{}}
+\big(\theta\partial_{\bm m}e_{\textsc{th}}(\bm m,\zeta,\theta)
-\partial_{\bm m}\hat C_{\rm v}(\bm m,\zeta,\theta)\big)\DT{\bm m}
+\big(\theta\partial_\zeta e_{\textsc{th}}(\bm m,\zeta,\theta)
-\partial_{\zeta}\hat C_{\rm v}(\bm m,\zeta,\theta)\big)\DT\zeta.
\label{theta-DTw}\end{align}
We can now perform the intended test of the regularized heat-transfer
problem 
\eqref{system.d-reg}--\eqref{BC-IC-reg}
in its Galerkin approximation by
$\theta_{\sigmanu k}$. Using \eq{theta-DTw} integrated 
over the time integral $I=[0,t]$, we thus obtain
\begin{align}
\nonumber&\int_{\varOmega} \hat C_{\rm v}(\bm m_{\sigmanu k}(t),\zeta_{\sigmanu k}(t),\theta_{\sigmanu k}(t))\,\d x
+\int_0^t\int_{\varOmega}
\BLUE\Big|
\bm K^{1/2}(\nabla\bm\chi_{\sigmanu k},\bm m_{\sigmanu k},\zeta_{\sigmanu k},\theta_{\sigmanu k})
\Ceenk\nabla\theta_{\sigmanu k}\Big|^2\color{black}
\,\d x\d t
\\\nonumber&= \int_0^t\int_{\varOmega}\bigg(
\frac{\tau_1|\DT{\bm m}_{\sigmanu k}|^2}{1{+}\sigmanu|\DT{\bm m}_{\sigmanu k}|^2}
+\frac{\tau_2\DT\zeta_{\sigmanu k}^2}{1{+}\sigmanu\DT\zeta_{\sigmanu k}^2}
+\frac{\bm M
(\nabla\bm\chi_{\sigmanu k},\bm m_{\sigmanu k},\zeta_{\sigmanu k},\theta_{\sigmanu k})\nabla\mu_{\sigmanu k}{\cdot}\nabla\mu_{\sigmanu k}}{1+\sigmanu|\nabla\mu_{\sigmanu k}|^2}\bigg)\theta_{\sigmanu k}
\\\nonumber&\qquad+
\big(\theta_{\sigmanu k}\partial_{\bm m}e_{\textsc{th}}(\bm m_{\sigmanu k},\zeta_{\sigmanu k},\theta_{\sigmanu k})
-\partial_{\bm m}\hat C_{\rm v}(\bm m_{\sigmanu k},\zeta_{\sigmanu k},\theta_{\sigmanu k})\big)\DT{\bm m}_{\sigmanu k}
\\\nonumber&\qquad+\big(\theta_{\sigmanu k}\partial_\zeta e_{\textsc{th}}(\bm m_{\sigmanu k},\zeta_{\sigmanu k},\theta_{\sigmanu k})
-\partial_{\zeta}\hat C_{\rm v}(\bm m_{\sigmanu k},\zeta_{\sigmanu k},\theta_{\sigmanu k})\big)\DT\zeta_{\sigmanu k}\,\d x\d t
\\&\qquad+
    \int_0^t\int_{\varGamma}
\BLUE
K\frac{\theta_{\rm e}}{1{+}\sigmanu\theta_{\rm e}}
\theta_{\sigmanu k}\color{black}\,
\d S\d t
+\int_{\varOmega} \hat C_{\rm v}(\bm m_0,\zeta_0,\theta_{0,\sigmanu})\,\d x
\label{test-heat}
\end{align}
From the assumption $\inf c_{\rm v}(\bm m,\zeta,\cdot)=c_0>0$, 
cf.\ \eq{ass-cv},
we know that $\hat C_{\rm v}(\bm m,\zeta,\theta)\ge c_0|\theta|^2/2$.
From \eq{ass-cv/mz}, $|\partial_{\bm m}\hat C_{\rm v}|\le C\theta$
and $|\partial_{\zeta}\hat C_{\rm v}|\le C\theta$.
Using also the assumption \eq{ass-psi/mz},
from \eq{test-heat} we can thus estimate
\begin{align}\nonumber
&\frac{c_0}2\|\theta_{\sigmanu k}(t)\|_{L^2({\varOmega})}^2
+\nu\!\int_0^t\!\Big\|\BLUE\Ceenk\color{black}\nabla\theta_{\sigmanu k}\Big\|_{L^2({\varOmega};\R^d)}^2\d t
+K\int_0^t\|\theta_{\sigmanu k}\|_{L^2(\varGamma)}^2\ \d t
\\&\quad\nonumber
\le\int_0^t\!\int_{\varOmega}\Big(\frac{\tau_1{+}\tau_2{+}\max|\bm M
|}\sigmanu\BLUE+C\Big)\color{black}
|\theta_{\sigmanu k}|
+2C\big(|\DT{\bm m}_{\sigmanu k}|{+}|\DT\zeta_{\sigmanu k}|\big)
|\theta_{\sigmanu k}|\,\d x\d t
\\[-.3em]&\qquad\qquad\qquad\qquad\qquad\quad+
\int_0^t\!\int_{\varGamma}
\frac K\eps|\theta_{\rm e}|
\,\d S\d t
+\int_{\varOmega} \hat C_{\rm v}(\bm m_0,\zeta_0,\theta_{0,\sigmanu})\,\d x\,,
\end{align}
\BLUE{}where 
$C$ is from \eq{ass-psi/mz} and where 
we denoted by $\sigma$ the positive-definiteness constant
of $\bm M$, cf.\ the assumption \eq{ass-M-K}. \color{black}
The boundary term $\|\theta_{\sigmanu k}\|_{L^2({\varGamma})}^2$ is to 
be estimated through the trace operator by $2N\|\theta_{\sigmanu k}\|_{L^2({\varOmega})}^2+2N\|\nabla\theta_{\sigmanu k}\|_{L^2({\varOmega};\R^d)}^2$ with $N$ denoting
the norm of the trace operator $H^1({\varOmega})\to L^2({\varGamma})$.
Taking $0<2\epsilon<\nu/(2N)$, the gradient term arising from this 
boundary term can be absorbed in the left-hand side. 
Realizing that we have $\DT{\bm m}_{\sigmanu k}$ and 
$\DT\zeta_{\sigmanu k}$ already estimates in $L^2(Q;\R^d)$ and 
$L^2(Q)$, respectively, we can use the Gronwall inequality to 
get the estimate \eq{est++}. 

\medskip

\allowdisplaybreaks
\noindent\emph{Step 4 - estimates \BLUE (\ref{est}h,j) and part of (\ref{est}j)\color{black}.} 
From \eq{est-zeta+}, 
we can now read the estimate for $\nabla\mu_{\sigmanu k}$. Indeed,
we have the bound
$\nabla\bm\chi_{\sigmanu k}\in L^\infty(Q;\R^{d\times d})
$ so that, realizing that $(\Cof(\nabla\bm\chi_{\sigmanu k}))^{-1}=
\nabla\bm\chi_{\sigmanu k}/\det(\nabla\bm\chi_{\sigmanu k})$,
we have
\begin{align}\nonumber
&\|\nabla\mu_{\sigmanu k}\|_{L^2(Q;\R^{d})}^{}=
\bigg\|\frac{\nabla\bm\chi_{\sigmanu k}(\Cof\nabla\bm\chi_{\sigmanu k})\nabla\mu_{\sigmanu k}}{\det(\nabla\bm\chi_{\sigmanu k})}\bigg\|_{L^2(Q;\R^{d})}^{}
\\&\qquad\le\|\nabla\bm\chi_{\sigmanu k}\|_{L^\infty(Q;\R^{d\times d})}^{}
\bigg\|\frac{1}{\sqrt{\det(\nabla\bm\chi_{\sigmanu k})}}\bigg\|_{L^\infty(Q)}
\bigg\|\frac{\Cof(\nabla\bm\chi_{\sigmanu k})\nabla\mu_{\sigmanu k}}{\sqrt{\det(\nabla\bm\chi_{\sigmanu k})}}\bigg\|_{L^2(Q;\R^d)}^{}.
\label{est-mu+++}\end{align}
we thus have recovered 
\eq{est-mu}. \BLUE Analogously, from \eq{est-nabla-theta} we can read 
the estimate of $\nabla\theta_{\sigmanu k}$ in $L^2(Q;\R^{d})$ contatined in 
\eq{est++}. \color{black} 
%
%
%
Moreover, we can now read \BLUE the first estimate in \color{black} \eq{est-w-grad} of $\nabla w_{\sigmanu k}$, namely
\begin{align}\nonumber
\nabla w_{\sigmanu k}&=
\partial_{\bm m} e_{\textsc{th}}(\bm m_{\sigmanu k},\zeta_{\sigmanu k},\theta_{\sigmanu k})\nabla\bm m_{\sigmanu k}
\\&\label{eq:31}+\partial_{\zeta}e_{\textsc{th}}(\bm m_{\sigmanu k},\zeta_{\sigmanu k},\theta_{\sigmanu k})\nabla\zeta_{\sigmanu k}
+
\partial_{\theta}e_{\textsc{th}}
(\bm m_{\sigmanu k},\zeta_{\sigmanu k},\theta_{\sigmanu k})
\nabla\theta_{\sigmanu k}.
\end{align}
We have already all three gradients on the right-hand side estimated in 
respective $L^2(Q)$-spaces while the coefficients are bounded by
our assumptions; more specifically, $\partial_{\bm m} e_{\textsc{th}}=\partial_{\bm m}\psi_{\textsc{th}}
+\theta\partial_{{\bm m}\theta}^2\psi_{\textsc{th}}$
and $\partial_\zeta e_{\textsc{th}}=\partial_\zeta\psi_{\textsc{th}}
+\theta\partial_{\zeta\theta}^2\psi_{\textsc{th}}$ is bounded 
due to \eq{ass-psi/mz} and \eq{ass-psith},
while $\partial_{\theta}e_{\textsc{th}}=c_{\rm v}$ is bounded due to \eq{ass-cv}.
Thus \BLUE the first estimate in \color{black} \eq{est-w-grad} 
is proved. 


Eventually,  we can also read an estimate of $
\DT w_{\sigmanu k}$ 
in the seminorm $|\cdot|_l$ defined in \eq{seminorm} for any $k\ge l$
is due to the comparison 
\begin{align}\nonumber
&\int_Q\DT w_{\sigmanu k}v\,\d x\d t
=\int_Q\bigg(\frac{\tau_1|\DT{\bm m}_{\sigmanu k}|^2}{1{+}\sigmanu|\DT{\bm m}_{\sigmanu k}|^2}
+\frac{\tau_2\DT\zeta_{\sigmanu k}^2}{1{+}\sigmanu\DT\zeta_{\sigmanu k}^2}
+\frac{\bm M
(\nabla\bm\chi_{\sigmanu k},\bm m_{\sigmanu k},\zeta_{\sigmanu k},\theta_{\sigmanu k})\nabla\mu_{\sigmanu k}{\cdot}\nabla\mu_{\sigmanu k}}{1+\sigmanu|\nabla\mu_{\sigmanu k}|^2}
\nonumber
\\[-.5em]&\nonumber
\hspace{.9em}+\partial_{\bm m}\psi_{\textsc{th}}(\bm m_{\sigmanu k},\zeta_{\sigmanu k},\theta_{\sigmanu k}){\cdot}\DT{\bm m}_{\sigmanu k}+\partial_{\zeta}\psi_{\textsc{th}}(\bm m_{\sigmanu k},\zeta_{\sigmanu k},\theta_{\sigmanu k})\DT\zeta_{\sigmanu k}\bigg)v
-\bm K
(\nabla\bm\chi_{\sigmanu k},\bm m_{\sigmanu k},\zeta_{\sigmanu k},\theta_{\sigmanu k})\nabla\theta_{\sigmanu k}{\cdot}\nabla v\,\d x\d t
\\[-.3em]&
\le\int_Q
\Big(
\frac{\tau_1{+}\tau_2{+}\max|\bm M
|}\sigmanu
+\max|\partial_{\bm m}\psi_{\textsc{th}}|\,|\DT{\bm m}_{\sigmanu k}|
+\max|\partial_{\zeta}\psi_{\textsc{th}}|\,|\DT\zeta_{\sigmanu k}|
\Big)|v|
+\max|\bm K
||\nabla\theta_{\sigmanu k}|
|\nabla v|
\,\d x\d t
\end{align}
provided $v$'s are valued in the finite-dimensional space $V_l$ and 
$k\ge l$.
Therefore, we have shown that 
$
\sup_{\|v\|_{L^2(I;H^1({\varOmega}))}\le1,\ v(t)\in V_l}\int_Q\DT w_{\sigmanu k}v\,\d x\d t
$
is bounded, so \BLUE the second  estimate in \color{black}
\eq{est-w-grad}.
%
\end{proof}

\begin{remark}\label{rem:g}
\upshape
We are now in position to better justify the statements in Remark \ref{rem:f} concerning the possible choices of the surface load $\bm g$. In fact,  in order for the the integration-by-parts formula \eqref{eq:13} to carry over in the limit passage,  we need  uniform control of the trace of $\dot{\bm g}$ on the parabolic boundary $\varSigma$. If took for $\bm g$ the option illustrated in that remark, then we would need to control the trace $\dot{\bm F}=\nabla\dot{\bm\chi}$, which would demand at least a viscous dissipation, which however we have excluded, since the technique we use in our analysis \BLUE is \color{black} 
already quite sophisticated.
\end{remark}

\begin{remark}[Qualification of $\bm g$]
\upshape
In Assumption \eqref{eq:12}, the qualification on $\bm g$ is stronger 
\BLUE than \color{black}
that of $\bm f$. Our reason for making this assumption is now clear, since in the energy balance the load $\bm g$ appears with its derivative, due to a lack of control of the trace of $\dot{\bm\chi}$. For $\bm f$ no integration by parts is necessary, since $\dot{\bm\chi}$ in the bulk is controlled by inertia.
\end{remark}

\begin{proposition}[Convergence of the Galerkin approximation]\label{prop1}
Let the assumptions \eqref{ass} be fulfilled and $\sigmanu>0$
be fixed. Then the Galerkin solution converges for $k\to\infty$ 
(in terms of selected subsequences)
in the weak* topologies indicated in the estimates \eq{est}.
Moreover, \BLUE every such a subsequence 
exhibits strong convergence 
\begin{subequations}\label{strong-conv}\begin{align}\label{strong-conv-m}
&&&\DT{\bm m}_{\sigmanu k}\to\DT{\bm m}_{\sigmanu}&&\text{ in $L^2(Q;\R^d)$},&&&&
\\&&&\DT{\zeta}_{\sigmanu k}\to\DT\zeta_{\sigmanu}&&\text{ in $L^2(Q)$},
\\&&&\nabla{\mu}_{\sigmanu k}\to\nabla{\mu}_{\sigmanu}&&\text{ in $L^2(Q;\R^d)$},
\end{align}\end{subequations}
 and \color{black} every 
\BLUE{}five\color{black}-tuple 
$(\bm\chi_{\sigmanu},\bm m_{\sigmanu},\zeta_{\sigmanu},\mu_{\sigmanu},\theta_{\sigmanu}
)$ 
obtained as such a limit is a weak solution 
to the regularized initial-boundary-value problem 
\BLUE (\ref{system}a-c,e)--(\ref{BC}a-c)--(\ref{IC-1}) with 
(\ref{system.d-reg})--(\ref{BC-IC-reg}). In addition,
$\theta_{\sigmanu}\ge0$ a.e.\ on $Q$\color{black}.
\end{proposition}

\begin{proof}


Fixing $\sigmanu$, 
we now can pass to the limit $k\to\infty$ in terms of a selected
subsequence. In particular, it is important that we made the 
$\eps$-regularization of $\psi$ so that all involved mappings 
are continuous and the limit passage in corresponding Nemytski\u\i\ mappings
is standard. 

More in detail, 
by the Aubin-Lions Theorem, we have compactness of $\nabla\bm\chi$'s in
$L^{1/\epsilon}(I;C(\bar{\varOmega};\R^d))$ for any $0<\epsilon\le1$. This means that 
$\nabla\bm\chi_{\sigmanu k}\to\nabla\bm\chi_{\sigmanu}$ strongly in 
$L^{1/\epsilon}(I;C(\bar{\varOmega};\R^d))$.
Similarly, by Aubin-Lions' theorem, we have compactness of $\bm m$'s in
$L^{1/\epsilon}(I;W^{1,2^*-\epsilon}({\varOmega};\R^d))$.
This means that 
$\bm m_{\sigmanu k}\to\bm m_{\sigmanu}$ strongly in 
$L^{1/\epsilon}(I;L^{2^*-\epsilon}({\varOmega};\R^d))$.
Similarly, also $\zeta_{\sigmanu k}\to\zeta_{\sigmanu}$ strongly in 
$L^{1/\epsilon}(I;L^{2^*-\epsilon}({\varOmega};\R^d))$
and $\theta_{\sigmanu k}\to\theta_{\sigmanu}$ strongly in 
$L^2(I;L^{2^*-\epsilon}({\varOmega};\R^d))$.
The last convergence is a bit tricky because we do not have an
explicit information about time-derivative of  $\theta_{\sigmanu k}$
so we cannot apply the Aubin-Lions theorem directly to it.
But we have such information about $w_{\sigmanu k}$, cf.\ 
\eq{est-w-grad}. So, using a variant allowing for time-derivatives valued in 
locally convex spaces as e.g.\ in \cite[Lemma 7.7]{Roub13NPDE}, we obtain 
$w_{\sigmanu k}\to w_{\sigmanu}$ strongly in 
$L^2(I;L^{2^*-\epsilon}({\varOmega};\R^d))$.
Realizing that $e_{\textsc{th}}(\bm m,\zeta,\cdot)$ 
is increasing with a continuous and bounded inverse since 
$\partial_\theta e_{\textsc{th}}=
c_{\rm v}$ is well controlled by the assumption \eq{ass-cv},
we can write 
$\theta_{\sigmanu k}=[e_{\textsc{th}}(\bm m_{\sigmanu k},\zeta_{\sigmanu k},\cdot)]^{-1}(w_{\sigmanu k})$
and thus we can read also the desired 
strong convergence for temperatures from the continuity of the 
Nemytski\u\i\ mapping  $(\bm m,\zeta,w)\mapsto[e_{\textsc{th}}(\bm m,\zeta,\cdot)]^{-1}(w)$.

As for $\mu_{\sigmanu k}$, let us realize that 
\BLUE both equations in
\eq{system.c} \color{black}
are linear in terms of $\mu$ so that
the weak convergence suffices. Here we benefit from that 
$\Cof(\nabla\bm\chi_{\sigmanu k})/\sqrt{\det
(\nabla\bm\chi_{\sigmanu k})}\to\Cof(\nabla\bm\chi_{\sigmanu})/\sqrt{\det
(\nabla\bm\chi_{\sigmanu})}$
strongly in any $L^p(Q;\R^{d\times d})$ with $p<\infty$
and thus also 
$\bm{M}
(\nabla\bm\chi_{\sigmanu k},\bm m_{\sigmanu k},\zeta_{\sigmanu k},\theta_{\sigmanu k})\nabla\mu_{\sigmanu k}\to \bm{M}
(\nabla\bm\chi_{\sigmanu},\bm m_{\sigmanu},\zeta_{\sigmanu},\theta_{\sigmanu})\nabla\mu_{\sigmanu}$
weakly in $L^2(Q;\R^d)$; in fact, the weak convergence in $L^1(Q;\R^d)$ would 
suffice, too. Altogether, the limit passage in the Galerkin approximation of 
(\ref{system}a-c) is obvious when taking into account that,
due to the $\eps$-regularization, all nonlinearities have a controlled 
growth so the conventional continuity of the related Nemytski\u\i\ mappings
can be used.

To pass to the limit in the heat equation, we need to prove strong convergence
of the dissipative terms occurring on its right-hand side.

We prove the strong $L^2$-convergence \BLUE \eq{strong-conv-m}\color{black}.
We take $\tilde{\bm m}_k\to{\bm m}_{\sigmanu}$ strongly in $H^1(Q;\R^d)$
valued in the finite-dimensional spaces used for the Galerkin approximation
of 
\eq{system.b}. Thus 
$({\bm m}_{\sigmanu k}-\tilde{\bm m}_k)\!\DT{\,^{}}$
is a legal test function. 
We can also assume that $\tilde{\bm m}_k(0)={\bm m}_0$ so that
\begin{align}\nonumber
&\int_Q\kappa_1\nabla{\bm m}_{\sigmanu k}{:}\nabla
\big(\DT{\bm m}_{\sigmanu k}{-}\DT{\tilde{\bm m}}_k\big)\,\d x\d t
=\int_Q\kappa_1\nabla{\tilde{\bm m}}_k{:}
\nabla\big(\DT{\bm m}_{\sigmanu k}{-}\DT{\tilde{\bm m}}_k\big)\,\d x\d t
-\int_{\varOmega}\kappa_1|{\bm m}_{\sigmanu k}(T){-}\tilde{\bm m}_k(T)|^2\,\d x
\\&
\qquad\qquad\le\int_Q\kappa_1\nabla{\tilde{\bm m}}_k{:}
\nabla\big(\DT{\bm m}_{\sigmanu k}{-}\DT{\tilde{\bm m}}_k\big)\,\d x\d t
=-\int_Q\kappa_1\Delta{\tilde{\bm m}}_k{\cdot}\big(\DT{\bm m}_{\sigmanu k}{-}\DT{\tilde{\bm m}}_k\big)\,\d x\d t.
\end{align}
We should take care about that $\nabla\DT{\bm m}_{\sigmanu}$ is not
well defined. However, we can rely on having 
$\Delta{\bm m}_{\sigmanu}\in L^2(Q;\R^d)$, cf.\ \eq{est-Delta-m}, and
to assume also that $\Delta\tilde{\bm m}_k\to\Delta\bm m_{\sigmanu}$
strongly in $L^2(Q;\R^d)$. Thus,
by performing this test, we can estimate
\begin{align}\nonumber
&\limsup_{k\to\infty}\frac{\tau_1}2\big\|\DT{\bm m}_{\sigmanu k}{-}\DT{\bm m}_{\sigmanu}\big\|_{L^2(Q;\R^d)}^2
\le\limsup_{k\to\infty}\int_Q\tau_1\big|\DT{\bm m}_{\sigmanu k}{-}\DT{\tilde{\bm m}}_k\big|^2\,\d x\d t
+\lim_{k\to\infty}\tau_1
\big\|\DT{\bm m}_{\sigmanu}{-}\DT{\tilde{\bm m}}_k\big\|_{L^2(Q;\R^d)}^2
\\&\nonumber\qquad
\le\lim_{k\to\infty}\!\int_Q\!\big((\nabla\bm{\chi}_{\sigmanu k})^\top\bm{\mathsf h}_{\rm e}{\circ}\bm{\chi}_{\sigmanu k}-\partial_{\bm m}\psi
(\nabla\bm \chi_{\sigmanu k},\bm m_{\sigmanu k},\zeta_{\sigmanu k},\theta_{\sigmanu k})\big){\cdot}
\big(\DT{\bm m}_{\sigmanu k}{-}\DT{\tilde{\bm m}}_k\big)
\\[-.3em]&\qquad\qquad\qquad\qquad\qquad\qquad\qquad\quad
-\tau_1\DT{\tilde{\bm m}}_k{\cdot}
\big(\DT{\bm m}_{\sigmanu k}{-}\DT{\tilde{\bm m}}_k\big)
+\kappa_1\Delta{\tilde{\bm m}}_k{\cdot}\big(\DT{\bm m}_{\sigmanu k}{-}\DT{\tilde{\bm m}}_k\big)
\,\d x\d t=0.
\label{dmdt-strong}\end{align}
Assumptions \eqref{ass-xi} 
 guarantee that
 $ |\partial_{\bm m}\psi
(\nabla\bm\chi,\bm m,\zeta,\theta)|\le C$.
Thus, $\partial_{\bm m}\psi
(\nabla\bm \chi_{\sigmanu k},\bm m_{\sigmanu k},\zeta_{\sigmanu k},\theta_{\sigmanu k})$ 
converges strongly in $L^2(Q;\R^d)$ 
even without any need to specify the limit at
this moment.

Further, we will prove the strong $L^2$-convergence 
\BLUE (\ref{strong-conv}b,c)\color{black},
which is also needed for the limit passage
in the right-hand side of the heat equation.
Like we did for 
\eq{system.b}, we now 
take $\tilde\mu_k\to\mu_{\sigmanu}$
strongly in $L^2(Q)$ and $\tilde\zeta_k\to\zeta_{\sigmanu}$
strongly in $H^1(Q)$ both valued in the finite-dimensional space used for the 
Galerkin approximation of \eq{system.c} and \eq{system.d-reg}.
Denoting by $\sigma>0$ the positive-definiteness constant of 
$\bm{\mathsf M}$ and using \BLUE the 
Galerkin approximation of the equtions \eq{system.c}
 tested respectively \color{black}
by $\mu_{\sigmanu k}-\tilde\mu_k$
and 
$(\zeta_{\sigmanu k}-\tilde\zeta_k)\!\DT{\,^{}}$,
 we can estimate 
\allowdisplaybreaks
\begin{align}\nonumber
&\limsup_{k\to\infty}\int_Q\frac{\tau_2}2\big(\DT\zeta_{\sigmanu k}
-\DT\zeta_{\sigmanu}\big)^2
+\frac\sigma2\bigg|\Ceenk\nabla\mu_{\sigmanu k}
-\Ceen\nabla\mu_{\sigmanu}\bigg|^2\,\d x\d t
\\\nonumber
&\le\limsup_{k\to\infty}
\int_Q\tau_2\big(\DT\zeta_{\sigmanu k}{-}\DT{\tilde\zeta}_k\big)^2
+\sigma\bigg|\Ceenk\nabla(\mu_{\sigmanu k}{-}\tilde\mu_k)\bigg|^2\,\d x\d t
\\\nonumber
&\qquad
+\lim_{k\to\infty}\bigg(\tau_2\big\|\DT{\tilde\zeta}_k{-}\DT\zeta_{\sigmanu}\big\|_{L^2(Q)}^2+\sigma\bigg\|\Ceenk\nabla\tilde\mu_k{-}\Ceen\nabla\mu_{\sigmanu}
\bigg\|_{L^2(Q;\R^d)}^2\bigg)
\\\nonumber
&
\le\limsup_{k\to\infty}
\int_Q\tau_2\big(\DT\zeta_{\sigmanu k}{-}\DT{\tilde\zeta}_k\big)^2
+\BLUE
\Big|\bm{\mathsf M}^{1/2}(\bm m_{\sigmanu k},\zeta_{\sigmanu k},\theta_{\sigmanu k})\Ceenk\nabla(\mu_{\sigmanu k}{-}\tilde\mu_k)
\Big|\color{black}
\,\d x\d t
\\\nonumber
&=\limsup_{k\to\infty}\bigg(
\int_Q
\big(\mu_{\sigmanu k}-\partial_\zeta^{}\psi
(\nabla\bm\chi_{\sigmanu k},\bm m_{\sigmanu k},\zeta_{\sigmanu k},\theta_{\sigmanu k})\big)\big(\DT\zeta_{\sigmanu k}{-}\DT{\tilde\zeta}_k\big)
\\\nonumber
&\qquad
-\DT\zeta_{\sigmanu k}(\mu_{\sigmanu k}{-}\tilde\mu_k)
-\kappa_2\nabla\zeta_{\sigmanu k}{\cdot}\nabla\big(\DT\zeta_{\sigmanu k}{-}\DT{\tilde\zeta}_k\big)
\,\d x\d t
+\int_{\varSigma}
\BLUE{}M\color{black}
(\mu_{\rm e}{-}\mu_{\sigmanu k})(\mu_{\sigmanu k}{-}\tilde\mu_k)\,\d S\d t\bigg)
\\\nonumber
&-\lim_{k\to\infty}\int_Q\tau_2\DT\zeta_{\sigmanu}{\cdot}\big(\DT\zeta_{\sigmanu k}-\DT{\tilde\zeta}_k\big)
+\sigmanu\nabla\tilde\mu_k{\cdot}\nabla(\mu_{\sigmanu k}{-}\tilde\mu_k)
\\&\nonumber\qquad
+\bm{\mathsf M}(\bm m_{\sigmanu k},\zeta_{\sigmanu k},\theta_{\sigmanu k})\Ceenk\nabla\tilde\mu_k{\cdot}\Ceenk\nabla(\mu_{\sigmanu k}{-}\tilde\mu_k)\,\d x\d t
\\\nonumber
&=\lim_{k\to\infty}\bigg(
\int_Q-\mu_{\sigmanu k}\DT{\tilde\zeta}_k
-\partial_\zeta^{}\psi
(\nabla\bm\chi_{\sigmanu k},\bm m_{\sigmanu k},\zeta_{\sigmanu k},\theta_{\sigmanu k})\big(\DT\zeta_{\sigmanu k}{-}\DT{\tilde\zeta}_k\big)
\\\nonumber
&\qquad
+\DT\zeta_{\sigmanu k}\tilde\mu_k
-\kappa_2\Delta\zeta_{\sigmanu k}{\cdot}\DT{\tilde\zeta}_k
\,\d x\d t+\int_{\varSigma}
\BLUE{}M\color{black}\mu_{\rm e}(\mu_{\sigmanu k}{-}\tilde\mu_k)
-
\BLUE{}M\color{black}\mu_{\sigmanu k}\tilde\mu_k\,\d S\d t\bigg)
\\\nonumber
&\qquad
+\limsup_{k\to\infty}\bigg(\int_{\varOmega}\frac{\kappa_2}2|\nabla\zeta_0|^2
-\frac{\kappa_2}2|\nabla\zeta_{\sigmanu k}(T)|^2\,\d x
-\int_{\varSigma}
\BLUE{}M\color{black}\mu_{\sigmanu k}^2\,\d S\d t\bigg)
\\\nonumber
&\le
\int_Q-\mu_{\sigmanu}\DT{\zeta}_{\sigmanu}
+\DT\zeta_{\sigmanu}\mu_{\sigmanu}
-\kappa_2\Delta\zeta_{\sigmanu}{\cdot}\DT\zeta_{\sigmanu}
\,\d x\d t+\int_{\varSigma}
\BLUE{}M\color{black}\mu_{\sigmanu}^2\,\d S\d t
\\&\qquad
+\int_{\varOmega}\frac{\kappa_2}2|\nabla\zeta_0|^2
-\frac{\kappa_2}2|\nabla\zeta_{\sigmanu}(T)|^2\,\d x
-\int_{\varSigma}
\BLUE{}M\color{black}\mu_{\sigmanu}^2\,\d S\d t=0.
\label{strong-zeta-mu}\end{align}
Note that performing this estimation simultaneously for  both equations 
\eq{system.c}
in their Galerkin approximation\BLUE, it \color{black}
 was important to benefit from the cancellation of the terms
$\pm\DT\zeta_{\sigmanu k}\mu_{\sigmanu k}$ which otherwise
separately would not converge. The last equality 
in \eq{strong-zeta-mu} have exploited the calculus 
\begin{align}\label{calculus}
\int_Q\DT\zeta_{\sigmanu}\Delta\zeta_{\sigmanu}\,\d x\d t=\int_{\varOmega}
\frac12|\nabla\zeta_0|^2-\frac12|\nabla\zeta_{\sigmanu}(T)|^2\,\d x
\end{align}
which can be proved e.g.\ by mollification in space relying on that both  
$\Delta\zeta_{\sigmanu}$ and $\DT\zeta_{\sigmanu}$ are 
in $L^2(Q;\R^d)$, cf.\  
\cite[Formula (3.69)]{PoRoTo10TCTF} or \cite[Formula (12.133b)]{Roub13NPDE}.
Thus we obtain the desired $L^2$-strong convergence $\DT\zeta_{\sigmanu k}$ 
and of $\Ceenk\nabla\mu_{\sigmanu k}$.



The limit passage in the heat equation is simple 
because we already proved the strong convergence of $\theta_{\sigmanu k}$
and of all dissipative-rate terms in the right-hand side, while 
we benefit from having estimated $\nabla\theta_{\sigmanu k}$ in which 
the Fourier law is linear so we can pass to the limit in it weakly.

Altogether, we thus obtain a weak solution 
$(\bm\chi_{\sigmanu},\bm m_{\sigmanu},\zeta_{\sigmanu},\mu_{\sigmanu},\theta_{\sigmanu}
)$ 
to the regularized initial-boundary-value problem 
\BLUE (\ref{system}a-c,e)--(\ref{BC}a-c)--(\ref{IC-1}) with 
(\ref{system.d-reg})--(\ref{BC-IC-reg}). Moreover, the 
non-negativity of temperature $\theta_{\sigmanu}$ can now be proved 
by testing the heat equation (\ref{system.d-reg})--(\ref{BC-IC-reg})
by the negative part of $\theta_{\sigmanu}$ which is now a legal 
test function (in contrast to the Galerkin approximation.
Here we use the assumptions that $\theta_{\rm e}\ge0$, $\theta_0\ge0$, and 
that $
M>0$ in
the boundary condition \eqref{BC4-reg}.
\color{black}.
\end{proof}

\begin{proposition}[Convergence of the regularization]\label{prop2}
The solution obtained in Proposition~\ref{prop1}
satisfies the apriori estimates {\rm(\ref{est}a-h)} with 
$k$ omitted and also the following \BLUE{}a-priori estimates 
\color{black}: 
\begin{subequations}\label{est+2}
\begin{align}
&\big\|\nabla\theta_{\sigmanu}\big\|_{L^r(Q;\R^d)}^{}\le 
C_r\ \ \text{ for }1\le r<(d+2)/(d+1),
\label{est+++}
\\&\label{LooL1-est}
\big\|\theta_{\sigmanu}\big\|_{
L^\infty(I;L^1({\varOmega}))}^{}\le C,
\\&\label{est-w-grad+}
\big\|\nabla w_{\sigmanu}\|_{L^r(Q;\R^d)}\le C_r\ \ \text{ for }1\le r<(d+2)/(d+1),
\\&\label{est-w+}
\big\|\DT w_{\sigmanu}\big\|_{L^1(I;H^3(\Omega)^*)}
\le C,
\end{align}
\end{subequations}
\BLUE{}with $w_{\sigmanu}=e_{\textsc{th}}(\bm m_{\sigmanu},\zeta_{\sigmanu},\theta_{\sigmanu}).$ \color{black}
For $\sigmanu\to0$, 
$(\bm\chi_{\sigmanu},\bm m_{\sigmanu},\zeta_{\sigmanu},\mu_{\sigmanu},\theta_{\sigmanu}
)$ converges weakly* 
(in terms of subsequences) in the topologies indicated in the estimates
{\rm(\ref{est}a-c,e,f,h-j)} and \eq{est+2}. Moreover, 
\BLUE every  such a subsequence 
exhibits strong convergence like \eq{strong-conv} but
now for $\sigmanu\to0$ and  with $k$ omitted and \color{black}
every limit $(\bm\chi,\bm m,\zeta,\mu,\theta
)$ obtained by this way is a weak solution
to the original problems according Definition~\ref{def}.
\end{proposition}

\begin{proof}
We divide the proof into three steps. 

\medskip

\noindent{\it Step 1 -- limit passage in 
\BLUE chemo-magneto-mechanical
part (\ref{system}a-c)\color{black}}.
We exploit that the constants in (\ref{est}a-h) are independent of $\eps$ and
these estimates are inherited for $(\bm\chi_{\sigmanu},\bm m_{\sigmanu},\zeta_{\sigmanu},\mu_{\sigmanu},\theta_{\sigmanu},w_{\sigmanu})$, too.
In fact, the \BLUE latter \color{black} estimate 
\eq{est-w-grad}
now can be ``translated'' for this limit
as 
\begin{align}\label{est-DT-w}
\|\DT w_{\sigmanu}\|_{L^2(I;H^1(\Omega)^*)}\le C
\end{align}
with $C$ the same constant as in \eq{est-w-grad}; 
 cf.\ \cite[Sect.\,8.4]{Roub13NPDE} for this argumentation. 
This can be then used for the Aubin-Lions theorem in the standard way.


\BLUE
Similarly as in the proof of Lemma~\ref{lem-1}, we can see that
\color{black}
for some $\eta>0$, 
we have 
$\det(\nabla\bm\chi_{\sigmanu})\ge\eta$ everywhere 
on $Q$.
%
We now can pass to the limit $\eps\to0$, still relying that
all nonlinearities have a controlled 
growth so the conventional continuity of the related Nemytsk\u\i\ mappings
can be used because the singularity of $\psi$ 
is effectively not
seen due to that $\det(\nabla\bm\chi_{\sigmanu})\ge\eta$.
Most arguments are identical with those used for the Galerkin approximation 
and we will not repeat them in detail, except that the estimation 
\eqref{dmdt-strong} must be slightly modified 
because, in contrast to the Galerkin approximation where
 $\nabla\DT{\bm m}_{\sigmanu k}$ was legitimate, 
here $\nabla\DT{\bm m}_{\sigmanu}$ is not well defined.
Yet, we can first pass to the limit in the semilinear equation 
\eq{system.b} by the weak convergence, using also that
$\partial_{\bm m}\psi
(\nabla\bm \chi_{\sigmanu},\bm m_{\sigmanu},\zeta_{\sigmanu},\theta_{\sigmanu})\to
\partial_{\bm m}\psi(\nabla\bm \chi,\bm m,\zeta,\theta)$. 
Then, testing the limit equation by 
$\DT{\bm m}_{\sigmanu}{-}\DT{\bm m}$,
we can directly estimate
\begin{align}
\nonumber
\limsup_{\eps\to0}\tau_1\big\|\DT{\bm m}_{\sigmanu}{-}\DT{\bm m}\big\|_{L^2(Q;\R^d)}^2
&\le\lim_{\eps\to0}\!\int_Q\!\big((\nabla\bm{\chi}_{\sigmanu})^\top\bm{\mathsf h}_{\rm e}{\circ}\bm{\chi}_{\sigmanu}-\partial_{\bm m}\psi
(\nabla\bm \chi_{\sigmanu},\bm m_{\sigmanu},\zeta_{\sigmanu},\theta_{\sigmanu})\big){\cdot}
\big(\DT{\bm m}_{\sigmanu}{-}\DT{\bm m}\big)
\\[-.5em]
&\nonumber\qquad\qquad
-\tau_1\DT{{\bm m}}{\cdot}
\big(\DT{\bm m}_{\sigmanu}{-}\DT{\bm m}\big)\,\d x\d t
-\limsup_{\eps\to0}\!\int_Q\!\kappa_1\Delta{\bm m}_{\sigmanu}{\cdot}\big(\DT{\bm m}_{\sigmanu}{-}\DT{\bm m}\big)
\,\d x\d t
\\
&
=\int_Q\! \kappa_1\Delta{\bm m}{\cdot}\DT{\bm m}
\,\d x\d t
+\limsup_{\eps\to0}\!\int_\Omega\!\frac{\kappa_1}2|\nabla{\bm m}_0|^2
-\frac{\kappa_1}2|\nabla{\bm m}_{\sigmanu}(T)|^2\,\d x\d t=0,
\label{dmdt-strong+}
\end{align}
where we have exploited the calculus as in \eq{calculus}
but now for ${\bm m}_{\sigmanu}$, i.e.\ 
\begin{align}\label{calculus+}
\int_Q\DT{\bm m}_{\sigmanu}\Delta\bm m_{\sigmanu}\,\d x\d t=\int_{\varOmega}
\frac12|\nabla\bm m_0|^2-\frac12|\nabla\bm m_{\sigmanu}(T)|^2\,\d x,
\end{align}
relying that both
$\Delta{\bm m}_{\sigmanu}$ and 
$\DT{\bm m}_{\sigmanu}$ are in $L^2(Q;\R^d)$. Combining \eqref{dmdt-strong+} 
with the identity
\begin{align}\label{calculus+++}
\int_Q\DT{\bm m}_{\sigmanu}\Delta\bm m_{\sigmanu}\,\d x\d t=\int_{\varOmega}
\frac12|\nabla\bm m_0|^2-\frac12|\nabla\bm m_{\sigmanu}(T)|^2\,\d x
\end{align}
we obtain the desired strong convergence.


Let us denote such a limit by $(\bm\chi,\bm m,\zeta,\mu,\theta)$. 
We thus obtain a weak solution to \BLUE the chemo-magneto-mechanical
part (\ref{system}a-c). \color{black}
The estimates (\ref{est}a-f,h) are inherited for these solutions, too.





\medskip


\noindent{\it Step 2 -- estimates \eqref{est+2}}. 
The further a-priori estimates concerns the heat equation which is
now continuous and allow for various ``nonlinear'' tests, in contrast to
its Galerkin approximation. 
\BLUE{}We now also use the \color{black} non-negativity of temperature 
\BLUE $\theta_{\sigmanu}$ proved already in Proposition~\ref{prop1}\color{black}.
%
This allows us reading the information from the \BLUE natural 
energy \color{black} test of the heat 
equation by 1, namely \eq{LooL1-est}. 

The second
``nonlinear'' test yields further estimation of $\nabla\theta$ independent 
of $\sigmanu$. More specifically, following \cite{BocGal89NEPE}
in the simplified variant of \cite{FeiMal06NSET},
we perform the test by $\chi_\varsigma(\theta_{\sigmanu})$ 
with an increasing concave function $\chi_\varsigma:[0,+\infty)\to[0,1]$ defined 
$\chi_\varsigma(w):=1-1/(1{+}w)^\varsigma$ for some $\varsigma>0$. 
Analogously as in \eq{hat-Cv}, we now define a primitive
function to $\theta\mapsto\chi(\theta)c_{\rm v}(\bm m,\zeta,\theta)$ as
\begin{align}\nonumber
\hat C_{{\rm v},\varsigma}(\bm m,\zeta,\theta):=\int_0^1\theta\chi_\varsigma(r\theta)c_{\rm v}(\bm m,\zeta,r\theta)\,\d r,
\end{align}
We notice that, thanks to the assumption \eqref{ass-cv/mz},
\begin{equation}\label{eq:32}
  |\partial_{\bm m} C_{\rm v,\varsigma}(\bm m,\zeta,\theta)|=\left|\partial_{\bm m}\int_0^1\theta\chi_\varsigma(r\theta)c_{\rm v}(\bm m,\zeta,\theta)\ \d r\right|\le C \frac \theta{1+\theta^{1+\varsigma}}\le C. 
\end{equation}
Similarly, we have
\begin{equation}\label{eq:33}
  |\partial_{\zeta} C_{\rm v,\varsigma}(\bm m,\zeta,\theta)|\le C.
\end{equation}
Like \eq{test-heat}, employing also that 
$\hat C_{{\rm v},\varsigma}(\bm m,\zeta,\theta)\ge0$, 
$0\le\chi_\varsigma(\theta_{\sigmanu})\le1$, $\chi_\varsigma'(w)=\varsigma/(1{+}w)^{1+\varsigma}$, and that, thanks to 
$K\theta_{\rm e}/(1+\eps \theta_{\rm e})\le K\theta_{\rm e}$ in \eq{BC4-reg},
this gives the estimate
\begin{align}
\nonumber&\varsigma\int_Q\frac{\sigma}{(1{+}\theta_{\sigmanu})^{1+\varsigma}}
\bigg|\frac{\Cof(\nabla\bm\chi_{\sigmanu})\nabla\theta_{\sigmanu}}
{\sqrt{\det(\nabla\bm\chi_{\sigmanu})}}\bigg|^2
\,\d x\d t
\le
\int_Q\bm K
(\nabla\bm\chi_{\sigmanu},\bm m_{\sigmanu},\zeta_{\sigmanu},\theta_{\sigmanu})\nabla\theta_{\sigmanu}\cdot\nabla\chi_\varsigma(\theta_{\sigmanu})\,\d x\d t
\\\nonumber&\quad
\le \int_Q\bigg(
\frac{\tau_1|\DT{\bm m}_{\sigmanu}|^2}{1{+}\sigmanu|\DT{\bm m}_{\sigmanu}|^2}
+\frac{\tau_2\DT\zeta_{\sigmanu}^2}{1{+}\sigmanu\DT\zeta_{\sigmanu}^2}
+\frac{\bm K
(\nabla\bm\chi_{\sigmanu},\bm m_{\sigmanu},\zeta_{\sigmanu},\theta_{\sigmanu})\nabla\mu_{\sigmanu}{\cdot}\nabla\mu_{\sigmanu}}{1+\sigmanu|\nabla\mu_{\sigmanu}|^2}\bigg)
\\\nonumber&\qquad
+
\big(
|\partial_{\bm m}e_{\textsc{th}}(\bm m_{\sigmanu},\zeta_{\sigmanu},\theta_{\sigmanu})|
-\partial_{\bm m}\hat C_{\rm v}(\bm m_{\sigmanu},\zeta_{\sigmanu},\theta_{\sigmanu})\big)\DT{\bm m}_{\sigmanu}
+\big(
|\partial_\zeta e_{\textsc{th}}(\bm m_{\sigmanu},\zeta_{\sigmanu},\theta_{\sigmanu})|
-\partial_{\zeta}\hat C_{\rm v}(\bm m_{\sigmanu},\zeta_{\sigmanu},\theta_{\sigmanu})\big)\DT\zeta_{\sigmanu}\,\d x\d t
\\&\qquad\qquad\qquad\qquad\qquad\qquad
+\int_{\varSigma} K\theta_{\rm e}
\d S\d t
    +\int_{\varOmega}\hat C_{{\rm v},\varsigma}(\bm m_0,\zeta_0,\theta_{0,\sigmanu})\,\d x
\:\le\:C
\label{test-heat+}
\end{align}
\BLUE{}with $\sigma>0$ the positive-definiteness constant of 
$\bm K$. \color{black}
To see that the last inequality holds true, we recall that, as pointed out in 
the paragraph after \eqref{eq:31}, we have $\partial_{\bm m}e_{\textsc{th}}$
\BLUE and \color{black} $\partial_\zeta e_{\textsc{th}}$ bounded. We also 
use \eqref{eq:32} and \eqref{eq:33}.

Next, we notice that
\begin{align}\nonumber
&\int_Q\bigg|\frac{\Cof(\nabla\bm\chi_{\sigmanu})\nabla\theta_{\sigmanu}}
{\sqrt{\det(\nabla\bm\chi_{\sigmanu})}}\bigg|^r\,\d x\d t
=\int_Q(1{+}\theta_\sigmanu)^{(1+\varsigma)r/2}
\frac{1}{(1{+}\theta_\sigmanu)^{(1+\varsigma)r/2}}\bigg|\frac{\Cof(\nabla\bm\chi_{\sigmanu})\nabla\theta_{\sigmanu}}
{\sqrt{\det(\nabla\bm\chi_{\sigmanu})}}\bigg|^r\,\d x\d t
\\&\qquad\qquad\label{est-of-cof-nabla-theta}
\le\bigg(\int_Q
(1{+}\theta_\sigmanu)^{(1+\varsigma)r/(2-r)}\,\d x\d t\bigg)^{1-r/2}
\bigg(\int_Q\frac{1}{(1{+}\theta_\sigmanu)^{1+\varsigma}}
\bigg|\frac{\Cof(\nabla\bm\chi_{\sigmanu})\nabla\theta_{\sigmanu}}
{\sqrt{\det(\nabla\bm\chi_{\sigmanu})}}\bigg|^2\,\d x\d t\bigg)^{r/2}
\end{align}
so that the last factor is bounded due to 
\eq{test-heat+}. Now, using the 
Gagliardo-Nirenberg inequality, we can interpolate $\int_Q
(1{+}\theta_\sigmanu)^{(1+\varsigma)r/(2-r)}\,\d x\d t$ with the already 
obtained estimate \eq{LooL1-est},
namely $\|v\|_{L^{(1+\varsigma)r/(2-r)}({\varOmega})}
\le C_{_{\rm GN}}\|v\|_{L^1({\varOmega})}^{1-\lambda}
\|v\|_{W^{1,r}({\varOmega})}^{\lambda}$ with $\|v\|_{W^{1,r}({\varOmega})}:=
\|v\|_{L^1({\varOmega})}+\|\nabla v\|_{L^r({\varOmega};\R^d)}$, used here for 
$v=1+\theta_\sigmanu(t,\cdot)$ to obtain the estimate:
\begin{align}\nonumber
\int_0^T\!\!\big\|1{+}\theta_\sigmanu(t,\cdot)\big\|^{(1+\varsigma)r/(2-r)}
_{L^{(1+\varsigma)r/(2-r)}({\varOmega})}\,\d t
&\le
\int_0^T\!\!
C_{_{\rm GN}}^{{(1+\varsigma)r}/{2-r}}C_0^{{(1-\lambda)(1+\varsigma)r}/{2-r}}
\Big(C_0{+}
\big\|\nabla\theta_\sigmanu(t,\cdot)
\big\|_{L^r({\varOmega};\R^d)}\Big)^{{\lambda(1+\varsigma)r}/{2-r}}\d t
\\&\nonumber
\ \le
\int_0^T\!\!
C_{_{\rm GN}}^{{(1+\varsigma)r}/{2-r}}C_0^{{(1-\lambda)(1+\varsigma)r}/{2-r}}
\Big(C_0{+}\big\|\nabla\theta_\sigmanu(t,\cdot)\big\|_{L^r({\varOmega};\R^d)}\Big)^{r}
\,\d t
\\&
\label{8-***}
\ \le C_1+C_2
\int_Q
\big|\nabla\theta_\sigmanu
\big|
^r
\,\d x\d t
=C_1+C_2\big\|\nabla\theta_{\sigmanu}\big\|_{L^r(Q;\R^d)}^r
\end{align}
with $C_0={\rm meas}_d({\varOmega})+C$ with $C$ from \eq{LooL1-est},
cf.\ e.g.\ \cite[Formula (12.20)]{Roub13NPDE}.
Here, this estimate is to be combined also with the estimate 
(analogous \eq{est-mu+++} for $\nabla\mu_{\sigmanu}$):
\begin{align}\nonumber
\big\|\nabla\theta_{\sigmanu}\big\|_{L^r(Q;\R^d)}^r
&=\bigg\|
\frac{\nabla\bm\chi_{\sigmanu}}{\sqrt{\det(\nabla\bm\chi_{\sigmanu})}}
\frac{\Cof(\nabla\bm\chi_{\sigmanu})\nabla\theta_{\nu}}{\sqrt{\det(\nabla\bm\chi_{\sigmanu})}}\bigg\|_{L^r(Q;\R^d)}^r
\\&\le\bigg\|
\frac{\nabla\bm\chi_{\sigmanu}}{\sqrt{\det(\nabla\bm\chi_{\sigmanu})}}\bigg\|_{L^\infty(Q;\R^{d\times d})}^r
\bigg\|
\frac{\Cof(\nabla\bm\chi_{\sigmanu})\nabla\theta_{\sigmanu}}{\!\sqrt{\det(\nabla\bm\chi_{\sigmanu})}}\bigg\|_{L^r(Q;\R^{d})}^r
\label{est-heat++}\end{align}
together with that we have 
$\nabla\bm\chi_{\sigmanu}/\sqrt{\det(\nabla\bm\chi_{\sigmanu})}$ already 
apriori bounded.

When raised to power $1-r/2$, \eq{8-***} merged with \eq{est-heat++} 
can be used to estimate the right-hand side of \eq{est-of-cof-nabla-theta}
by the function of the left-hand side of \eq{est-of-cof-nabla-theta}
but in a power less than one, namely $1-r/2$.
Thus we obtain the estimate
\begin{align}
\bigg\|\frac{\Cof(\nabla\bm\chi_{\sigmanu})\nabla\theta_{\sigmanu}}
{\!\sqrt{\det(\nabla\bm\chi_{\sigmanu})}}\bigg\|_{L^r(Q;\R^{d})}\le C_r
\label{est-heat+}
\end{align}
for any $1\le r<(d+2)/(d+1)$.
From it, we can read the estimate for $\nabla\theta_{\sigmanu}$ in $L^r(Q;\R^d)$
by using again \eq{est-heat++}\BLUE, i.e.\ \eq{est+++}\color{black}.
%
Having $\nabla\theta_{\sigmanu}$ estimated, also the estimate \eq{est-w-grad+} of 
$\nabla w_{\sigmanu}$ can be read from the calculus:
$$
\nabla w_{\sigmanu}=
\partial_{\bm m} e_{\textsc{th}}(\bm m_{\sigmanu},\zeta_{\sigmanu},\theta_{\sigmanu})\nabla\bm m_{\sigmanu}
+\partial_{\zeta}e_{\textsc{th}}(\bm m_{\sigmanu},\zeta_{\sigmanu},\theta_{\sigmanu})\nabla\zeta_{\sigmanu}
+\partial_{\theta}e_{\textsc{th}}(\bm m_{\sigmanu},\zeta_{\sigmanu},\theta_{\sigmanu})\nabla\theta_{\sigmanu}.
$$




\medskip

\noindent{\it Step 3 -- limit passage \BLUE in the 
heat equation\color{black}
}.
This proof actually imitates the argumentation 
from \BLUE the proof of Proposition~\ref{prop1}\color{black}.
Another modification consists in 
the regularized dissipation rates which converge strongly in $L^1(Q)$, i.e.
$$
\frac{\tau_1|\DT{\bm m}_\sigmanu|^2}{1{+}\sigmanu|\DT{\bm m}_\sigmanu|^2}
+\frac{\tau_2\DT\zeta_\sigmanu^2}{1{+}\sigmanu\DT\zeta_\sigmanu^2}
+\frac{\bm M(\nabla\bm\chi_\sigmanu,\bm m_\sigmanu,\zeta_\sigmanu,\theta_\sigmanu)\nabla\mu_\sigmanu{\cdot}\nabla\mu_\sigmanu}{1+\sigmanu|\nabla\mu_\sigmanu|^2}
\to\tau_1|\DT{\bm m}_\sigmanu|^2
+\tau_2\DT\zeta^2
+\bm M(\nabla\bm\chi,\bm m,\zeta,\theta)\nabla\mu{\cdot}\nabla\mu.
$$
This can be seen easily when proving the strong $L^2(Q)$-convergence 
of $\DT{\bm m}_\sigmanu\to\DT{\bm m}$, $\DT\zeta_\sigmanu\to\DT\zeta$,
and $\nabla\mu_\sigmanu\to\nabla\mu$ by the techniques we used already
before, see \eqref{dmdt-strong+}, \eqref{strong-zeta-mu}.
\end{proof}

\section*{Acknowledgments}
The authors are thankful to Miroslav \v{S}ilhav\'y for fruitful discussions
about modelling aspects. \BLUE{}Also, many conceptual and other comments
of three anonymous referees have been very useful for improving the 
presentation. \color{black} 


\bibliographystyle{abbrv}

\begin{thebibliography}{10}

\bibitem{Anan11TMCT}
L.~Anand.
\newblock A thermo-mechanically-coupled theory accounting for hydrogen
  diffusion and large elastic-viscoplastic deformations of metals.
\newblock {\em Int. J. Solids Struct.}, 48:962--971, 2011.

\bibitem{Ball02SOPE}
J.~M. Ball.
\newblock Some open problems in elasticity.
\newblock In P.~Newton, P.~Holmes, and A.~Weinstein, editors, {\em Geometry,
  Mechanics, and Dynamics}, pages 3--59. Springer, New York, 2002.

\bibitem{BocGal89NEPE}
L.~Boccardo and T.~Gallou{\"e}t.
\newblock Non-linear elliptic and parabolic equations involving measure data.
\newblock {\em J. Funct. Anal.}, 87:149--169, 1989.

\bibitem{BohliusBrandPleiner2004}
S.~Bohlius, H.~Brand, and H.~Pleiner.
\newblock Macroscopic dynamics of uniaxial magnetic gels.
\newblock {\em Phys. Rev. E}, 70(6):061411, 2004.

\bibitem{BonetCL2012NATMA}
E.~Bonetti, P.~Colli, and P.~Lauren\c{c}ot.
\newblock Global existence for a hydrogen storage model with full energy
  balance.
\newblock {\em Nonlin. Anal.: Theory Met. Appl.}, 75:3558--3573, 2012.

\bibitem{BonetCT2015non}
E.~Bonetti, P.~Colli, and G.~Tomassetti.
\newblock A non-smooth regularization of a forward-backward parabolic equation.
\newblock {\em Math. Mod. Met. Appl. Sci.}, in print. DOI:
  http://dx.doi.org/10.1142/S0218202517500129.

\bibitem{BorceB2001magneto}
L.~Borcea and O.~Bruno.
\newblock On the magneto-elastic properties of elastomer--ferromagnet
 composites.
\newblock {\em J. Mech. Phys. Solids}, 49:2877--2919, 2001.

\bibitem{Brown1963}
W.~F. Brown.
\newblock {\em {Micromagnetics}}.
\newblock Krieger Publ. Co., New York, 1963.

\bibitem{CahnH1958JCP}
J.~W. Cahn and J.~E. Hilliard.
\newblock Free energy of a nonuniform system. {I}. {I}nterfacial free energy.
\newblock {\em J. Chem. Phys.}, 28:258--267, 1958.

\bibitem{Ciarl1988}
P.~G. Ciarlet.
\newblock {\em Mathematical elasticity. {V}ol. {I}}.
\newblock North-Holland, Amsterdam, 1988.

\bibitem{CiaNec87ISCN}
P.~G. Ciarlet and J.~Ne{\v{c}}as.
\newblock Injectivity and self-contact in nonlinear elasticity.
\newblock {\em Arch. Rational Mech. Anal.}, 97:171--188, 1987.

\bibitem{CollinAuernhammerGavatEtAl2003}
D.~Collin, G.~K. Auernhammer, O.~Gavat, P.~Martinoty, and H.~R. Brand.
\newblock Frozen-in magnetic order in uniaxial magnetic gels: Preparation and
  physical properties.
\newblock {\em Macromol. Rapid Comm.}, 24(12):737--741, 2003.

\bibitem{DesPod96CTDF}
A.~De{S}imone and P.~Podio-Guidigli.
\newblock On the continuum theory of deformable ferromagnetic solids.
\newblock {\em Arch. Rational Mech. Anal.}, 136:201--233, 1996.

\bibitem{DuSoFi10TSMF}
F.~P. Duda, A.~C. Souza, and E.~Fried.
\newblock A theory for species migration in a finitely strained solid with
  application to polymer network swelling.
\newblock {\em J. Mech. Phys. Solids}, 58:515--529, 2010.

\bibitem{duda2015stress}
F.~P. Duda and G. Tomassetti.
\newblock Stress effects on the kinetics of hydrogen adsorption in a spherical 
particle: An analytical model.
\newblock {\em Int. J. Hydr. Energy}, 40:17009-17016, 2015.

\BLUE
\bibitem{Edelen1976}
D.~G.~B. Edelen.
\newblock Nonlocal field theories.
\newblock Part II in {\emph Continuum Physics, Vol. IV - 
Polar and Nonlocal Field Theories} (Ed. A.~C. Eringen), 
Acad. Press, New York, 1976.
\color{black}

\bibitem{EllGar96CHED}
C.~M. Elliott and H.~Garcke.
\newblock On the {C}ahn-{H}illiard equation with degenerate mobility.
\newblock {\em SIAM J. Math. Anal.}, 27:404--423, 1996.

\BLUE
\bibitem{Erin02NCFT}
A.~C. Eringen.
\newblock {\em Nonlocal Continuum Field Theories}.
Springer, New York, 2002.
\color{black}

\bibitem{FeiMal06NSET}
E.~Feireisl and J.~M\'alek.
\newblock On the {N}avier-{S}tokes equations with temperature-dependent
  transport coefficients.
\newblock {\em Diff. Equations Nonlin. Mech.}, pages 14pp.(electronic), 
Art.ID~90616, 2006.

\bibitem{FoHrMi03LGPN}
M.~Foss, W.~J. Hrusa, and V.~J. Mizel.
\newblock The {L}avrentiev gap phenomenon in nonlinear elasticity.
\newblock {\em Arch. Rat. Mech. Anal.}, 167:337--365, 2003.

\bibitem{FriGur06TBBC}
E.~Fried and M.~E. Gurtin.
\newblock Tractions, balances, and boundary conditions for nonsimple materials
  with application to liquid flow at small-lenght scales.
\newblock {\em Arch.\ Rational Mech.\ Anal.}, 182:513--554, 2006.

\bibitem{Gurt96GGLC}
M.~E. Gurtin.
\newblock Generalized {G}inzburg-{L}andau and {C}ahn-{H}illiard equations based
  on a microforce balance.
\newblock {\em Physica D}, 92:178--192, 1996.

\bibitem{HeaKro09IWSS}
T.~J. Healey and S.~Kr\"omer.
\newblock Injective weak solutions in second-gradient nonlinear elasticity.
\newblock {\em ESAIM: Control, Optim. \& Cal. Var.}, 15:863--871, 2009.

\bibitem{HuberS1998}
A.~Hubert and R.~Sch\"{a}fer.
\newblock {\em {Magnetic Domains}}.
\newblock Springer, 1998.

\bibitem{James2002CMaT}
R.~James.
\newblock Configurational forces in magnetism with application to the dynamics
  of a small-scale ferromagnetic shape memory cantilever.
\newblock {\em Cont. Mech. Thermodyn.}, 14:55--86, 2002.

\bibitem{JamKind93TMAT}
R.~D. James and D.~Kinderlehrer.
\newblock Theory of magnetostriction with applications to
  {Tb$_x$Dy$_{1-x}$Fe$_2$}.
\newblock {\em Phil. Mag. B}, 68:237--274, 1993.

\bibitem{kolomiets1997rnial}
A.~Kolomiets, L.~Havela, A.~Andreev, V.~Sechovsk\'y, and V.~Yartys.
\newblock {RNiAl} hydrides and their magnetic properties.
\newblock {\em J. Alloys Compounds}, 262:206--210, 1997.

\bibitem{KolomHRBNYHDID2000Magnetic}
A.~Kolomiets, L.~Havela, D.~Rafaja, H.~Bordallo, H.~Nakotte, V.~Yartys,
  B.~Hauback, H.~Drulis, W.~Iwasieczko, and L.~DeLong.
\newblock Magnetic properties and crystal structure of {HoNiAl} and {UNiAl}
  hydrides.
\newblock {\em J. Appl. Phys.}, 87(9; PART 3):6815--6817, 2000.

\bibitem{KolomHSYHo1999Structural}
A.~Kolomiets, L.~Havela, V.~Sechovsk\'y, V.~Yartys, I.~Harris, et~al.
\newblock Structural and magnetic properties of equiatomicrare-earth ternaries.
\newblock {\em Int. J. Hydr. Energy}, 24:119--127, 1999.

\bibitem{KolomHYA1997Hydrogen}
A.~Kolomiets, L.~Havela, V.~Yartys, and A.~Andreev.
\newblock Hydrogen absorption--desorption, crystal structure and magnetism in
  {RENiAl} intermetallic compounds and their hydrides.
\newblock {\em J. Alloys Compounds}, 253:343--346, 1997.

\bibitem{KolomHYA1999Hydrogenation}
A.~Kolomiets, L.~Havela, V.~Yartys, and A.~Andreev.
\newblock Hydrogenation and its effect on crystal structure and magnetism in
  {RENiAl} intermetallic compounds.
\newblock {\em J. Phys. Studies}, 3:55--59, 1999.

\bibitem{KruRou18MMCM}
M.~Kru{\v{z}}{\'\i}k and T.~Roub{\'\i}{\v{c}}ek.
\newblock {\em Mathematical Methods in Contiuum Mechanics of Solids}.
\newblock IMM Series. Springer, Cham/Heidelberg, to appear 2018.

\bibitem{KrStZe15ERIM}
M.~Kru{\v{z}}{\'\i}k, U.~Stefanelli, and J.~Zeman.
\newblock Existence results for incompressible magnetoelasticity.
\newblock {\em Disc. Cont. Dynam. Systems}, 35:2615--2623, 2015.

\BLUE
\bibitem{Kuni82-3EMM}
I.~A. Kunin.
\newblock {\em Elastic Media with Microstructure I: One-Dimensional Models,
II: Three-Dimensional Models.}
\newblock Springer, Berlin, 1982/83.
\color{black}


\bibitem{LewicM2013local}
M.~Lewicka and P.~B. Mucha.
\newblock A local existence result for a system of viscoelasticity with
  physical viscosity.
\newblock {\em Evol. Eq.s Contr. Th.}, 2:337--353, 2013.

\bibitem{MieRou??RIEF}
A.~Mielke and T.~Roub{\'\i}{\v{c}}ek.
\newblock Rate-independent elastoplasticity at finite strains and its numerical
  approximation.
\newblock {\em Math. Mod. Meth. Appl. Sci.}, 26:2203-2236, 2017.

\bibitem{MielkeRoubicek2016viscolargestrains}
A.~Mielke and T.~Roub\'i\v{c}ek.
\newblock 
Thermoviscoelasticity in {K}elvin--{V}oigt rheology at large strains.
\newblock In preparation.

\bibitem{MinEsh68FSTL}
R.~Mindlin and N.~Eshel.
\newblock On first strain-gradient theories in linear elasticity.
\newblock {\em Intl. J. Solid Structures}, 4:109--124, 1968.

\bibitem{miranville1999model}
A.~Miranville.
\newblock A model of {C}ahn--{H}illiard equation based on a microforce balance.
\newblock {\em Compt. Rend. Acad. Sci. - Ser.\,I-Math.}, 328:1247--1252, 1999.

\bibitem{Podi02CISM}
P.~Podio-Guidugli.
\newblock Contact interactions, stress, and material symmetry, for nonsimple
  elastic materials.
\newblock {\em Theor. Appl. Mech.}, 28-29:261--276, 2002.

\bibitem{PoRoTo10TCTF}
P.~{Podio--Guidugli}, T.~Roub\'\i\v{c}ek, and G.~Tomassetti.
\newblock A thermodynamically-consistent theory of the ferro/para\-magnetic
  transition.
\newblock {\em Archive Rat. Mech. Anal.}, 198:1057--1094, 2010.

\bibitem{RSSH01HPSM}
P.~Raj, K.~Shashikala, A.~Sathyamoorthy, N.~H. Kumar, C.~R.~V. Rao, and S.~K.
  Malik.
\newblock Hydride phases, structure, and magnetic properties of the
  {UNiAlH}$_y$ system.
\newblock {\em Phys. Rev. B}, 63:094414, 2001.

\BLUE
\bibitem{Roge88NVPE}
Rogers, R. C.
\newblock Nonlocal variational problems in electromagneto-elastostatics. 
\newblock {\em SIAM J. Math. Anal.}, 19:1329--1347, 1988.

\bibitem{Roge93ERLD}
Rogers, R. C.
\newblock Existence results for large deformations of magnetostrictive 
materials.
\newblock {\em J. Intelligent Mater. Syst. Structures}, 4:477--483, 1993.

\color{black}



\bibitem{Ross05TCGV}
R.~Rossi.
\newblock On two classes of generalized viscous {C}ahn-{H}illiard equations.
\newblock {\em Comm. Pure Appl. Anal.}, 4:405--430, 2005.

\bibitem{Roub13NPDE}
T.~Roub{\'\i}{\v{c}}ek.
\newblock {\em Nonlinear Partial Differential Equations with Applications}.
\newblock Birkh\"auser, Basel, 2nd edition, 2013.

\bibitem{RoubT2013ARMA}
T.~Roub{\'\i}{\v{c}}ek and G.~Tomassetti.
\newblock Phase transformations in electrically conductive ferromagnetic
  shape-memory alloys, their thermodynamics and analysis.
\newblock {\em Arch. Ration. Mech. An.}, 210:1--43, 2013.

\bibitem{RouTom14THSM}
T.~Roub{\'\i}{\v{c}}ek and G.~Tomassetti.
\newblock Thermomechanics of hydrogen storage in metallic hydrides: Modeling
  and analysis.
\newblock {\em Disc. Cont. Dynam. Syst. B}, 19:2313--2333, 2014.

\bibitem{RoubT2010ZAMM}
T.~Roub\'i\v{c}ek and G.~Tomassetti.
\newblock Thermodynamics of shape-memory alloys under electric current.
\newblock {\em Zeit. Angew. Math. Mech.},
  61(61):1--20, 2010.

\bibitem{RybLus05EEMM}
P.~Rybka and M.~Luskin.
\newblock Existence of energy minimizers for magnetostrictive materials.
\newblock {\em SIAM J. Math. Anal.}, 36:2004--2019, 2005.

\bibitem{TKFK07SHMR}
Z.~Tarnawski, L.~Kolwicz-Chodak, H.~Figiel, N.-T.~H. Kim-Ngan, L.~Havela,
  K.~Miliyanchuk, V.~Sechovsk\'y, E.~Santav\'a, and J.~\v{S}ebek.
\newblock Specific heat and magnetization of {RMn}$_2${(H,D)}$_2$.
\newblock {\em J. Alloys Compounds}, 442:372--374, 2007.

\bibitem{Silh88PTNB}
M.~{\v{S}}ilhav\'{y}.
\newblock Phase transitions in non-simple bodies.
\newblock {\em Arch. Rat. Mech. Anal.}, 88:135--161, 1985.

\bibitem{Silha1997Mechanics}
M.~\v{S}ilhav\'{y}.
\newblock {\em The {M}echanics and {T}hermodynamics of {C}ontinuous {M}edia}.
\newblock Springer, Berlin, 1997.



\bibitem{Tomas2015Some}
G.~Tomassetti.
\newblock 
\BLUE{}Smooth and non-smooth \color{black}
 regularization of the nonlinear diffusion
  equation.
\newblock 
{\em Disc. Cont. Dynam. Syst. - Ser. S.} \BLUE 
10:1519--1537, 2017.\color{black}

\bibitem{Toup62EMCS}
R.~Toupin.
\newblock Elastic materials with couple stresses.
\newblock {\em Arch. Rat. Mech. Anal.}, 11:385--414, 1962.

\bibitem{TriAif86GALD}
N.~Triantafyllidis and E.~Aifantis.
\newblock A gradient approach to localization of deformation. {I}.
  {H}yperelastic materials.
\newblock {\em J. Elast.}, 16:225--237, 1986.

\bibitem{ZrS2001Muscular}
M.~Zr{\'\i}nyi and D.~Szab{\'o}.
\newblock Muscular contraction mimiced by magnetic gels.
\newblock {\em Int. J. Mod. Phys. B}, 15(06n07):557--563, 2001.

\end{thebibliography}

\end{sloppypar}
\end{document}